\documentclass[11pt]{article}

\newcommand {\ignore} [1] {}

\usepackage[margin=1in]{geometry}
\usepackage{amsmath,amssymb,amsthm,mathtools}
\usepackage{microtype}
\usepackage{enumitem}
\usepackage{booktabs}
\usepackage{array}
\usepackage{cite}
\usepackage[colorlinks=true,linkcolor=blue,citecolor=blue,urlcolor=blue]{hyperref}
\usepackage{amsthm}

\usepackage{graphicx}

\theoremstyle{plain}
\newtheorem{theorem}{Theorem}[section]
\newtheorem{lemma}{Lemma}[section]
\newtheorem{proposition}{Proposition}[section]
\newtheorem{corollary}{Corollary}[section]
\newtheorem{fact}{Fact}[section]

\theoremstyle{definition}

\theoremstyle{remark}
\newtheorem{remark}{Remark}[section]
\DeclareMathOperator{\Var}{Var}
\DeclareMathOperator{\Med}{median}

\newcommand{\Pa}{P_\alpha}
\newcommand{\Phat}{\widehat P_\alpha}
\newcommand{\Hhat}{\widehat H_\alpha}
\newcommand{\Kalpha}{K_\alpha}
\newcommand{\eps}{\varepsilon}

\newcommand{\Unif}{\operatorname{Unif}}

\newcommand{\bbB}{\mathbb{B}}

\newcommand{\Bin}{\operatorname{Bin}}
\newcommand{\Tcal}{\mathcal{T}}

\newcommand{\TV}{\operatorname{TV}}
\newcommand{\Poi}{\operatorname{Poi}}
\newcommand{\E}{\mathbb{E}}
\newcommand{\Prob}{\mathbb{P}}

\newcommand{\one}{\mathbf{1}}

\title{
Tight Sample Bounds for R\'{e}nyi and Min-Entropy Estimation}
\author{
Arman Adibi \qquad Piotr Krysta\\[0.5em]
Department of Computer Science\\
School of Computer and Cyber Sciences\\
Augusta University\\
Augusta, Georgia, USA
}
\date{}

\begin{document}
\maketitle

\thispagestyle{empty} 

\begin{abstract}
Estimating entropy from samples is a fundamental problem in information
theory and distribution property testing. The required number of samples
depends strongly on the notion of entropy. Shannon entropy measures
average uncertainty and can be estimated to constant additive accuracy
over a $k$-symbol alphabet using $\Theta(k/\log k)$ samples.
Min-entropy, in contrast, depends only on the probability of the most
likely symbol and measures worst-case predictability. Both Shannon entropy and min-entropy are special cases of order-$\alpha$ R\'{e}nyi entropy, denoted $H_{\alpha}$, depending on the order $\alpha$.
We provide complete characterizations of sample complexity of estimating min-entropy and R\'{e}nyi entropy for all regimes of $k$ and integer orders $\alpha > 1$. Furthermore, our lower bound also hold for non-integers $\alpha \geq 1.001$.

We prove that estimating min-entropy to constant additive accuracy has
sample complexity $\Theta(k\log k)$. The upper bound is achieved by the
largest empirical frequency and follows from a support-sensitive
concentration argument based on a dyadic grouping of the symbol
probabilities. The matching lower bound uses a mixture of distributions
in which a slightly heavier symbol is hidden at a uniformly random
location. Thus, min-entropy estimation requires a factor of
$\Theta(\log^2 k)$ more samples than Shannon entropy estimation. This
result also resolves a previously stated $\Theta(k/\log k)$
characterization of classical min-entropy estimation. We next study R\'enyi entropy when its order grows with the alphabet
size. For every integer $\alpha$ where $2\le\alpha\le c_0\log k$, we prove the matching
fixed-accuracy characterization
$\Theta_{c_0}\!\left(\alpha k^{1-1/\alpha}\right)$. Previous results
established $\Omega_\alpha(k^{1-1/\alpha})$ for each fixed integer
$\alpha>1$, and $O_{c_0}(\alpha^2 k^{1-1/\alpha})$, for every integer $\alpha > 1$. Our upper bound analyzes an unbiased
falling-factorial estimator based on $\alpha$-way collisions, while a
hidden-heavy-coordinate construction gives the matching lower bound and
shows that the factor $\alpha$ is unavoidable when the integer order
grows. More generally, for every real order
$1.001\le\alpha\le c_0\log k$, we establish the uniform lower bound
$\Omega_{c_0}\!\left(\alpha k^{1-1/\alpha}\right)$. Finally, the deterministic comparison
$0\le H_\alpha(p)-H_\infty(p)\le \log k/(\alpha-1)$ shows that, once
$\alpha$ is a sufficiently large constant multiple of $\log k$, min-entropy is a
uniformly accurate proxy for $H_\alpha$. Combining this reduction with
the min-entropy upper and lower bounds gives
$\Theta_\varepsilon(k\log k)$ sample complexity throughout the
sufficiently high-order regime, when $\alpha > c_0\log k$.
\end{abstract}

\newpage
\setcounter{page}{1}
\setcounter{tocdepth}{1}
\tableofcontents

\section{Introduction}
\label{sec:introduction}

\setcounter{page}{1}

Entropy estimation is a fundamental problem in statistics, information
theory, cryptography, learning, and property testing
\cite{bravo2022private,skorski2023towards,ObremskiSkorski2017,
AcharyaOrlitskySureshTyagi2017,canonne2024entropy,AcharyaOrlitskySureshTyagi2015}.
Given independent samples from an unknown distribution $p$ on a finite
alphabet $[k]=\{1,\ldots,k\}$, the goal is to estimate a numerical
measure of its uncertainty or randomness by using the smallest number of samples from $p$. Shannon entropy
\[
   H(p)=-\sum_{i\in[k]}p_i\log p_i
\]
measures average uncertainty and governs the fundamental limits of
compression and coding \cite{Shannon1948}. For $\alpha>0$ and
$\alpha\neq1$, the R\'enyi entropy is
\[
   H_\alpha(p)
   =
   \frac{1}{1-\alpha}
   \log\sum_{i\in[k]}p_i^\alpha.
\]
With the standard limiting definitions, this family includes Hartley
entropy at $\alpha=0$, Shannon entropy at $\alpha=1$, collision entropy
at $\alpha=2$, and min-entropy as $\alpha\to\infty$
\cite{Hartley1928,Renyi1961}. The latter is
\[
   H_\infty(p)=-\log\|p\|_\infty
\]
and is determined entirely by the probability of the most likely
symbol. It therefore measures worst-case predictability rather than
average uncertainty.

Min-entropy is central in randomness extraction and cryptographic source
analysis, where a source with a heavy atom may be insecure even when its
average uncertainty is large
\cite{Vadhan2012,BennettBrassardCrepeauMaurer1995,
ImpagliazzoZuckerman1989}. It also plays a central role in the NIST
framework for assessing IID and non-IID entropy sources used in random
bit generation \cite{NISTSP80090B}. More generally, R\'enyi entropies
arise in guessing, cutoff rates, collision search, and key-derivation
analyses \cite{Arikan1996,Csiszar1995,vanOorschotWiener1999}.

The statistical difficulty of entropy estimation varies considerably
across this family. Shannon entropy over $k$ symbols can be estimated to
constant additive accuracy using $\Theta(k/\log k)$ samples. This line
of work includes early consistency and bias analyses by Paninski
\cite{Paninski2003,Paninski2004}, the sublinear-sample estimation and
lower-bound framework of Valiant and Valiant
\cite{ValiantValiant2011,ValiantValiantLinear2011,
ValiantValiantJACM2017}, and the minimax-optimal
polynomial-approximation estimators of Jiao, Venkat, Han, and Weissman \cite{JiaoVenkatHanWeissman2015}, and of Wu and Yang \cite{WuYang2016}. These
results demonstrate that the empirical plug-in estimator can be highly
suboptimal in large alphabets.

For R\'enyi entropy, Acharya, Orlitsky, Suresh, and Tyagi established
the foundational fixed-order theory
\cite{AcharyaOrlitskySureshTyagi2015,
AcharyaOrlitskySureshTyagi2017}. For every fixed $\alpha$, they
characterize the large-alphabet dependence of the sample complexity:
roughly $k^{1/\alpha}$ samples are required for $0<\alpha<1$; nearly
linear sample size is required for fixed noninteger $\alpha>1$; and
$\Theta(k^{1-1/\alpha})$ samples are necessary and sufficient for fixed
integer $\alpha>1$. Obremski and Skorski subsequently sharpened several
distribution-dependent and worst-case bounds, including lower bounds for
real $\alpha>1$ \cite{ObremskiSkorski2017}.

These results primarily treat $\alpha$ as a fixed constant. Thus, the notation
$\Theta_\alpha(k^{1-1/\alpha})$ permits the hidden constants to depend
arbitrarily on $\alpha$ and does not characterize the complexity when
the order grows with the alphabet size. Determining this dependence is
also necessary for understanding the transition from finite-order
R\'enyi entropy to min-entropy.

More recently, Skorski developed an exact U-statistic analysis of the
same collision estimator and used it to obtain adaptive,
distribution-dependent guarantees
\cite{skorski2023towards}. His results are particularly useful when the
distribution has low or moderate entropy and its effective support is
substantially smaller than the ambient alphabet. When specialized to
the worst case over distributions on $[k]$, and to constant entropy
accuracy and confidence, the resulting analysis gives an upper bound of
order
$
   O\!\left(
      \alpha^2 k^{1-1/\alpha}
   \right)
$
for integer $\alpha>1$. Thus, the best explicit growing-order dependence
available from the previous collision-estimator analyses is quadratic
in $\alpha$. Our upper bound improves this dependence to
$O(\alpha k^{1-1/\alpha})$. The improvement does not come from a different estimator. We use the
same falling-factorial collision statistic, but analyze the event needed
for additive accuracy of $H_\alpha$ directly. Previous analyses control
a fixed relative error for the power sum
$P_\alpha(p)=\sum_i p_i^\alpha$. In contrast, additive error
$\varepsilon$ in
$H_\alpha=(1-\alpha)^{-1}\log P_\alpha$ permits the wider
multiplicative window
$
   {P_\alpha}e^{-(\alpha-1)\varepsilon}
   \le
   {\widehat P_\alpha}
   \le
   {P_\alpha}e^{(\alpha-1)\varepsilon},
$ where $\widehat P_\alpha$ is the estimator of $P_\alpha$. Exploiting this order-dependent window, rather than imposing an
order-independent relative-error requirement on $P_\alpha$, is what
removes one factor of $\alpha$.

The endpoint $\alpha=\infty$ presents a separate difficulty. Li and Wu \cite{LiWu2019},
in their work on the quantum query complexity of entropy estimation,
state that the classical complexity of min-entropy estimation is
$\Theta(k/\log k)$, based on an application of Valiant and Valiant's
``Estimating the Unseen'' theorem
\cite{ValiantValiant2011}. We show that this framework does not
apply to min-entropy. The Valiant--Valiant canonical estimator applies to
symmetric properties satisfying a continuity condition with respect to
relative earthmover distance between probability histograms
\cite{ValiantValiant2011,ValiantValiantJACM2017}. Shannon entropy
satisfies this condition, but min-entropy does not: distributions can
become arbitrarily close in relative earthmover distance while their
largest atoms, and hence their min-entropies, remain separated by a
constant. The explicit counterexample is given in
Section~\ref{sec:prelim-rem}.

\begin{table}[t]
\centering
\small
\setlength{\tabcolsep}{5pt}
\renewcommand{\arraystretch}{1.18}
\begin{tabular}{
    p{0.30\linewidth}
    p{0.28\linewidth}
    p{0.30\linewidth}}
\toprule
Regime
&
Previously known
&
\textbf{This paper}
\\
\midrule

Shannon entropy, $\alpha=1$
&
$\Theta(k/\log k)$
\cite{ValiantValiant2011,ValiantValiantJACM2017,
JiaoVenkatHanWeissman2015,WuYang2016}
&
---
\\

\addlinespace

Fixed integer $\alpha>1$
&
$\Theta_\alpha\!\left(k^{1-1/\alpha}\right)^{{\scalebox{1}{\dag}}}$
\cite{AcharyaOrlitskySureshTyagi2015,
AcharyaOrlitskySureshTyagi2017}
&
$\boldsymbol{
O\!\left(
\varepsilon^{-2}\alpha k^{1-1/\alpha}
\log(1/\delta)
\right)}$
\\

\addlinespace

Integer $2\le\alpha\le c_0\log k$
&
$O_{c_0}\!\left(
\alpha^2 k^{1-1/\alpha}
\right)^{{\scalebox{1}{\dag}}}$~\cite{skorski2023towards,AcharyaOrlitskySureshTyagi2017}
&
$\boldsymbol{
\Theta_{c_0}\!\left(
\alpha k^{1-1/\alpha}
\right)}$
\\

\addlinespace

Real $1.001\le\alpha\le c_0\log k$
&
$\Omega(k^{1-1/\alpha})^{{\scalebox{1}{\dag\dag}}}$
\cite{
ObremskiSkorski2017}
&
$\boldsymbol{
\Omega_{c_0}\!\left(
\alpha k^{1-1/\alpha}
\right)}$
\\

\addlinespace

Real $\alpha > c_0\log k$
&
---
&
$\boldsymbol{
O\!\left(
k\varepsilon^{-2}\log(k/\delta)
\right)}$
\\

\addlinespace

Min-entropy, $\alpha=\infty$
&
$\Theta(k/\log k)$ claimed
\cite{LiWu2019}
&
$\boldsymbol{\Theta(k\log k)}$
\\

\bottomrule
\end{tabular}

\caption{
Comparison of sample-complexity bounds for estimating entropy over distributions on a $k$-symbol alphabet. Unless the accuracy and confidence parameters are displayed explicitly, all rates refer to constant additive accuracy and constant success probability. The notation $\Theta_\alpha(\cdot)$ allows the hidden constants to depend on the fixed order $\alpha$, whereas $\Theta_{c_0}(\cdot)$ denotes a bound that holds uniformly over $2\le\alpha\le c_0\log k$, with constants depending only on $c_0$ and the fixed accuracy; see Remark~\ref{rem:comparison-fixed-alpha} for further discussion. For real noninteger orders, this paper establishes only the displayed lower bound. The high-order upper bound follows from $0\le H_\alpha(p)-H_\infty(p)\le\log k/(\alpha-1)$. The previously stated min-entropy rate is based on an application of the Valiant--Valiant framework that does not apply to min-entropy, as explained in Section~\ref{sec:prelim-rem}.\\
$^{{\scalebox{1}{\dag}}}$We extract dependence on $\alpha$, as $\alpha$ can grow as a function of $k$, in the lower bound $\Omega(C_1(\alpha) \cdot k^{1-1/\alpha})$ and upper bound $O(C_2(\alpha) \cdot k^{1-1/\alpha})$, proofs in~\cite{AcharyaOrlitskySureshTyagi2017}, see Appendix \ref{app:Acharya_LB_UB}. Namely, $C_1(\alpha) = \Theta(1)$ and $C_2(\alpha) = \Theta(\alpha^2)$, thus our bound $\Theta(\alpha k^{1-1/\alpha})$ improves on both bounds by a factor of $\alpha$.\\
$^{{\scalebox{1}{\dag\dag}}}$The lower bound of Obremski-Skorski for real $\alpha>1$ has a multiplicative constant that depends on the target accuracy but is independent of $\alpha$ \cite{ObremskiSkorski2017}, see Appendix \ref{app:Obremski_Skorski_LB}. Thus our lower bound $\Omega_{c_0}\!\left(
\alpha k^{1-1/\alpha}
\right)$ and upper bound $O_{c_0}\!\left(
\alpha k^{1-1/\alpha}
\right)$ improve by a factor of $\alpha$ on the results in both papers \cite{skorski2023towards,AcharyaOrlitskySureshTyagi2017}.
}
\label{tab:comparison}
\end{table}

\paragraph{Our contributions.}
We determine the sample complexity of min-entropy and growing
integer-order R\'enyi entropy in the i.i.d.\ sampling model.

First, we prove that estimating min-entropy to constant additive
accuracy has minimax sample complexity $\Theta(k\log k)$. The upper bound
is achieved by the empirical maximum and follows from a
distribution-sensitive concentration argument based on a dyadic
decomposition of the symbol probabilities. The lower bound hides a
slightly heavier coordinate at a uniformly random location and shows
that the resulting mixture cannot be distinguished from the uniform
distribution with $o(k\log k)$ samples.

Second, for every integer  $2\le\alpha\le c_0\log k$, we prove the
matching fixed-accuracy rate
$\Theta_{c_0}(\alpha k^{1-1/\alpha})$. The upper bound analyzes an
unbiased falling-factorial estimator of the power sum and keeps the
dependence on the growing order explicit. The lower bound shows that the
factor $\alpha$, which is absorbed into constants in fixed-order
analyses, is statistically unavoidable.

Finally, when $\alpha$ is sufficiently large relative to
$\log k/\varepsilon$, the deterministic proximity of $H_\alpha$ to
$H_\infty$ transfers the min-entropy estimator to high-order R\'enyi
entropy. The precise statements and proof ideas are developed in the
next section.

\section{Main Results and Proofs Overview}
\label{sec:main-results}

Recall from Section~\ref{subsec:notation} that
$S_\alpha(k,\varepsilon,\delta)$ denotes the minimax number of
independent samples required to estimate $H_\alpha(p)$ uniformly over
$p\in\Delta_k$ to additive error $\varepsilon$ with failure probability
at most $\delta$. The case $\alpha=\infty$ corresponds to min-entropy.

Unless stated otherwise, the asymptotic rates in this section are for
constant additive accuracy and constant success probability. Subscripts
on $O$, $\Omega$, and $\Theta$ indicate the parameters on which the
hidden constants may depend. For other notation used here, please refer to Section~\ref{subsec:notation}.

\subsection{Min-entropy}
\label{subsec:main-min-entropy}

Our first theorem gives a complete worst-case characterization of
min-entropy estimation.

\begin{theorem}[Min-entropy sample complexity]
\label{thm:main-min-entropy-summary}
There exist universal constants $c,C,\varepsilon_0>0$ such that, for all
sufficiently large $k$,
$
   S_\infty(k,\varepsilon,\delta)
   \le
   C\,\frac{k}{\varepsilon^2}\log\frac{k}{\delta}
$
for every $\varepsilon\in(0,\epsilon_0]$ and $\delta\in(0,1/2)$, whereas
$
   S_\infty(k,\varepsilon,1/3)
   \ge
   c\,k\log k.
$
Consequently,
$
   S_\infty(k,\Theta(1),1/3)=\Theta(k\log k).
$
\end{theorem}

The upper bound uses the plug-in estimator
$\widehat H_\infty=-\log\widehat p_\star$, where
$\widehat p_\star=\max_{i\in[k]}N_i/n$. The difficulty is not estimating
the frequency of a fixed maximizer; it is controlling the possibility
that one of many lighter symbols fluctuates upward and becomes the
empirical maximum.

We prove a distribution-dependent guarantee in terms of
$p_\star=\|p\|_\infty$. If $p_\star\ge q$, then
Corollary~\ref{cor:minentropy-upper-q} gives
$n=O(q^{-1}\varepsilon^{-2}\log(1/(q\delta)))$. Since
$p_\star\ge1/k$ for every distribution on $[k]$, the choice $q=1/k$
gives the stated worst-case upper bound.

The main idea is to group symbols according to their probability scale.
At level $r$, there are at most $O(2^r/p_\star)$ symbols with
probability of order $p_\star/2^r$, while the probability that any one
of them reaches the empirical level of the true maximizer decreases
rapidly with $r$. Summing the binomial tails across these dyadic levels
replaces a crude factor $k$ by the intrinsic scale $1/p_\star$.

For the lower bound, Theorem~\ref{thm:minentropy-lower-bound} compares
the uniform distribution with a mixture of alternatives. Under an
alternative indexed by $j$, coordinate $j$ has mass
$(1+\gamma)/k$, for a fixed $\gamma>0$, and the remaining masses are
adjusted to preserve normalization. The alternatives have a constant
min-entropy gap from the null, but the exceptional coordinate is hidden
among $k$ possible locations. A likelihood-ratio second-moment argument
shows that the mixture remains close to the uniform distribution when
$n=o(k\log k)$. Thus, reliable estimation is impossible below the
$k\log k$ scale.

This theorem also resolves the previously stated
$\Theta(k/\log k)$ characterization in \cite{LiWu2019}. The discussion related to the cited Valiant--Valiant approach is treated separately in
Section~\ref{sec:prelim-rem}.

\subsection{Growing integer-order R\'enyi entropy}
\label{subsec:main-growing-renyi}

We next determine the sample complexity when the integer order
$\alpha$ grows with $k$. Previous work gives
$\Theta_\alpha(k^{1-1/\alpha})$ for each fixed integer $\alpha>1$
\cite{AcharyaOrlitskySureshTyagi2015,
AcharyaOrlitskySureshTyagi2017}, but the hidden dependence on $\alpha$
is unspecified. The previously known upper bound $O_{c_0}\!\left(
\alpha^2 k^{1-1/\alpha}
\right)$~\cite{skorski2023towards} explicitly proves the dependence on $\alpha$.

\begin{theorem}[Growing-order R\'enyi entropy]
\label{thm:main-growing-renyi-summary}
Fix $c_0>0$. For every integer $\alpha\ge2$,
$\varepsilon\in(0,1)$, and $\delta\in(0,1/2)$,
$
   S_\alpha(k,\varepsilon,\delta)
   \le
   C\,
   \varepsilon^{-2}
   \alpha k^{1-1/\alpha}
   \log\frac{1}{\delta},
$
where $C>0$ is universal. Moreover, there exist constants
$c=c(c_0)>0$ such that, for
every real order $1.001\le\alpha\le c_0\log k$, we have,
$
   S_\alpha(k,\varepsilon,1/3)
   \ge
   c\,\alpha k^{1-1/\alpha}.
$ Consequently, uniformly over integer orders
$2\le\alpha\le c_0\log k$,
$
   S_\alpha(k,\Theta(1),1/3)
   =
   \Theta_{c_0}\!\left(
      \alpha k^{1-1/\alpha}
   \right).
$
\end{theorem}

The theorem shows that the dependence hidden by fixed-order notation is
linear in $\alpha$ throughout the growing-order regime.

For the upper bound, Section~\ref{sec:Renyi_Entropy_UB} analyzes the
falling-factorial estimator
\(
   \widehat P_\alpha
   =
   \frac{1}{(n)_\alpha}
   \sum_{i=1}^k (N_i)_\alpha
\)
of the power sum $P_\alpha(p)=\sum_{i=1}^k p_i^\alpha$. For integer
$\alpha$, this estimator is unbiased and can be interpreted as the
normalized number of ordered $\alpha$-tuples of observations that all
take the same value.

The dependence on $\alpha$ enters through the variance. Two collision
indicators interact only when their underlying $\alpha$-tuples overlap.
We classify pairs of tuples by their overlap size and bound each
contribution using inequalities between power sums. This yields a
constant-success estimator with
$n=O(\varepsilon^{-2}\alpha k^{1-1/\alpha})$ samples. Repeating the
estimator independently and taking a median reduces the failure
probability to $\delta$ at the cost of a factor $\log(1/\delta)$ (Fact~\ref{fact:median-amplification}).

The lower bound in Section~\ref{sec:Renyi_Entropy_LB} again hides a
heavy coordinate, but requires a more delicate likelihood analysis.
After Poissonization, the true alternative mixture is compared with an
idealized product-Poisson experiment. A truncation step removes rare
large occupancies, and a second-moment expansion shows that the mixture
remains statistically indistinguishable from the null below the scale
$\alpha k^{1-1/\alpha}$. The perturbation is calibrated so that the
corresponding R\'enyi entropies differ by a constant.

The upper bound is restricted to integer orders because
falling-factorial moments give unbiased estimators of integer power sums.
For real noninteger orders in the range
$1.001\le\alpha\le c_0\log k$, we prove the displayed lower bound but do not
claim a matching upper bound.

\subsection{The high-order regime}
\label{subsec:main-high-order}

For large $\alpha$, R\'{e}nyi entropy is uniformly close to
min-entropy. By Lemma~\ref{lem:renyi-min-gap}, for every
$p\in\Delta_k$ and every $\alpha>1$,
\[
   0
   \le
   H_\alpha(p)-H_\infty(p)
   \le
   \frac{\log k}{\alpha-1}.
\]
Hence, if
\(
   \alpha
   \ge
   1+\frac{2\log k}{\varepsilon},
\)
then $|H_\alpha(p)-H_\infty(p)|\le\frac{\varepsilon}{2}$ uniformly over $p\in\Delta_k$. Therefore, an $\varepsilon/2$-accurate estimator of $H_\infty$ is an
$\varepsilon$-accurate estimator of $H_\alpha$, and
\[
   S_\alpha(k,\varepsilon,\delta)
   \le
   S_\infty\left(k,\frac{\varepsilon}{2},\delta\right)
   \le
   C\frac{k}{\varepsilon^2}\log\frac{k}{\delta}.
\]

The min-entropy lower bound transfers in the reverse direction. Fix
$\gamma\in(0,1/2)$ and let
$\delta_\gamma=\log(1+\gamma)$. If
$\varepsilon<\delta_\gamma/4$ and an estimator satisfies
\(
   |\widehat H_\alpha-H_\alpha(p)|
   \le
   \varepsilon,
\)
then, under the same condition on $\alpha$,
\[
\begin{aligned}
   |\widehat H_\alpha-H_\infty(p)|
   &\le
   |\widehat H_\alpha-H_\alpha(p)|
   +
   |H_\alpha(p)-H_\infty(p)| \le
   \frac{3\varepsilon}{2}
   <
   \frac{\delta_\gamma}{2}.
\end{aligned}
\]
Thus, such an estimator would also estimate min-entropy to the accuracy
required in Theorem~\ref{thm:minentropy-lower-bound}. Consequently,
\(
   S_\alpha(k,\varepsilon,1/3)
   \ge
   c_\gamma k\log k.
\)

We conclude that, for every fixed sufficiently small $\varepsilon>0$,
uniformly over
\(
   \alpha
   \ge
   1+\frac{2\log k}{\varepsilon},
\)
the sample complexity is
\(
   S_\alpha(k,\varepsilon,1/3)
   =
   \Theta_\varepsilon(k\log k).
\)
This agrees with the growing-order rate because
$\alpha k^{1-1/\alpha}=\Theta(k\log k)$ when
$\alpha=\Theta(\log k)$.
\subsection{Technical novelty and comparison with previous methods}
\label{subsec:technical-novelty}

The estimators used in this paper are simple, but sharp analysis requires
different proof techniques in different entropy regimes. Min-entropy is
an extreme-value functional determined by the largest atom, whereas
finite-order R\'enyi entropy is governed by multiway collision
probabilities. Our proofs are organized around these distinct
structures rather than treating both quantities as generic symmetric properties. For undefined notation used here, please refer to Section~\ref{subsec:notation}.

\paragraph{Min-entropy: direct analysis of the empirical maximum.}
General-purpose methods for estimating symmetric distribution
properties often reconstruct the probability histogram and then evaluate
the target functional on the reconstructed distribution
\cite{ValiantValiant2011,ValiantValiantJACM2017}. This approach is
effective for aggregate quantities such as Shannon entropy, but it does
not directly apply to min-entropy, which depends only on
$p_\star:=\max_i p_i$. A perturbation of a single coordinate can change
$H_\infty(p)=-\log p_\star$ by a constant even when the overall
histogram changes only slightly.

We therefore analyze the empirical maximum directly. The main issue is
not concentration of the empirical frequency of a true maximizer, but
the possibility that one of many lighter symbols fluctuates upward and
overtakes it. A direct union bound over all $k$ symbols obscures the
relevant scale. Instead, we partition the alphabet into dyadic classes:
the $r$th class contains symbols with probabilities of order
$p_\star/2^r$. Normalization limits the size of this class, while the
probability that one of its symbols reaches the level of the maximum
decreases rapidly with $r$. Balancing these effects yields a
distribution-sensitive concentration bound governed by $1/p_\star$
rather than by the ambient alphabet size. In the worst case
$p_\star=1/k$, this gives the optimal $k\log k$ sample complexity.

\paragraph{Min-entropy lower bound: a hidden extreme coordinate.}
Earlier entropy lower bounds often compare two complete probability
histograms or construct distributions with matching low-order
information
\cite{ValiantValiant2011,AcharyaOrlitskySureshTyagi2017}. For
min-entropy, we instead use a construction tailored to an extreme-value
functional. Starting from the uniform distribution, we increase one
coordinate by a constant factor and choose its location uniformly at
random. The resulting alternatives have a constant min-entropy gap, but
the estimator does not know which coordinate is special.

The likelihood ratio of this mixture admits an exact second-moment
decomposition. If the location were known, evidence for the heavy
coordinate would accumulate quickly; hiding it among $k$ possible
locations introduces the additional search cost. The balance between
these two effects produces the extra $\log k$ factor and yields the
$k\log k$ lower bound. Thus, the logarithmic factor arises from locating
a hidden extreme coordinate rather than from reconstructing unseen
probability mass.

\paragraph{Growing-order R\'enyi upper bound: entropy-specific analysis
of the collision estimator.}
For integer $\alpha\ge2$, our estimator is the classical normalized
falling-factorial statistic
$\widehat P_\alpha=(n)_\alpha^{-1}\sum_{x=1}^k(N_x)_\alpha$, which is the
ordered-tuple form of the birthday estimator studied by Acharya,
Orlitsky, Suresh, and Tyagi
\cite{AcharyaOrlitskySureshTyagi2015,
AcharyaOrlitskySureshTyagi2017}, Obremski and Skorski
\cite{ObremskiSkorski2017}, and Skorski
\cite{skorski2023towards}. The novelty is therefore not the statistic,
but the event that we analyze and the way the dependence on $\alpha$ is
kept explicit.

Acharya, Orlitsky, Suresh, and Tyagi\ prove
$\Theta_\alpha(k^{1-1/\alpha})$ for every fixed integer
$\alpha>1$. Since $\alpha$ is fixed, all overlap multiplicities and
order-dependent constants are absorbed into
$\Theta_\alpha(\cdot)$. Retaining these terms in their moment
calculation yields a worst-case upper bound
$O(\alpha^2k^{1-1/\alpha})$ for constant accuracy and constant success
probability; see Appendix~\ref{app:Acharya_LB_UB}.

Skorski derives an exact U-statistic variance formula for the same
estimator and obtains the distribution-dependent guarantee
$O(\alpha^2\varepsilon^{-2}P_\alpha(p)^{-1/\alpha}
\log(1/\delta))$ for relative estimation of the power sum
\cite{skorski2023towards}. This can be substantially smaller than an
alphabet-dependent bound for distributions of low or moderate
R\'enyi entropy and also supports adaptive estimation. In the worst
case, however, $P_\alpha(p)\ge k^{1-\alpha}$, so the bound becomes
$O(\alpha^2\varepsilon^{-2}k^{1-1/\alpha}\log(1/\delta))$.

Our analysis differs substantially from the above mentioned papers. We use the weaker accuracy requirement that is actually
needed for entropy estimation. Since
$H_\alpha(p)=(1-\alpha)^{-1}\log P_\alpha(p)$, the condition
$|\widehat H_\alpha-H_\alpha(p)|\le\varepsilon$ is equivalent to
$e^{-(\alpha-1)\varepsilon}\le
\widehat P_\alpha/P_\alpha(p)\le
e^{(\alpha-1)\varepsilon}$. Thus, the admissible multiplicative window
for the power sum widens exponentially with
$(\alpha-1)\varepsilon$. Requiring an order-independent relative error
for $P_\alpha$ is sufficient, but unnecessarily strong when $\alpha$
grows.

To exploit this wider window, we represent the estimator as an average
over ordered $\alpha$-tuples and organize its variance by the number
$j$ of shared sample positions between two tuples. The $j$th overlap
term involves the power sum $P_{2\alpha-j}(p)$. We retain the binomial
and factorial coefficients in this decomposition rather than absorbing
them into a fixed-order constant. For moderate
$(\alpha-1)\varepsilon$, the required entropy window is controlled by a
variance estimate and Chebyshev's inequality. For large
$(\alpha-1)\varepsilon$, we use a separate lower-tail argument: the
alphabet is divided into light and heavy symbols, the light contribution
is handled through a restricted variance bound, and the heavy
contribution is controlled by concentration of the associated
falling-factorial counts. This light--heavy argument has no analogue in
the relative-moment analysis in the above papers.

Consequently, one batch of
$O(\varepsilon^{-2}\alpha k^{1-1/\alpha})$ samples succeeds with
constant probability, and independent repetitions followed by a median
give
$O(\varepsilon^{-2}\alpha k^{1-1/\alpha}\log(1/\delta))$ samples with
failure probability at most $\delta$. The entropy-specific analysis
therefore improves the previous worst-case dependence from quadratic to
linear in $\alpha$.

\paragraph{Growing-order R\'enyi lower bound: a hidden-coordinate
mixture and truncated likelihood ratio.}
Our lower-bound argument differs from the previous fixed-order
constructions at the level of the testing experiment. Acharya, Orlitsky, Suresh, and Tyagi\
\cite{AcharyaOrlitskySureshTyagi2015,
AcharyaOrlitskySureshTyagi2017} construct explicit pairs of
distributions with separated R\'enyi entropies and statistically close
sampling laws. Their analysis determines the alphabet-size exponent for
each fixed order, but does not produce a uniform factor growing with
$\alpha$.

Obremski and Skorski
\cite{ObremskiSkorski2017} use a two-point Le Cam construction with priors $P$ and $Q$ in
which a nearly uniform distribution is perturbed at one fixed, known
coordinate. If the perturbation size is $\gamma$, then the one-sample
total variation distance is of order $\gamma$, and the product or hybrid
bound gives
$\operatorname{TV}(P^{\otimes n},Q^{\otimes n})
\le n\operatorname{TV}(P,Q)=O(n\gamma)$. Choosing
$\gamma\asymp k^{-1+1/\alpha}$ creates a constant entropy separation and
yields $n=\Omega(k^{1-1/\alpha})$. The multiplicative constant can be
chosen independently of $\alpha$ in the range considered here; see
Appendix~\ref{app:Obremski_Skorski_LB}. However, a known-location two-point
argument cannot produce an additional factor of $\alpha$.

Our hard experiment uses a completely different approach and hides the informative coordinate under both hypotheses. Let $a:=k^{-1+1/\alpha}$ and
$J\sim\operatorname{Unif}([k])$, and define
$p_J^{(\theta,J)}=\theta a$ and
$p_i^{(\theta,J)}=(1-\theta a)/(k-1)$ for $i\ne J$, with $\theta=1$
under the null prior and $\theta=B$ under the alternative. Both
hypotheses therefore contain a heavy coordinate at the critical scale
$a$, but its location is unknown and its height differs by a constant
factor. The power sum satisfies
$P_\alpha(p^{(\theta,J)})
=k^{1-\alpha}[\theta^\alpha+
(k/(k-1))^{\alpha-1}(1-\theta a)^\alpha]$, so changing $\theta$ from
$1$ to $B$ changes the power sum by a factor exponential in $\alpha$
and hence changes $H_\alpha$ by a constant.

The substantive technical difficulty is proving indistinguishability of the two
mixtures. Conditional on a fixed location $J=j$, a direct two-point
comparison would again stop at the
$\Omega(k^{1-1/\alpha})$ scale. The additional factor $\alpha$ arises
only after averaging over the hidden location and analyzing the mixture
likelihood ratio.

We Poissonize the sample size and set
$\tau:=n/k^{1-1/\alpha}$ and
$\lambda:=n/k=\tau k^{-1/\alpha}$. In the idealized product-Poisson
experiment, the hidden coordinate has mean $\theta\tau$ and every light
coordinate has mean $\lambda$. Relative to the baseline
$\operatorname{Poi}(\lambda)^{\otimes k}$, the mixture likelihood ratio
is
$L_\theta=k^{-1}\sum_{j=1}^kR_{\theta\tau}(N_j)$, where $R_h$ is the
one-coordinate likelihood ratio between
$\operatorname{Poi}(h)$ and $\operatorname{Poi}(\lambda)$.

A direct second-moment calculation is unstable because rare
high-occupancy coordinates may have very large likelihood ratios and
dominate $\mathbb E[L_\theta^2]$. To deal with this issue, we therefore restrict to
$\max_jN_j\le\lfloor\alpha\rfloor$. Poisson upper tails control the
discarded event, while on the truncated event we bound the second moment
of
$k^{-1}\sum_{j=1}^k
(R_{B\tau}(N_j)-R_\tau(N_j))
\mathbf 1\{N_j\le\lfloor\alpha\rfloor\}$.
The cutoff at $\lfloor\alpha\rfloor$ is tied to the entropy order and is
what permits a bound uniform in growing $\alpha$. We then compare the
idealized and normalized experiments through Hellinger tensorization and
finally de-Poissonize.

This proves that the two mixture experiments remain close whenever
$\tau=n/k^{1-1/\alpha}\le c\alpha$, equivalently whenever
$n\le c\alpha k^{1-1/\alpha}$. The improvement over the previous lower
bounds is therefore not a sharper estimate within the same two-point
construction. It comes from a different experiment---a hidden-coordinate
mixture---and from a truncated likelihood-ratio argument that preserves
indistinguishability for an additional factor of $\alpha$ samples.

\begin{remark}[Comparison with fixed-order results]
\label{rem:comparison-fixed-alpha}
For every fixed integer $\alpha>1$, Acharya, Orlitsky, Suresh, and
Tyagi prove
$\Theta_\alpha(k^{1-1/\alpha})$
\cite{AcharyaOrlitskySureshTyagi2015,
AcharyaOrlitskySureshTyagi2017}. Retaining the order-dependence in their
upper-bound calculation yields
$O(\alpha^2k^{1-1/\alpha})$ for constant accuracy and constant success
probability; see Appendix~\ref{app:Acharya_LB_UB}. The worst-case
specialization of the distribution-dependent analysis in
\cite{skorski2023towards} has the same quadratic dependence.

The lower bounds of Acharya, Orlitsky, Suresh, and Tyagi\ and Obremski--Skorski give
$\Omega(k^{1-1/\alpha})$ for real $\alpha>1$. In the
Obremski--Skorski argument, the multiplicative constant can be chosen
independently of $\alpha$ for fixed accuracy; see
Appendix~\ref{app:Obremski_Skorski_LB}. These known-location two-point
constructions do not produce a growing factor of $\alpha$.

Our upper bound reduces the worst-case order-dependence from quadratic
to linear, and the hidden-coordinate lower bound shows that this linear
factor is necessary. Hence, for every fixed $c_0>0$ and every integer
$2\le\alpha\le c_0\log k$,
$S_\alpha(k,\Theta(1),1/3)
=\Theta_{c_0}(\alpha k^{1-1/\alpha})$.
For real $1.001\le\alpha\le c_0\log k$, the same construction gives
$S_\alpha(k,\Theta(1),1/3)
=\Omega_{c_0}(\alpha k^{1-1/\alpha})$, although we do not claim
sharpness for noninteger orders.
\end{remark}

\section{Preliminaries}
\label{sec:preliminaries}

\subsection{Notation and problem formulation}
\label{subsec:notation}

Let $[k]:=\{1,2,\ldots,k\}$, and let
\[
   \Delta_k
   :=
   \left\{
      p=(p_1,\ldots,p_k):
      p_i\ge0,\ 
      \sum_{i=1}^k p_i=1
   \right\}
\]
denote the probability simplex on $[k]$. We observe independent samples
$
   X_1,\ldots,X_n\overset{\mathrm{i.i.d.}}{\sim}p
$
from an unknown distribution $p\in\Delta_k$. For each $i\in[k]$, define
the empirical count and empirical frequency
$
   N_i
   :=
   \sum_{t=1}^n \mathbf 1\{X_t=i\},$ and $\widehat p_i
   :=
   \frac{N_i}{n}.
$
We also write $ p_\star:=\max_{i\in[k]}p_i=\|p\|_\infty$,
    and $
   \widehat p_\star:=\max_{i\in[k]}\widehat p_i.
$
We will also often denote distributions from $\Delta_k$ as capitals, for example $Q \in \Delta_k$ and then refer to probability of $i \in [k]$ as $Q(i)$, especially in complicated formulas.

Unless stated otherwise, all logarithms are natural. Entropies measured
in bits are obtained by replacing $\log$ with $\log_2$; this changes
only constant factors in additive-accuracy sample-complexity bounds.

The Shannon entropy of $p$ is
$
   H(p)
   :=
   -\sum_{i:p_i>0}p_i\log p_i.
$
For $\alpha\in(0,\infty)$ with $\alpha\neq1$, define the power sum
$
   P_\alpha(p)
   :=
   \sum_{i=1}^k p_i^\alpha
$
and the R\'{e}nyi entropy
\[
   H_\alpha(p)
   :=
   \frac{1}{1-\alpha}\log P_\alpha(p).\]
At $\alpha=1$, R\'{e}nyi entropy is defined by continuity and equals
Shannon entropy. Its limiting value as $\alpha\to\infty$ is the
min-entropy
\[
   H_\infty(p)
   :=
   \lim_{\alpha\to\infty}H_\alpha(p)
   =
   -\log p_\star.
\]

An estimator based on $n$ samples is a measurable function
$\widehat F:[k]^n\to\mathbb R$. For a functional
$F:\Delta_k\to\mathbb R$, we say that $\widehat F$ estimates $F$ to
additive error $\varepsilon$ with failure probability at most $\delta$
if
\[
   \sup_{p\in\Delta_k}
   \Prob_{X^n\sim p^n}
   \left(
      \left|\widehat F(X^n)-F(p)\right|>\varepsilon
   \right)
   \le\delta.
\]
The corresponding minimax sample complexity is
\[
   S_F(k,\varepsilon,\delta)
   :=
   \min\left\{
      n:
      \text{there exists an estimator satisfying the above guarantee}
   \right\}.
\]
When $F=H_\alpha$, we write
$S_\alpha(k,\varepsilon,\delta)$; the notation
$S_\infty(k,\varepsilon,\delta)$ refers to min-entropy.

Throughout the paper, $\alpha$ is reserved for the R\'{e}nyi order,
$\varepsilon$ denotes additive estimation accuracy, and $\delta$ denotes
failure probability. Constants denoted by $c,C>0$ are universal unless
their parameter dependence is indicated explicitly.

We use $\mathbb P$ and $\mathbb E$ for probability and expectation, and $\operatorname{Var}$ for variance. The notation $\operatorname{Bin}(n,p)$, $\operatorname{Poi}(\lambda)$, and $\operatorname{Unif}(A)$ denotes the binomial distribution with parameters $(n,p)$, the Poisson distribution with mean $\lambda$, and the uniform distribution on a finite set $A$, respectively. For a probability distribution $P$, we write $P^{\otimes n}$ for its $n$-fold product distribution. For two probability measures $P$ and $Q$ on the same measurable space, $\operatorname{TV}(P,Q)$ denotes their total variation distance, $\chi^2(P\|Q)$ their chi-squared divergence, and $H(P,Q)$ their Hellinger distance. When $P$ is absolutely continuous with respect to $Q$, $dP/dQ$ denotes the Radon--Nikodym derivative. We write $\mathbf 1\{A\}$ for the indicator of an event $A$, and $o(1)$ for a quantity that converges to zero in the asymptotic regime under consideration. For integers $x,r\ge0$, the falling factorial is $(x)_r:=x(x-1)\cdots(x-r+1)$, with $(x)_0:=1$. For positive quantities $a$ and $b$, we write $a\asymp b$ if there exist constants $0<c\le C<\infty$ such that $cb\le a\le Cb$. For positive quantities $f$ and $g$, we write $f\sim g$ if $f/g\to1$ in the relevant limit, and $f\gtrsim_{\varepsilon} g$ if $f\ge c_{\varepsilon}g$ for some constant $c_{\varepsilon}>0$ depending only on $\varepsilon$.
\subsection{Why the Valiant--Valiant relative-earthmover theorem does
not apply to min-entropy}
\label{sec:prelim-rem}

Valiant and Valiant's canonical estimator applies to symmetric
distribution properties satisfying a continuity condition with respect
to the relative earthmover distance between probability histograms
\cite{ValiantValiant2011,ValiantValiantJACM2017}. We recall the relevant
metric and show that min-entropy does not satisfy this condition.

For $p\in\Delta_k$, represent its probability histogram by the
probability measure
$
   \mathcal H_p
   :=
   \sum_{i=1}^k p_i\,\delta_{p_i}
$
on $(0,1]$. Moving a unit of histogram mass from location $x$ to
location $y$ incurs cost $|\log(x/y)|$. The relative earthmover distance
between distributions $p$ and $q$ is therefore
\[
   R(p,q)
   :=
   \inf_{\Gamma\in\Pi(\mathcal H_p,\mathcal H_q)}
   \int
      \left|\log\frac{x}{y}\right|
   \,d\Gamma(x,y),
\]
where $\Pi(\mathcal H_p,\mathcal H_q)$ denotes the set of couplings of
the two histograms.

Min-entropy is not continuous with respect to this metric. Fix
$\rho\in(0,1/2)$, let $U_k$ denote the uniform distribution on $[k]$,
and define $Q_{k,\rho}\in\Delta_k$ by
\[
   Q_{k,\rho}(1)
   :=
   \frac{1+\rho}{k},
   \qquad
   Q_{k,\rho}(i)
   :=
   \frac{1-\rho/(k-1)}{k},
   \quad 2\le i\le k.
\]
Their min-entropies satisfy
\[
   H_\infty(U_k)=\log k,
   \qquad
   H_\infty(Q_{k,\rho})
   =
   \log k-\log(1+\rho).
\]
Thus,
\[
   H_\infty(U_k)-H_\infty(Q_{k,\rho})
   =
   \log(1+\rho),
\]
which is a positive constant independent of $k$.

In contrast, their relative earthmover distance converges to zero. To
see this, couple the histogram mass at $1/k$ under $U_k$ with the two
histogram locations appearing under $Q_{k,\rho}$. Move mass
$(1+\rho)/k$ from $1/k$ to $(1+\rho)/k$, and move the remaining mass
from $1/k$ to
$\bigl(1-\rho/(k-1)\bigr)/k$. This coupling gives (see Fact~\ref{fact:elementary-log-exp})
\[
\begin{aligned}
   R(U_k,Q_{k,\rho})
   &\le
   \frac{1+\rho}{k}\log(1+\rho)
   +
   \left(
      1-\frac{1+\rho}{k}
   \right)
   \left|
      \log\left(1-\frac{\rho}{k-1}\right)
   \right| \le
   \frac{C_\rho}{k}
\end{aligned}
\]
for all sufficiently large $k$. Hence
$R(U_k,Q_{k,\rho})\to0$, while the corresponding min-entropies remain
separated by the constant $\log(1+\rho)$. Thus, min-entropy is not uniformly continuous with a $k$-independent modulus.

Therefore, closeness in relative earthmover distance does not imply
closeness in min-entropy. Consequently, the Valiant--Valiant estimator
for relative-earthmover-continuous properties cannot yield an
$O(k/\log k)$ estimator for min-entropy. Whereas min-entropy is not $(a,b)$-continuous with respect to relative earthmover distance where $a,b$ do not depend on $k$, such continuity holds when $a$ is a function of $k$. For example, one can prove that min-entropy is $(W(b k),b)$-continuous with respect to relative earthmover distance for any $b > 0$, where $W()$ is the Lambert W-function, but this leads to inferior sample complexity such as $O(k^3/\log k)$.

\subsection{R\'{e}nyi entropy versus min-entropy}
\label{subsec:renyi-min-comparison}

The following elementary inequality quantifies the uniform proximity of
high-order R\'{e}nyi entropy to min-entropy.

\begin{lemma}[R\'{e}nyi--min-entropy comparison]
\label{lem:renyi-min-gap}
For every $p\in\Delta_k$ and every $\alpha>1$,
\[
   0
   \le
   H_\alpha(p)-H_\infty(p)
   \le
   \frac{\log k}{\alpha-1}.
\]
\end{lemma}

\begin{proof}
Let $M:=p_\star=\max_i p_i$. Since
$p_i^\alpha\le M^{\alpha-1}p_i$ for every $i$,
\[
   P_\alpha(p)
   =
   \sum_i p_i^\alpha
   \le
   M^{\alpha-1}.
\]
Because $1/(1-\alpha)<0$, this implies
$
   H_\alpha(p)
   \ge
   -\log M
   =
   H_\infty(p).
$

For the reverse inequality, $P_\alpha(p)\ge M^\alpha$. Therefore,
\[
\begin{aligned}
   H_\alpha(p)-H_\infty(p)
   &=
   -\frac{1}{\alpha-1}\log P_\alpha(p)+\log M \\
   &\le
   -\frac{\alpha}{\alpha-1}\log M+\log M \\
   &=
   \frac{-\log M}{\alpha-1}
   =
   \frac{H_\infty(p)}{\alpha-1}.
\end{aligned}
\]
Finally, $M\ge1/k$, and hence $H_\infty(p)\le\log k$. Combining the
bounds proves the result.
\end{proof}

As an immediate consequence, if
$\alpha\ge1+2\log k/\varepsilon$, then
\[
   0
   \le
   H_\alpha(p)-H_\infty(p)
   \le
   \frac{\varepsilon}{2}
   \qquad
   \text{for every }p\in\Delta_k.
\]
Thus, an estimator of $H_\infty(p)$ with additive error at most
$\varepsilon/2$ is also an estimator of $H_\alpha(p)$ with additive
error at most $\varepsilon$. This deterministic comparison alone is not
sharp for the intermediate range $\alpha=O(\log k)$, which is analyzed for R\'{e}nyi entropy directly in Sections~\ref{sec:Renyi_Entropy_UB}
and~\ref{sec:Renyi_Entropy_LB}.

\section{Upper Bound for Estimating Min-Entropy}
\label{sec:Min_Entropy_UB}

In this section, we prove a distribution-dependent upper bound for
min-entropy estimation. The result is stated for a finite or countable
alphabet, slightly extending the finite-alphabet setting introduced in
Section~\ref{subsec:notation}.

Let $p=(p_i)_{i\in\mathcal X}$ be a discrete distribution on a finite
or countable alphabet $\mathcal X$, and define
\[
    p_\star:=\sup_{i\in\mathcal X}p_i.
\]
Since $\sum_i p_i=1$, the supremum is attained. Hence, there exists at
least one index $i^\star\in\mathcal X$ such that
$p_{i^\star}=p_\star$.

Given $n$ i.i.d.\ samples $X_1,\ldots,X_n\sim p$, define
\[
    N_i:=\sum_{t=1}^n\mathbf 1\{X_t=i\},
    \qquad
    \widehat p_i:=\frac{N_i}{n},
    \qquad
    \widehat p_\star:=\sup_{i\in\mathcal X}\widehat p_i.
\]
We consider the plug-in estimator
\[
    \widehat H_\infty:=-\log\widehat p_\star
\]
of the min-entropy
\[
    H_\infty(p):=-\log p_\star.
\]
Throughout this section, we assume that $p_\star\ge q$ for some known
$q\in(0,1)$.

\subsection{Auxiliary concentration bounds}
\label{subsec:minentropy-concentration}

We begin with standard binomial concentration inequalities that will be
used repeatedly (see \cite{mitzenmacher2017probability} for more details).

\begin{lemma}[Chernoff bounds for a binomial random variable]
\label{lem:chernoff-binomial}
Let $X\sim\mathrm{Bin}(n,p)$.

For every $\beta\in(0,1)$,
\[
    \mathbb P\bigl(X\le(1-\beta)np\bigr)
    \le
    \exp\left(-\frac{\beta^2}{2}np\right).
\]
For every $\beta>0$,
\[
    \mathbb P\bigl(X\ge(1+\beta)np\bigr)
    \le
    \exp\left(-\frac{\beta^2}{2+\beta}np\right).
\]
\end{lemma}

For $0<p<t<1$, define the binary Kullback-Leibler (KL) divergence
\[
    D(t\|p)
    :=
    t\log\frac{t}{p}
    +(1-t)\log\frac{1-t}{1-p}.
\]
We extend this definition continuously to $t=1$ by setting
\[
    D(1\|p):=\log\frac1p.
\]

\begin{lemma}[Binomial upper tail in KL form]
\label{lem:bin_KL}
Let $X\sim\mathrm{Bin}(n,p)$ and let $t\in(p,1]$. Then
\[
    \mathbb P\left(\frac{X}{n}\ge t\right)
    \le
    \exp\bigl(-nD(t\|p)\bigr).
\]
\end{lemma}

The following elementary lower bounds on the binary KL divergence will
be useful when the symbol probability is substantially smaller than
$p_\star$.

\begin{lemma}[A convenient lower bound on $D(t\|p)$]
\label{lem:KL_lb}
For every $0<p<t\le1$,
\[
    D(t\|p)
    \ge
    t\log\frac{t}{p}-(t-p),
    \qquad\text{and}\qquad
    D(t\|p)
    \ge
    t\left(\log\frac{t}{p}-1\right).
\]
\end{lemma}

\begin{proof}
We first consider $t<1$. By definition,
\[
    D(t\|p)
    =
    t\log\frac{t}{p}
    +(1-t)\log\frac{1-t}{1-p}.
\]
It is therefore enough to prove that
\[
    (1-t)\log\frac{1-t}{1-p}\ge-(t-p).
\]
Let
\[
    u:=\frac{t-p}{1-p}\in(0,1),
\]
so that
\[
    \frac{1-t}{1-p}=1-u.
\]
By Fact~\ref{fact:elementary-log-exp},
\[
    \log(1-u)\ge-\frac{u}{1-u},
    \qquad u\in(0,1),
\]
we obtain
\[
\begin{aligned}
    (1-t)\log\frac{1-t}{1-p}
    &=
    (1-p)(1-u)\log(1-u) \\
    &\ge
    (1-p)(1-u)\left(-\frac{u}{1-u}\right) \\
    &=
    -(1-p)u \\
    &=
    -(t-p).
\end{aligned}
\]
Consequently,
\[
    D(t\|p)
    \ge
    t\log\frac{t}{p}-(t-p).
\]
Since $t-p\le t$, it also follows that
\[
    D(t\|p)
    \ge
    t\left(\log\frac{t}{p}-1\right).
\]

For $t=1$, the same inequalities follow by continuity. Indeed,
\[
    D(1\|p)=\log\frac1p,
\]
and
\[
    \log\frac1p
    \ge
    \log\frac1p-(1-p)
    \ge
    \log\frac1p-1.
\]
\end{proof}

\subsection{A refined upper-tail bound for the empirical maximum}
\label{subsec:minentropy-upper-tail}

Fix a relative-error parameter $\beta\in(0,1]$ and define
\[
    t:=(1+\beta)p_\star.
\]
We seek to bound
\[
    \mathbb P(\widehat p_\star\ge t)
    =
    \mathbb P\bigl(
        \exists i\in\mathcal X:
        \widehat p_i\ge t
    \bigr).
\]
If $t>1$, this event is empty, so the desired upper bound is immediate.
We may therefore assume throughout the remainder of this subsection that
$t\le1$.

A direct union bound over the entire alphabet can be wasteful when the
alphabet is large. Instead, we partition the symbols into dyadic
probability classes. For $r=0,1,2,\ldots$, define
\[
    B_r
    :=
    \left\{
        i\in\mathcal X:
        \frac{p_\star}{2^{r+1}}
        <
        p_i
        \le
        \frac{p_\star}{2^r}
    \right\}.
\]

The normalization of $p$ immediately controls the size of each bucket.

\begin{lemma}[Bucket-size bound]
\label{lem:bucket_ub}
For every $r\ge0$,
\[
    |B_r|
    \le
    \frac{2^{r+1}}{p_\star}.
\]
\end{lemma}

\begin{proof}
For every $i\in B_r$,
\[
    p_i>\frac{p_\star}{2^{r+1}}.
\]
Therefore,
\[
    1
    \ge
    \sum_{i\in B_r}p_i
    >
    |B_r|\frac{p_\star}{2^{r+1}},
\]
which implies
\[
    |B_r|
    \le
    \frac{2^{r+1}}{p_\star}.
\]
\end{proof}

We now control the contribution of each bucket separately.

\paragraph{Bucket \(r=0\).}
If $i\in B_0$, then
\[
    \frac{p_\star}{2}<p_i\le p_\star.
\]
Define
\[
    \lambda_i:=\frac{t}{p_i}-1.
\]
Then
\[
    \lambda_i
    \ge
    \frac{(1+\beta)p_\star}{p_\star}-1
    =
    \beta.
\]
Moreover, because $p_i>p_\star/2$,
\[
    \lambda_i
    <
    \frac{(1+\beta)p_\star}{p_\star/2}-1
    =
    1+2\beta
    \le3.
\]
By Lemma~\ref{lem:chernoff-binomial},
\[
\begin{aligned}
    \mathbb P(\widehat p_i\ge t)
    &=
    \mathbb P\bigl(
        N_i\ge(1+\lambda_i)np_i
    \bigr) \\
    &\le
    \exp\left(
        -\frac{\lambda_i^2}{2+\lambda_i}np_i
    \right).
\end{aligned}
\]
Using $\lambda_i\ge\beta$, $2+\lambda_i\le5$, and
$p_i\ge p_\star/2$, we obtain
\[
    \frac{\lambda_i^2}{2+\lambda_i}p_i
    \ge
    \frac{\beta^2}{5}\cdot\frac{p_\star}{2}
    =
    \frac{\beta^2p_\star}{10}.
\]
Hence,
\[
    \mathbb P(\widehat p_i\ge t)
    \le
    \exp\left(
        -\frac{\beta^2}{10}np_\star
    \right).
\]
Applying Lemma~\ref{lem:bucket_ub} gives
\[
\begin{aligned}
    \mathbb P\bigl(
        \exists i\in B_0:
        \widehat p_i\ge t
    \bigr)
    &\le
    |B_0|
    \exp\left(
        -\frac{\beta^2}{10}np_\star
    \right) \\
    &\le
    \frac{2}{p_\star}
    \exp\left(
        -\frac{\beta^2}{10}np_\star
    \right).
\end{aligned}
\]

\paragraph{Bucket \(r=1\).}
If $i\in B_1$, then
\[
    \frac{p_\star}{4}<p_i\le\frac{p_\star}{2}.
\]
Therefore,
\[
    \lambda_i
    =
    \frac{t}{p_i}-1
    \ge
    \frac{(1+\beta)p_\star}{p_\star/2}-1
    =
    1+2\beta
    \ge1.
\]
Again, Lemma~\ref{lem:chernoff-binomial} gives
\[
    \mathbb P(\widehat p_i\ge t)
    \le
    \exp\left(
        -\frac{\lambda_i^2}{2+\lambda_i}np_i
    \right).
\]
Since $\lambda_i\ge1$,
\[
    \frac{\lambda_i^2}{2+\lambda_i}
    \ge
    \frac13,
\]
and since $p_i\ge p_\star/4$,
\[
    \mathbb P(\widehat p_i\ge t)
    \le
    \exp\left(
        -\frac{1}{12}np_\star
    \right).
\]
By the union bound in Fact~\ref{fact:basic-probability},
\[
\begin{aligned}
    \mathbb P\bigl(
        \exists i\in B_1:
        \widehat p_i\ge t
    \bigr)
    &\le
    |B_1|
    \exp\left(
        -\frac{1}{12}np_\star
    \right) \\
    &\le
    \frac{4}{p_\star}
    \exp\left(
        -\frac{1}{12}np_\star
    \right).
\end{aligned}
\]

\paragraph{Buckets \(r\ge2\).}
If $i\in B_r$ with $r\ge2$, then
\[
    p_i\le\frac{p_\star}{2^r},
\]
and hence
\[
    \frac{t}{p_i}
    \ge
    (1+\beta)2^r.
\]
By Lemmas~\ref{lem:bin_KL} and~\ref{lem:KL_lb},
\[
\begin{aligned}
    \mathbb P(\widehat p_i\ge t)
    &\le
    \exp\bigl(-nD(t\|p_i)\bigr) \\
    &\le
    \exp\left(
        -nt\left(
            \log\frac{t}{p_i}-1
        \right)
    \right).
\end{aligned}
\]
Since
\[
    \log\frac{t}{p_i}
    \ge
    r\log2+\log(1+\beta),
\]
we obtain
\[
    \mathbb P(\widehat p_i\ge t)
    \le
    \exp\left(
        -nt\bigl(
            r\log2+\log(1+\beta)-1
        \bigr)
    \right).
\]
Therefore, defining
\[
    A_r
    :=
    \mathbb P\bigl(
        \exists i\in B_r:
        \widehat p_i\ge t
    \bigr),
\]
we have
\[
    A_r
    \le
    a_r,
\]
where
\[
    a_r
    :=
    \frac{2^{r+1}}{p_\star}
    \exp\left(
        -nt\bigl(
            r\log2+\log(1+\beta)-1
        \bigr)
    \right).
\]

Now suppose that
\[
    np_\star\ge2.
\]
Since $t=(1+\beta)p_\star\ge p_\star$, we also have
\[
    nt\ge2.
\]
The sequence $(a_r)_{r\ge2}$ decreases geometrically, since
\[
\begin{aligned}
    \frac{a_{r+1}}{a_r}
    &=
    2\exp(-nt\log2) \\
    &\le
    2\exp(-2\log2) \\
    &=
    \frac12.
\end{aligned}
\]
It follows that
\[
    \sum_{r=2}^\infty A_r
    \le
    \sum_{r=2}^\infty a_r
    \le
    2a_2.
\]
Moreover,
\[
    a_2
    =
    \frac{8}{p_\star}
    \exp\left(
        -nt\bigl(
            2\log2+\log(1+\beta)-1
        \bigr)
    \right).
\]
Since $\beta>0$,
\[
    2\log2+\log(1+\beta)-1
    \ge
    2\log2-1
    >
    \frac13.
\]
Therefore,
\[
\begin{aligned}
    \sum_{r=2}^\infty A_r
    &\le
    2a_2 \\
    &\le
    \frac{16}{p_\star}
    \exp\left(-\frac{nt}{3}\right) \\
    &\le
    \frac{16}{p_\star}
    \exp\left(-\frac{np_\star}{3}\right),
\end{aligned}
\]
where the final inequality uses $t\ge p_\star$.

Combining the three bucket estimates yields the following upper-tail
bound.

\begin{lemma}[Upper tail for $\widehat p_\star$]
\label{lem:upper-tail-pstar}
For every $\beta\in(0,1]$ and every $n$ satisfying
$np_\star\ge2$,
\[
\begin{aligned}
    \mathbb P\bigl(
        \widehat p_\star
        \ge
        (1+\beta)p_\star
    \bigr)
    &\le
    \frac{2}{p_\star}
    e^{-(\beta^2/10)np_\star}
    +
    \frac{4}{p_\star}
    e^{-(1/12)np_\star} \\
    &\quad+
    \frac{16}{p_\star}
    e^{-(1/3)np_\star}.
\end{aligned}
\]
\end{lemma}

\subsection{A two-sided relative-error bound}
\label{subsec:minentropy-two-sided}

We next complement the upper-tail estimate with a lower-tail bound.

\begin{lemma}[Lower tail for $\widehat p_\star$]
\label{lem:lower-tail-pstar}
For every $\beta\in(0,1]$,
\[
    \mathbb P\bigl(
        \widehat p_\star
        \le
        (1-\beta)p_\star
    \bigr)
    \le
    \exp\left(
        -\frac{\beta^2}{2}np_\star
    \right).
\]
\end{lemma}

\begin{proof}
Let $i^\star$ satisfy $p_{i^\star}=p_\star$. Since
$\widehat p_\star\ge\widehat p_{i^\star}$,
\[
    \mathbb P\bigl(
        \widehat p_\star
        \le
        (1-\beta)p_\star
    \bigr)
    \le
    \mathbb P\bigl(
        \widehat p_{i^\star}
        \le
        (1-\beta)p_\star
    \bigr).
\]
Since $N_{i^\star}\sim\mathrm{Bin}(n,p_\star)$,
Lemma~\ref{lem:chernoff-binomial} yields
\[
    \mathbb P\bigl(
        \widehat p_{i^\star}
        \le
        (1-\beta)p_\star
    \bigr)
    \le
    \exp\left(
        -\frac{\beta^2}{2}np_\star
    \right).
\]
\end{proof}

Combining the upper- and lower-tail bounds gives the main concentration
result for the empirical maximum.

\begin{theorem}[Relative accuracy of the empirical maximum]
\label{thm:empirical-maximum}
There exist universal constants $C,c>0$---for example,
$C=23$ and $c=1/12$---such that, for every $\beta\in(0,1]$ and every
$n$ satisfying $np_\star\ge2$,
\[
    \mathbb P\left(
        |\widehat p_\star-p_\star|
        >
        \beta p_\star
    \right)
    \le
    \frac{C}{p_\star}
    \exp\bigl(
        -c\beta^2np_\star
    \bigr).
\]
\end{theorem}

\begin{proof}
By the union bound in Fact~\ref{fact:basic-probability},
\[
\begin{aligned}
    \mathbb P\left(
        |\widehat p_\star-p_\star|
        >
        \beta p_\star
    \right)
    &\le
    \mathbb P\bigl(
        \widehat p_\star
        \le
        (1-\beta)p_\star
    \bigr) \\
    &\quad+
    \mathbb P\bigl(
        \widehat p_\star
        \ge
        (1+\beta)p_\star
    \bigr).
\end{aligned}
\]
Since $\beta\in(0,1]$, we have $\beta^2\le1$, and therefore
\[
    e^{-(1/12)np_\star}
    \le
    e^{-(\beta^2/12)np_\star},
\]
\[
    e^{-(1/3)np_\star}
    \le
    e^{-(\beta^2/12)np_\star},
\]
and
\[
    e^{-(\beta^2/10)np_\star}
    \le
    e^{-(\beta^2/12)np_\star}.
\]
Furthermore, because $p_\star\le1$,
\[
    e^{-(\beta^2/2)np_\star}
    \le
    \frac{1}{p_\star}
    e^{-(\beta^2/12)np_\star}.
\]
Applying Lemmas~\ref{lem:lower-tail-pstar} and
\ref{lem:upper-tail-pstar}, we obtain
\[
\begin{aligned}
    \mathbb P\left(
        |\widehat p_\star-p_\star|
        >
        \beta p_\star
    \right)
    &\le
    \frac{1}{p_\star}
    e^{-(\beta^2/12)np_\star}
    +
    \frac{2}{p_\star}
    e^{-(\beta^2/12)np_\star} \\
    &\quad+
    \frac{4}{p_\star}
    e^{-(\beta^2/12)np_\star}
    +
    \frac{16}{p_\star}
    e^{-(\beta^2/12)np_\star}.
\end{aligned}
\]
Summing the coefficients gives
\[
    \mathbb P\left(
        |\widehat p_\star-p_\star|
        >
        \beta p_\star
    \right)
    \le
    \frac{23}{p_\star}
    e^{-(\beta^2/12)np_\star}.
\]
\end{proof}

\subsection{Incorporating the lower bound \(p_\star\ge q\)}
\label{subsec:minentropy-prior-q}

The preceding theorem can be stated in terms of the known lower bound
$q$.

\begin{corollary}[Relative estimation under $p_\star\ge q$]
\label{cor:pstar-relative-q}
Assume that $p_\star\ge q$ for a known $q\in(0,1)$. Then, for every
$\beta\in(0,1]$ and every $n$ satisfying $nq\ge2$,
\[
    \mathbb P\left(
        |\widehat p_\star-p_\star|
        >
        \beta p_\star
    \right)
    \le
    \frac{C}{q}
    \exp\bigl(-c\beta^2nq\bigr),
\]
where one may take $C=23$ and $c=1/12$.

In particular, to guarantee
\[
    \mathbb P\left(
        |\widehat p_\star-p_\star|
        >
        \beta p_\star
    \right)
    \le
    \delta,
\]
it is sufficient to take
\[
    n
    \ge
    \frac{1}{c\beta^2q}
    \log\frac{C}{\delta q}.
\]
With the above constants, this becomes
\[
    n
    \ge
    \frac{12}{\beta^2q}
    \log\frac{23}{\delta q}.
\]
\end{corollary}

\begin{proof}
Apply Theorem~\ref{thm:empirical-maximum} and use
\[
    \frac{1}{p_\star}
    \le
    \frac{1}{q},
    \qquad
    np_\star
    \ge
    nq.
\]
\end{proof}

\subsection{From relative error to min-entropy error}
\label{subsec:minentropy-log-transfer}

The final step converts relative accuracy for $p_\star$ into additive
accuracy for $-\log p_\star$.

\begin{lemma}[Relative error implies logarithmic error]
\label{lem:rel-error-entropy}
Let $x>0$ and suppose that
$|x-1|\le\beta$ for some $\beta\in(0,1)$. Then
\[
    |\log x |
    \le
    \frac{\beta}{1-\beta}.
\]
In particular, if $\beta\le1/2$, then
\[
    |\log x |
    \le
    2\beta.
\]
\end{lemma}

\begin{proof}
If $|x-1|\le\beta$, then
$x\in[1-\beta,1+\beta]$. On this interval,
\[
    \left|
        \frac{d}{dx}\log x
     \right|
    =
    \frac{1}{x}
    \le
    \frac{1}{1-\beta}.
\]
Hence,
\[
    |\log x-\log1 |
    \le
    \frac{1}{1-\beta}|x-1|
    \le
    \frac{\beta}{1-\beta}.
\]
\end{proof}

Applying this lemma with
\[
    x=\frac{\widehat p_\star}{p_\star}
\]
gives the desired min-entropy guarantee.

\begin{corollary}[Min-entropy estimation under $p_\star\ge q$]
\label{cor:minentropy-upper-q}
Assume that $p_\star\ge q$ for a known $q\in(0,1)$. Fix an additive
accuracy $\varepsilon\in(0,1]$ and a failure probability
$\delta\in(0,1)$. Let
\[
    \beta:=\frac{\varepsilon}{2}.
\]
Then $\beta\le1/2$. If
\[
    n
    \ge
    \frac{12}{\beta^2q}
    \log\frac{23}{\delta q}
    =
    48\cdot
    \frac{1}{\varepsilon^2q}
    \log\frac{23}{\delta q},
\]
and also $nq\ge2$, then
\[
    \mathbb P\bigl(
        |\widehat H_\infty-H_\infty(p)|
        >
        \varepsilon
    \bigr)
    \le
    \delta.
\]
Equivalently,
\[
    n
    =
    O\left(
        \frac{1}{\varepsilon^2q}
        \left(
            \log\frac{1}{\delta}
            +
            \log\frac{1}{q}
        \right)
    \right).
\]
For fixed $\varepsilon$ and fixed $\delta$, this simplifies to
\[
    n
    =
    O\left(
        \frac{\log(1/q)}{q}
    \right).
\]
\end{corollary}

\begin{proof}
If
\[
    |\widehat p_\star-p_\star|
    \le
    \beta p_\star,
\]
then
\[
    \left|
        \frac{\widehat p_\star}{p_\star}-1
    \right|
    \le
    \beta.
\]
Since $\beta\le1/2$,
Lemma~\ref{lem:rel-error-entropy} gives
\[
\begin{aligned}
    |\widehat H_\infty-H_\infty(p)|
    &=
    |-\log\widehat p_\star+\log p_\star| \\
    &=
    \left|
        \log\frac{p_\star}{\widehat p_\star}
    \right| \\
    &=
    \left|
        \log\frac{\widehat p_\star}{p_\star}
    \right| \\
    &\le
    2\beta \\
    &=
    \varepsilon.
\end{aligned}
\]
Therefore,
\[
\begin{aligned}
    \mathbb P\bigl(
        |\widehat H_\infty-H_\infty(p)|
        >
        \varepsilon
    \bigr)
    &\le
    \mathbb P\bigl(
        |\widehat p_\star-p_\star|
        >
        \beta p_\star
    \bigr).
\end{aligned}
\]
The result follows from
Corollary~\ref{cor:pstar-relative-q}.
\end{proof}

\begin{remark}[Origin of the logarithmic factor]
The factor $\log(1/q)$ arises from the prefactor $1/p_\star$ in the
concentration inequality. There may be as many as
$O(1/p_\star)$ symbols whose probabilities are comparable to
$p_\star$, and each of them may fluctuate upward and overtake a true
maximizer. The dyadic decomposition makes this phenomenon explicit:
symbols whose probabilities are much smaller than $p_\star$ are
exponentially unlikely to reach the level $(1+\beta)p_\star$, whereas
the dominant contribution comes from symbols with probabilities
comparable to $p_\star$.
\end{remark}

\section{Upper Bound for Estimating R\'enyi Entropy}\label{sec:Renyi_Entropy_UB}

We prove here that the standard falling-factorial collision estimator estimates the order-\(\alpha\) R\'enyi entropy of an unknown distribution on \([k]\) using
\[
  O\!\left(\eps^{-2}\alpha k^{1-1/\alpha}\log\frac1\delta\right)
\]
samples, for every integer \(\alpha\ge 2\), every \(0<\eps\le 1\), and every failure probability \(0<\delta<1\). In particular, for fixed \(\eps\) and \(\delta\), the sample complexity is \(O(\alpha k^{1-1/\alpha})\), and this covers the range \(2\le \alpha\le c_0 \log k\), for any constant $c_0 > 0$. All logarithms in the main theorem are natural; a final corollary states the equivalent bit-valued guarantee.

\subsection{Problem and estimator}

Let \(p=(p_1,\ldots,p_k) \in \Delta_{k}\) be an unknown distribution on \([k]=\{1,\ldots,k\}\), where \(k\ge 2\). For an integer \(\alpha\ge 2\), define the order-\(\alpha\) power sum and recall R\'enyi entropy in natural units
\[
  \Pa(p)=\sum_{x=1}^k p_x^\alpha,
  \qquad
  H_\alpha(p)=\frac{1}{1-\alpha}\log \Pa(p).
\]
For an integer \(r\ge 0\) and a real number \(z\), write
\[
  (z)_r=z(z-1)\cdots(z-r+1),\qquad (z)_0=1.
\]

By Fact~\ref{fact:power-sums},
\begin{equation}\label{eq:power-lower}
  \Pa(p)\ge k^{1-\alpha}.
\end{equation}
Indeed, this is the power-mean inequality, or equivalently, Jensen's inequality applied to the convex function \(t\mapsto t^\alpha\).

Set
\[
  \Kalpha=k^{1-1/\alpha}.
\]
For a numerical constant
\[
  C_0=10^6,
\]
choose one batch size
\begin{equation}\label{eq:batch-size}
  m=\left\lceil C_0\eps^{-2}\alpha \Kalpha\right\rceil.
\end{equation}
Given one batch \(X_1,\ldots,X_m\overset{\mathrm{i.i.d.}}{\sim}p\), let
\[
  N_x=\sum_{i=1}^m \mathbf 1\{X_i=x\}
\]
be the empirical count of symbol \(x\). Define the falling-factorial estimator
\begin{equation}\label{eq:power-estimator}
  \Phat=\frac{1}{(m)_\alpha}\sum_{x=1}^k (N_x)_\alpha.
\end{equation}
Equivalently, \(\Phat\) is the fraction of ordered \(\alpha\)-tuples of distinct sample positions whose observed symbols are all equal.

To avoid the harmless event \(\Phat=0\) inside the logarithm, define the clipped estimator
\begin{equation}\label{eq:clipped-power}
  \Phat^\circ=
  \max\left\{\Phat,\ e^{-(\alpha-1)\eps}k^{1-\alpha}\right\}.
\end{equation}
The corresponding one-batch entropy estimator is
\begin{equation}\label{eq:entropy-estimator-batch}
  \Hhat=\frac{1}{1-\alpha}\log \Phat^\circ.
\end{equation}
The final estimator uses independent repetition. Let
\begin{equation}\label{eq:num-batches}
  B=\left\lceil 8\log\frac1\delta\right\rceil.
\end{equation}
Draw \(B\) independent batches of size \(m\), compute the \(B\) one-batch entropy estimates \(\widehat H_{\alpha,1},\ldots,\widehat H_{\alpha,B}\), and output
\begin{equation}\label{eq:median-estimator}
  \widehat H_{\alpha,\mathrm{med}}=\Med\{\widehat H_{\alpha,1},
  \ldots,\widehat H_{\alpha,B}\}.
\end{equation}

\subsection{Main theorem}

\begin{theorem}[Integer-order R\'enyi entropy upper bound]\label{thm:main_1}
Let \(k\ge 2\) and \(p \in \Delta_k\), let \(\alpha\ge 2\) be an integer, \(0<\eps\le 1\), and \(0<\delta<1\). The estimator \eqref{eq:median-estimator}, with \(m\) and \(B\) chosen as in \eqref{eq:batch-size} and \eqref{eq:num-batches}, satisfies
\[
  \mathbb P\left(
    \left|\widehat H_{\alpha,\mathrm{med}}-H_\alpha(p)\right|>\eps
  \right)
  \le \delta.
\]
It uses
\[
  Bm
  =
  \left\lceil 8\log\frac1\delta\right\rceil
  \left\lceil C_0\eps^{-2}\alpha k^{1-1/\alpha}\right\rceil
\]
samples. Hence its sample complexity is
\[
  O\!\left(\eps^{-2}\alpha k^{1-1/\alpha}\log\frac1\delta\right).
\]
In particular, for fixed constant \(\eps\) and \(\delta\), it uses \(O(\alpha k^{1-1/\alpha})\) samples. The result holds for every integer \(\alpha\ge 2\), and therefore in particular for every integer \(2\le \alpha\le c_0 \log k\).
\end{theorem}

The rest of this section proves Theorem~\ref{thm:main_1}.

\subsection{Unbiasedness and variance preliminaries}

\begin{lemma}[Unbiasedness]\label{lem:unbiased}
For one batch,
\[
  \mathbb E[\Phat]=\Pa(p).
\]
\end{lemma}

\begin{proof}
For each symbol \(x\), \(N_x\sim \operatorname{Bin}(m,p_x)\). The factorial moment identity for a binomial random variable gives (see Fact \ref{app:Fact_binomial-factorial} and Section \ref{app:appl_falling} for proof in Appendix):
\[
  \mathbb E[(N_x)_\alpha]=(m)_\alpha p_x^\alpha.
\]
Therefore
\[
  \mathbb E[\Phat]
  =\frac{1}{(m)_\alpha}\sum_{x=1}^k\mathbb E[(N_x)_\alpha]
  =\sum_{x=1}^k p_x^\alpha
  =\Pa(p).
\]
\end{proof}

For a subset \(A\subseteq [k]\), define
\[
  P_s(A)=\sum_{x\in A}p_x^s
\]
and the restricted estimator
\begin{equation}\label{eq:restricted-estimator}
  \widehat P_\alpha(A)=
  \frac{1}{(m)_\alpha}\sum_{x\in A}(N_x)_\alpha.
\end{equation}
This restricted estimator is only a proof device; the algorithm does not need to know \(A\).

\begin{lemma}[Tuple-overlap variance bound]\label{lem:variance}
Let \(A\subseteq [k]\). If \(m\ge 2\alpha\), then
\begin{equation}\label{eq:variance-bound}
  \Var\!\left(\widehat P_\alpha(A)\right)
  \le
  \sum_{j=1}^{\alpha}
  \binom{\alpha}{j}^2 j!
  \left(\frac{2}{m}\right)^j
  P_{2\alpha-j}(A).
\end{equation}
\end{lemma}

\begin{proof}
Let
\[
  \mathcal T=\{(i_1,\ldots,i_\alpha): i_r\in[m],\ i_r\ne i_s\text{ for }r\ne s\}
\]
be the set of ordered \(\alpha\)-tuples of distinct sample positions. Then \(|\mathcal T|=(m)_\alpha\). For \(t=(i_1,\ldots,i_\alpha)\in\mathcal T\), set
\[
  I_t^A=\mathbf 1\{X_{i_1}=\cdots=X_{i_\alpha}\in A\}.
\]
Then (see Fact \ref{app:fact:tuple-sum-estimator} in Appendix for proof):
\[
  \widehat P_\alpha(A)=\frac{1}{(m)_\alpha}\sum_{t\in\mathcal T} I_t^A.
\]

If two tuples \(s,t\in\mathcal T\) have disjoint sets of sample positions, then \(I_s^A\) and \(I_t^A\) are independent. Thus their covariance is zero. If they overlap in exactly \(j\ge 1\) sample positions, then the event \(I_s^A=I_t^A=1\) forces all \(2\alpha-j\) distinct samples involved to be equal to the same symbol in \(A\). Therefore (for proof see Fact \ref{app:fact:overlap-indicator-moment} in Appendix):
\[
  \mathbb E[I_s^A I_t^A]=P_{2\alpha-j}(A).
\]
For a fixed tuple \(s\), the number of tuples \(t\) overlapping \(s\) in exactly \(j\) locations is at most
\[
  \binom{\alpha}{j}^2 j!(m-\alpha)_{\alpha-j}.
\]
Indeed, choose the \(j\) overlapped coordinates of \(s\), choose the \(j\) coordinates of \(t\) in which they appear, biject them in \(j!\) ways, and choose the remaining \(\alpha-j\) coordinates of \(t\) from the \(m-\alpha\) positions outside \(s\).

Since \(\operatorname{Cov}(I_s^A,I_t^A)\le \mathbb E[I_s^A I_t^A]\), summing over all overlapping pairs gives
\[
  \Var\!\left(\widehat P_\alpha(A)\right)
  \le
  \frac{1}{(m)^2_\alpha}
  \sum_{s \in \mathcal T}
  \sum_{j=1}^\alpha
  \binom{\alpha}{j}^2j!(m-\alpha)_{\alpha-j}
  P_{2\alpha-j}(A)
  =
  \frac{1}{(m)_\alpha}
  \sum_{j=1}^\alpha
  \binom{\alpha}{j}^2j!(m-\alpha)_{\alpha-j}
  P_{2\alpha-j}(A) \, ,
\] where the equality follows by the fact that $|\mathcal T| = (m)_\alpha$ and that the inner sum does not depend on $s$.
For \(m\ge 2\alpha\),
\[
  \frac{(m-\alpha)_{\alpha-j}}{(m)_\alpha}
  \le (m-\alpha+1)^{-j}
  \le \left(\frac{2}{m}\right)^j.
\]
This proves \eqref{eq:variance-bound}.
\end{proof}

\begin{lemma}[Power-sum bounds]\label{lem:powersum}
Let \(\Kalpha=k^{1-1/\alpha}\). For every \(1\le j\le \alpha\),
\begin{equation}\label{eq:global-powersum-bound}
  \frac{P_{2\alpha-j}([k])}{\Pa(p)^2}
  \le \Kalpha^j.
\end{equation}
Moreover, let \(A\subseteq[k]\). Suppose every \(x\in A\) satisfies
\[
  p_x\le \frac{\tau}{\Kalpha}
\]
for some \(0<\tau<1\), and suppose
\[
  P_\alpha(A)\ge \frac12\Pa(p).
\]
Then, for every \(1\le j\le \alpha\),
\begin{equation}\label{eq:light-powersum-bound}
  \frac{P_{2\alpha-j}(A)}{P_\alpha(A)^2}
  \le 2\tau^{\alpha-j}\Kalpha^j.
\end{equation}
\end{lemma}

\begin{proof}
Since \(p_x^\alpha\le \Pa(p)\), we have
\[
  p_x\le \Pa(p)^{1/\alpha}
  \qquad\text{for all } x \in [k].
\]
Hence
\[
  P_{2\alpha-j}([k])
  =\sum_x p_x^\alpha p_x^{\alpha-j}
  \le \Pa(p)^{(\alpha-j)/\alpha} \sum_x p_x^\alpha
  =\Pa(p)^{(\alpha-j)/\alpha}  \Pa(p)
  =\Pa(p)^{(2\alpha-j)/\alpha}.
\]
Dividing by \(\Pa(p)^2\),
\[
  \frac{P_{2\alpha-j}([k])}{\Pa(p)^2}
  \le \Pa(p)^{-j/\alpha}.
\]
Using \eqref{eq:power-lower},
\[
  \Pa(p)^{-j/\alpha}
  \le \bigl(k^{1-\alpha}\bigr)^{-j/\alpha}
  =k^{j(\alpha-1)/\alpha}
  =\Kalpha^j.
\]
This proves \eqref{eq:global-powersum-bound}.

For the second claim, if \(x\in A\), then \(p_x\le \tau/\Kalpha\). Therefore
\[
  P_{2\alpha-j}(A)
  =\sum_{x\in A}p_x^\alpha p_x^{\alpha-j}
  \le \left(\frac{\tau}{\Kalpha}\right)^{\alpha-j}P_\alpha(A).
\]
Also,
\[
  P_\alpha(A)\ge \frac12\Pa(p)\ge \frac12 k^{1-\alpha}=\frac12\Kalpha^{-\alpha}.
\]
Thus
\[
  \frac{P_{2\alpha-j}(A)}{P_\alpha(A)^2}
  \le
  \left(\frac{\tau}{\Kalpha}\right)^{\alpha-j}
  \frac{1}{P_\alpha(A)}
  \le
  2\tau^{\alpha-j}\Kalpha^j.
\]
\end{proof}

\subsection{One-batch success probability}

Let
\[
  P=\Pa(p),
  \qquad
  \eta=(\alpha-1)\eps.
\]
If
\begin{equation}\label{eq:power-success}
  e^{-\eta}P\le \Phat^\circ\le e^\eta P,
\end{equation}
then
\[
  |\Hhat-H_\alpha(p)|
  =\frac{1}{\alpha-1}\left|\log\frac{\Phat^\circ}{P}\right|
  \le \eps.
\]
So the goal is to prove \eqref{eq:power-success} with constant probability for one batch.

Because of \eqref{eq:power-lower}, the clipping floor obeys
\[
  e^{-\eta}k^{1-\alpha}\le e^{-\eta}P.
\]
Consequently, if
\[
  e^{-\eta}P\le \Phat\le e^\eta P,
\]
then \eqref{eq:power-success} holds for \(\Phat^\circ\). It is therefore enough to prove the corresponding event for the unclipped estimator \(\Phat\).

\begin{lemma}[One-batch success]\label{lem:one-batch}
With \(m\) as in \eqref{eq:batch-size} and \(C_0=10^6\),
\[
  \mathbb P\left(e^{-\eta}P\le \Phat^\circ\le e^\eta P\right)\ge \frac34.
\]
Equivalently, the one-batch entropy estimate satisfies
\[
  \mathbb P\left(|\Hhat-H_\alpha(p)|\le \eps\right)\ge \frac34.
\]
\end{lemma}

\begin{proof}
The batch size satisfies \(m\ge C_0\alpha\), hence \(m\ge 2\alpha\). We use Lemma~\ref{lem:variance} throughout.

\subsection*{Regime 1: \(\eta\le 4\)}
Let
\[
  t=1-e^{-\eta}.
\]
If \(|\Phat-P|\le tP\), then \(\Phat\ge e^{-\eta}P\), and also
\[
  \Phat\le (1+t)P=(2-e^{-\eta})P\le e^\eta P, \mbox{ because } \forall \eta : 2 \leq e^{-\eta} + e^\eta \, .
\] The inequality $2 \leq e^{-\eta} + e^\eta$ follows by the inequality between geometric and arithmetic means.
Thus it suffices to show that \(|\Phat-P|\le tP\) with probability at least \(15/16\).

By Lemmas~\ref{lem:variance} and~\ref{lem:powersum}, applied with \(A=[k]\),
\[
  \frac{\Var(\Phat)}{P^2}
  \le
  \sum_{j=1}^{\alpha}
  \binom{\alpha}{j}^2j!
  \left(\frac{2\Kalpha}{m}\right)^j.
\]
Since
\[
  \binom{\alpha}{j}^2j!
  =\binom{\alpha}{j}(\alpha)_j
  \le \binom{\alpha}{j}\alpha^j,
\]
we get by the binomial theorem
\begin{align*}
  \frac{\Var(\Phat)}{P^2}
  &\le
  \sum_{j=1}^{\alpha}
  \binom{\alpha}{j}
  \left(\frac{2\alpha\Kalpha}{m}\right)^j \\
  &=
  \left(1+\frac{2\alpha\Kalpha}{m}\right)^\alpha-1.
\end{align*}
The choice of \(m\) implies
\[
  \frac{2\alpha\Kalpha}{m}\le \frac{2\eps^2}{C_0}.
\]
Therefore, by using that $(1+x)^{\alpha} \leq \exp(\alpha x)$ holds for any $\alpha \geq 0, x > -1$:
\[
  \frac{\Var(\Phat)}{P^2}
  \le
  \exp\left(\frac{2\alpha\eps^2}{C_0}\right)-1.
\]
In the present regime, \(\eta=(\alpha-1)\eps\le 4\), so
\[
  \alpha\le 1+\frac4\eps.
\]
Because \(0<\eps\le 1\),
\[
  \frac{2\alpha\eps^2}{C_0}\le \frac{10}{C_0}\le 1.
\]
Using \(e^x-1\le 2x\) for \(0\le x\le 1\),
\begin{equation}\label{eq:small-eta-var}
  \frac{\Var(\Phat)}{P^2}
  \le \frac{4\alpha\eps^2}{C_0}.
\end{equation}

The function \(x\mapsto 1-e^{-x}\) is concave, so for \(0\le x\le 4\),
\[
  1-e^{-x}\ge c_*x,
  \qquad
  c_*:=\frac{1-e^{-4}}{4}.
\]
Therefore
\[
  t^2\ge c_*^2\eps^2(\alpha-1)^2.
\]
Combining this with \eqref{eq:small-eta-var}, and using \(\alpha/(\alpha-1)^2\le 2\) for \(\alpha\ge 2\), gives
\[
  \frac{\Var(\Phat)}{t^2P^2}
  \le
  \frac{8}{C_0c_*^2}
  <\frac1{16}.
\]
Chebyshev's inequality (Fact~\ref{fact:basic-probability}) yields
\[
  \mathbb P\bigl(|\Phat-P|>tP\bigr)<\frac1{16}.
\]
Thus Regime 1 succeeds with probability at least \(15/16\).

\subsection*{Regime 2: \(\eta>4\)}
The upper tail is controlled by Markov's inequality and Lemma~\ref{lem:unbiased}:
\[
  \mathbb P(\Phat>e^\eta P)
  \le e^{-\eta}<e^{-4}<\frac1{16}.
\]
It remains to bound the lower tail.

Fix
\[
  \tau=\frac1{16}.
\]
Define the light and heavy sets
\[
  L=\left\{x:p_x\le \frac{\tau}{\Kalpha}\right\},
  \qquad
  H=[k]\setminus L.
\]
Let
\[
  P_L=P_\alpha(L),
  \qquad
  P_H=P_\alpha(H).
\]
Exactly one of the following two cases 2a, 2b below holds.

\paragraph{Case 2a: \(P_L\ge P/2\).}
Let \(\widehat P_L=\widehat P_\alpha(L)\). Then \(\Phat\ge \widehat P_L\) and \(\mathbb E\widehat P_L=P_L\). 
We begin with
\[
P_{2\alpha-j}(L)
=
\sum_{x\in L}p_x^{2\alpha-j}
=
\sum_{x\in L}p_x^\alpha p_x^{\alpha-j}.
\]
Since $x\in L$ implies $p_x\le \frac{\tau}{\Kalpha}$,
we have
\[
P_{2\alpha-j}(L)
\le
\left(\frac{\tau}{\Kalpha}\right)^{\alpha-j}
\sum_{x\in L}p_x^\alpha
=
\left(\frac{\tau}{\Kalpha}\right)^{\alpha-j}P_L.
\]
Dividing by $P_L^2$ gives
\[
\frac{P_{2\alpha-j}(L)}{P_L^2}
\le
\left(\frac{\tau}{\Kalpha}\right)^{\alpha-j}\frac{1}{P_L}.
\]

Next, by assumption, $P_L\ge \frac12P$, and by 
$P=\sum_{x=1}^k p_x^\alpha\ge k^{1-\alpha}$ and since
$
\Kalpha^{-\alpha}
=
\left(k^{1-1/\alpha}\right)^{-\alpha}
=
k^{1-\alpha}$, we obtain
\[
P_L\ge \frac12 \Kalpha^{-\alpha},
\mbox{ and hence }
\frac1{P_L}\le 2\Kalpha^\alpha.
\]
Substituting this estimate into the previous inequality yields
\[
\frac{P_{2\alpha-j}(L)}{P_L^2}
\le
\left(\frac{\tau}{\Kalpha}\right)^{\alpha-j}
2\Kalpha^\alpha.
\]
Simplifying the right-hand side,
\[
\left(\frac{\tau}{\Kalpha}\right)^{\alpha-j}
2\Kalpha^\alpha
=
2\tau^{\alpha-j}\Kalpha^{\alpha-(\alpha-j)}
=
2\tau^{\alpha-j}\Kalpha^j.
\]
Therefore
\[
\frac{P_{2\alpha-j}(L)}{P_L^2}
\le
2\tau^{\alpha-j}\Kalpha^j.
\]

Now we divide the variance bound by $P_L^2$ and use Lemma~\ref{lem:variance}:
\[
\frac{\Var(\widehat P_L)}{P_L^2}
\le
\sum_{j=1}^{\alpha}
\binom{\alpha}{j}^2 j!
\left(\frac{2}{m}\right)^j
\frac{P_{2\alpha-j}(L)}{P_L^2}.
\]
Using the just-proved bound gives
\[
\frac{\Var(\widehat P_L)}{P_L^2}
\le
\sum_{j=1}^{\alpha}
\binom{\alpha}{j}^2 j!
\left(\frac{2}{m}\right)^j
\left(2\tau^{\alpha-j}\Kalpha^j\right).
\]
Equivalently,
\[
\frac{\Var(\widehat P_L)}{P_L^2}
\le
2
\sum_{j=1}^{\alpha}
\binom{\alpha}{j}^2 j!
\left(\frac{2\Kalpha}{m}\right)^j
\tau^{\alpha-j}.
\]
Using \(\binom{\alpha}{j}^2j!\le \binom{\alpha}{j}\alpha^j\),
\begin{align*}
  \frac{\Var(\widehat P_L)}{P_L^2}
  &\le
  2\sum_{j=1}^{\alpha}
  \binom{\alpha}{j}
  \left(\frac{2\alpha\Kalpha}{m}\right)^j
  \tau^{\alpha-j} \\
  &\le
  2\left(\tau+\frac{2\alpha\Kalpha}{m}\right)^\alpha.
\end{align*}
Since \(0<\eps\le 1\), and by the choice of $m$,
\[
  \frac{2\alpha\Kalpha}{m}\le \frac{2\eps^2}{C_0}\le \frac{2}{C_0}.
\]
With \(\tau=1/16\) and \(C_0=10^6\),
\[
  \tau+\frac{2}{C_0}<\frac18.
\]
Hence, since \(\alpha\ge 2\),
\[
  \frac{\Var(\widehat P_L)}{P_L^2}
  \le 2\left(\frac18\right)^\alpha
  \le \frac1{32}.
\]
Chebyshev's inequality (Fact~\ref{fact:basic-probability}) gives
\[
  \mathbb P\left(\widehat P_L<\frac12P_L\right)
  \le
  4\frac{\Var(\widehat P_L)}{P_L^2}
  \le \frac18.
\]
On the complementary event,
\[
  \Phat\ge \widehat P_L\ge \frac12P_L\ge \frac14P.
\]
Because \(\eta>4>\log 4\), we have \(P/4\ge e^{-\eta}P\). Therefore
\[
  \mathbb P(\Phat<e^{-\eta}P)\le \frac18
\]
in Case 2a.

\paragraph{Case 2b: \(P_H>P/2\).}
For each \(x\in H\),
\[
  p_x>\frac{\tau}{\Kalpha}.
\]
Set
\[
  \mu_x=mp_x.
\]
Then
\[
  \mu_x>m\frac{\tau}{\Kalpha}
  \ge C_0\tau\eps^{-2}\alpha.
\]
Since \(C_0=10^6\), \(\tau=1/16\), and \(0<\eps\le 1\), in particular
\begin{equation}\label{eq:mu-large}
  \mu_x\ge \frac{16\alpha}{\eps}.
\end{equation}
For each \(x\in H\), define
\[
  G_x=\left\{N_x\ge \left(1-\frac{\eps}{16}\right)\mu_x\right\}.
\]
By the multiplicative Chernoff lower-tail bound,
\[
  \mathbb P(G_x^c)
  \le
  \exp\left(-\frac{\eps^2\mu_x}{512}\right)
  \le
  \exp\left(-\frac{C_0\tau\alpha}{512}\right)
  <\frac1{64}.
\]
Define the random missed heavy \(\alpha\)-power mass
\[
  W=\sum_{x\in H}p_x^\alpha\mathbf 1_{G_x^c}.
\]
By linearity of expectation,
\[
  \mathbb E[W]
  \le \frac1{64}P_H.
\]
Markov's inequality gives
\[
  \mathbb P\left(W>\frac12P_H\right)
  \le \frac1{32}.
\]

Now suppose \(G_x\) occurs for some $x \in H$. Using \eqref{eq:mu-large},
\[
  N_x-\alpha
  \ge
  \left(1-\frac{\eps}{16}\right)\mu_x-\alpha
  \ge
  \left(1-\frac{\eps}{8}\right)\mu_x.
\]
Thus \(N_x\ge \alpha\) and
\[
  (N_x)_\alpha
  \ge
  (N_x-\alpha+1)^\alpha
  \ge
  (N_x-\alpha)^\alpha
  \ge
  \left(1-\frac{\eps}{8}\right)^\alpha\mu_x^\alpha.
\]
Since \((m)_\alpha\le m^\alpha\) and \(\mu_x=mp_x\),
\begin{equation}\label{eq:heavy-symbol-contribution}
  \frac{(N_x)_\alpha}{(m)_\alpha}
  \ge
  \left(1-\frac{\eps}{8}\right)^\alpha p_x^\alpha.
\end{equation}
For \(0<\eps\le 1\), \(\log(1-\eps/8)\ge -\eps/4\) (see Fact \ref{app:Fact_4} in the Appendix), and hence
\[
  \left(1-\frac{\eps}{8}\right)^\alpha\ge e^{-\eps\alpha/4}.
\]
Combining this with \eqref{eq:heavy-symbol-contribution}, on the event \(W\le P_H/2\),
\begin{align*}
  \Phat
  &\ge
  \sum_{x\in H: G_x}
  \frac{(N_x)_\alpha}{(m)_\alpha} \\
  &\ge
  e^{-\eps\alpha/4}
  \sum_{x\in H:G_x}p_x^\alpha \\
  &=
  e^{-\eps\alpha/4}(P_H-W) \\
  &\ge
  \frac12e^{-\eps\alpha/4}P_H
  \ge
  \frac14e^{-\eps\alpha/4}P.
\end{align*}
Since \(\eta=\eps(\alpha-1)>4\) and \(0<\eps\le 1\),
\[
  \eps\left(\frac{3\alpha}{4}-1\right)
  =
  \frac34\eps(\alpha-1)-\frac\eps4
  >3-\frac14
  >\log 4.
\]
The following inequality can be shown to be equivalent to the above inequality (by taking the logarithms of its both sides and rearranging the terms):
\[
  \frac14e^{-\eps\alpha/4} > e^{-\eps(\alpha-1)}=e^{-\eta}.
\]
Therefore, on the event \(W\le P_H/2\),
\[
  \Phat\ge e^{-\eta}P \, ,
\] which means that $\{ W\le P_H/2 \} \subseteq \{ \Phat\ge e^{-\eta}P \}$.
It follows that $\{ \Phat < e^{-\eta}P \} \subseteq \{ W > P_H/2 \}$ and by $\mathbb P(W > P_H/2)\le \frac1{32}$, we obtain
\[
  \mathbb P(\Phat<e^{-\eta}P)\le \frac1{32}
\]
in Case 2b.

Combining the two cases in Regime 2,
\[
  \mathbb P(\Phat<e^{-\eta}P)\le \frac18,
  \qquad
  \mathbb P(\Phat>e^\eta P)\le \frac1{16}.
\]
Thus Regime 2 succeeds with probability at least
\[
  1-\frac18-\frac1{16}=\frac{13}{16}>\frac34.
\] We have also shown that Regime 1 succeeds with probability at least \(15/16 > \frac34\).
Therefore, Lemma~\ref{lem:one-batch} follows from the two regimes.
\end{proof}

\subsection{Amplification and proof of the main theorem}

\begin{proof}[Proof of Theorem~\ref{thm:main_1}]
For each batch \(b\), let
\[
  E_b=\left\{|\widehat H_{\alpha,b}-H_\alpha(p)|\le \eps\right\}.
\]
By Lemma~\ref{lem:one-batch},
\[
  \mathbb P(E_b)\ge \frac34.
\]
The events \(E_1,\ldots,E_B\) are independent because the batches are independent. Let
\[
  Z_b=\mathbf 1_{E_b^c}.
\]
Then \(\mathbb E Z_b\le 1/4\). The median fails only if at least half of the batch estimates fail, namely only if
\[
  \sum_{b=1}^B Z_b\ge \frac{B}{2}.
\]
Hoeffding's inequality gives
\[
  \mathbb P\left(\sum_{b=1}^B Z_b\ge \frac{B}{2}\right)
  \le
  \exp\left(-2B\left(\frac12-\frac14\right)^2\right)
  =e^{-B/8}.
\]
By the definition of \(B\) in \eqref{eq:num-batches}, \(e^{-B/8}\le \delta\). Hence
\[
  \mathbb P\left(|\widehat H_{\alpha,\mathrm{med}}-H_\alpha(p)|>\eps\right)
  \le \delta.
\]
The sample count is exactly \(Bm\), which gives the stated bound.
\end{proof}

\subsection{Version for entropy measured in bits}

The preceding theorem uses natural logarithms, so the error \(\eps\) is measured in nats. If R\'enyi entropy is measured in bits,
\[
  H_{\alpha,2}(p)=\frac{1}{1-\alpha}\log_2 \Pa(p),
\]
then additive error \(\eps\) bits is equivalent to additive error \(\eps\log 2\) nats. Therefore Theorem~\ref{thm:main_1} immediately gives the following.

\begin{corollary}[Bit-valued R\'enyi entropy]\label{cor:bits}
Let \(0<\eps\le 1\) be a target additive error in bits. The estimator above, with the batch size
\[
  m=
  \left\lceil
  C_0(\eps\log 2)^{-2}\alpha k^{1-1/\alpha}
  \right\rceil,
\]
and \(B=\lceil 8\log(1/\delta)\rceil\) batches, satisfies
\[
  \mathbb P\left(
  |\widehat H_{\alpha,2,\mathrm{med}}-H_{\alpha,2}(p)|>\eps
  \right)
  \le \delta.
\]
In particular, the sample complexity in bits is still
\[
  O\!\left(\eps^{-2}\alpha k^{1-1/\alpha}\log\frac1\delta\right),
\]
with a universal constant larger by a factor \((\log 2)^{-2}\).
\end{corollary}

\subsection{Summary of the construction}

For integer \(\alpha\ge 2\), the algorithm is:
\begin{enumerate}[label=\arabic*.]
  \item Split the samples into \(B=\lceil 8\log(1/\delta)\rceil\) independent batches, each of size \(m=\lceil C_0\eps^{-2}\alpha k^{1-1/\alpha}\rceil\) for natural-log entropy.
  \item In each batch, compute counts \(N_x\) and the unbiased power-sum estimate
  \[
    \Phat=\frac{1}{(m)_\alpha}\sum_x (N_x)_\alpha.
  \]
  \item Clip \(\Phat\) from below at \(e^{-(\alpha-1)\eps}k^{1-\alpha}\), take
  \[
    \Hhat=(1-\alpha)^{-1}\log \Phat^\circ,
  \]
  and return the median over batches.
\end{enumerate}
The proof above establishes the claimed sample complexity uniformly over all distributions on \([k]\).

\section{Lower Bounds -- Le Cam's Reduction}

We will use the following elementary form of Le Cam's reduction for separated priors in proving our both lower bounds. It applies both to the fixed-sample experiment and to Poissonized count-vector experiments, used in our lower bounds.

This is the standard Le Cam two-hypotheses argument; see, e.g.,
Tsybakov~\cite[Chapter~2]{Tsybakov2009}.
We include the short proof for completeness.

\begin{lemma}[Le Cam reduction for separated priors]\label{lem:Le_Cam}
Let \(F:\Delta_k\to\mathbb R\) be a distribution functional. Let
\(\Pi_0,\Pi_1\) be two priors supported on subsets
\(\mathcal P_0,\mathcal P_1\subseteq \Delta_k\). Let \(\mathsf P_0\) and
\(\mathsf P_1\) denote the induced laws of the observations after first drawing
\(p\sim \Pi_i\) and then sampling from \(p\) according to the experiment under
consideration.

Suppose that the two prior supports are \(2\varepsilon\)-separated in the
functional value:
\[
\inf_{p\in\mathcal P_0,\;q\in\mathcal P_1}
|F(p)-F(q)| > 2\varepsilon .
\]
Then every estimator \(\widehat F\) satisfies
\[
\sup_{p\in\mathcal P_0\cup\mathcal P_1}
\mathbb P_p\bigl(|\widehat F-F(p)|>\varepsilon\bigr)
\ge
\frac{1-\operatorname{TV}(\mathsf P_0,\mathsf P_1)}{2}.
\]
Equivalently, if \(\widehat F\) estimates \(F\) to error at most
\(\varepsilon\) with success probability at least \(q\) uniformly over
\(\mathcal P_0\cup\mathcal P_1\), then
\[
\operatorname{TV}(\mathsf P_0,\mathsf P_1)\ge 2q-1.
\]
\end{lemma}

\begin{proof}
Let
\[
A_i=\{x\in\mathbb R:\operatorname{dist}(x,F(\mathcal P_i))\le \varepsilon\},
\qquad i=0,1 .
\]
The separation assumption implies that \(A_0\cap A_1=\varnothing\).
Given an estimator output \(\widehat F\), define a test between the two priors
by deciding hypothesis \(i\) if \(\widehat F\in A_i\), with arbitrary tie-breaking
outside \(A_0\cup A_1\). Under any \(p\in\mathcal P_i\), whenever
\(|\widehat F-F(p)|\le \varepsilon\), the test is correct. Hence the average
testing error of this test is at most the average estimation error under the
two priors, and therefore at most the displayed supremum risk.

On the other hand, the optimal average error for testing two probability laws
\(\mathsf P_0,\mathsf P_1\) under equal prior probabilities is
\[
\frac{1-\operatorname{TV}(\mathsf P_0,\mathsf P_1)}{2}.
\]
Combining the two inequalities proves the first claim. The second follows by
setting the supremum error probability equal to \(1-q\).
\end{proof}

\ignore{

\begin{lemma}[Likelihood-ratio bounds for total variation]
Let \(\mu,\nu\) be probability measures.

\begin{enumerate}
\item If \(\nu\ll\mu\) and \(L=d\nu/d\mu\), then
\[
\operatorname{TV}(\mu,\nu)
\le
\frac12\sqrt{\mathbb E_\mu[L^2]-1}.
\]

\item If \(\nu_0,\nu_1\ll\mathsf B\) with likelihood ratios
\(L_0=d\nu_0/d\mathsf B\) and \(L_1=d\nu_1/d\mathsf B\), then
\[
\operatorname{TV}(\nu_0,\nu_1)
=
\frac12\mathbb E_{\mathsf B}|L_1-L_0|.
\]
\end{enumerate}
\end{lemma}
} 

\section{Lower Bound for Estimating Min-Entropy}
\label{sec:Min_Entropy_LB}

We prove the lower bound by reducing min-entropy estimation to a
hypothesis-testing problem. The null hypothesis is the uniform
distribution on $[k]$. Under the alternative, one uniformly random
coordinate is slightly heavier than the others. Although the two
hypotheses have min-entropies separated by a constant, their induced
sample distributions remain close in total variation below the
$k\log k$ scale.

\subsection{Statement}

\begin{theorem}[Lower bound for min-entropy estimation]
\label{thm:minentropy-lower-bound}
Fix a constant $\gamma\in(0,1/2)$, and let
\[
    \delta_\gamma:=\log(1+\gamma).
\]
Let $\varepsilon\in(0,\delta_\gamma/2)$. Suppose that
$\widehat H=\widehat H(X_1,\ldots,X_n)$ satisfies
\[
    \mathbb P_{X_1,\ldots,X_n\sim p}
    \left(
        |\widehat H-H_\infty(p)|
        \le
        \varepsilon
    \right)
    \ge
    \frac23
\]
for every distribution $p$ on $[k]$, where
\[
    H_\infty(p):=-\log\max_{x\in[k]}p(x).
\]
Then there exists a constant $c_\gamma>0$, depending only on
$\gamma$, such that
\[
    n\ge c_\gamma k\log k
\]
for all sufficiently large $k$. In particular, the sample complexity
of additive-constant min-entropy estimation is
\[
    \Omega(k\log k).
\]
\end{theorem}

We establish the theorem through a sequence of lemmas. We first define
the hard family and reduce estimation to testing. We then compute the
second moment of the likelihood ratio of the alternative mixture.

\subsection{Hard priors and entropy separation}

\begin{lemma}[Hard family of alternatives]
\label{lem:hard-family}
Let $U$ be the uniform distribution on $[k]$:
\[
    U(x)=\frac1k,
    \qquad x\in[k].
\]
For each $j\in[k]$, define a distribution $Q_j$ on $[k]$ by
\[
    Q_j(j)=\frac{1+\gamma}{k},
    \qquad
    Q_j(x)
    =
    \frac{1-\gamma/(k-1)}{k}
    \quad\text{for }x\neq j.
\]
Then each $Q_j$ is a probability distribution on $[k]$, and
\[
    H_\infty(U)=\log k,
    \qquad
    H_\infty(Q_j)=\log k-\log(1+\gamma).
\]
Consequently,
\[
    H_\infty(U)-H_\infty(Q_j)
    =
    \delta_\gamma
    \qquad\text{for every }j\in[k].
\]
\end{lemma}

\begin{proof}
For every $j\in[k]$,
\[
\begin{aligned}
    Q_j(j)+\sum_{x\neq j}Q_j(x)
    &=
    \frac{1+\gamma}{k}
    +(k-1)\frac{1-\gamma/(k-1)}{k} \\
    &=
    \frac{1+\gamma+k-1-\gamma}{k} \\
    &=1.
\end{aligned}
\]
Moreover, $Q_j(x)>0$ for every $x$, since
$\gamma\in(0,1/2)$.

Because
\[
    \max_x U(x)=\frac1k,
\]
we have
\[
    H_\infty(U)
    =
    -\log(1/k)
    =
    \log k.
\]
For $Q_j$, the largest atom is the one at coordinate $j$:
\[
    \max_x Q_j(x)
    =
    Q_j(j)
    =
    \frac{1+\gamma}{k},
\]
because, for every $x\neq j$,
\[
    Q_j(x)
    =
    \frac{1-\gamma/(k-1)}{k}
    <
    \frac1k
    <
    \frac{1+\gamma}{k}.
\]
Therefore,
\[
    H_\infty(Q_j)
    =
    -\log\left(\frac{1+\gamma}{k}\right)
    =
    \log k-\log(1+\gamma).
\]
Subtracting the two entropies gives the claimed gap.
\end{proof}

Let \(\Pi_{\infty,0}\) be the point mass at \(U\). We then place the alternatives under a uniform prior on the location of
the heavy coordinate -- let
\(\Pi_{\infty,1}\) be the uniform prior on \(\{Q_1,\ldots,Q_k\}\). The induced
fixed-sample laws are
\[
\mathsf P_{\infty,0}^{(n)}=U^{\otimes n},
\qquad
\mathsf P_{\infty,1}^{(n)}=\frac1k\sum_{j=1}^k Q_j^{\otimes n}
=: \overline Q^{(n)} .
\]
By Lemma~\ref{lem:hard-family}, the two prior supports are separated in \(H_\infty\) by
\(\delta_\gamma\). Therefore, if \(0<\varepsilon<\delta_\gamma/2\) and an
estimator succeeds with probability at least \(2/3\) uniformly over \(\Delta_k\),
the Le Cam reduction in Lemma \ref{lem:Le_Cam} implies
\begin{equation}\label{eq:TVD_1_3}
\operatorname{TV}\bigl(U^{\otimes n},\overline Q^{(n)}\bigr)\ge \frac13 \, .
\end{equation}


\subsection{Indistinguishability of the min-entropy hard priors}

We will prove in Lemma \ref{lem:small-chi2} that 
there exists a constant \(c_\gamma>0\) such that, if \(n\le c_\gamma k\log k\), then
\[
\operatorname{TV}\bigl(U^{\otimes n},\overline Q^{(n)}\bigr)\to 0 \,\,\, \mbox{ as } \,\,\, k \to \infty.
\] 

Lemma \ref{lem:small-chi2} will be proved in a sequence of lemmas \ref{lem:likelihood}--\ref{lem:chi2-exact}. Towards this goal we will upper-bound the total variation distance using the second moment
of the mixture likelihood ratio.

\begin{lemma}[Likelihood ratios for the mixture]
\label{lem:likelihood}
For $x^n\in[k]^n$, let
\[
    L_j(x^n)
    :=
    \frac{dQ_j^{\otimes n}}{dU^{\otimes n}}(x^n),
    \qquad
    L(x^n)
    :=
    \frac{d\overline Q^{(n)}}{dU^{\otimes n}}(x^n).
\]
Then
\[
    L(x^n)
    =
    \frac1k\sum_{j=1}^kL_j(x^n),
\]
and
\[
    L_j(x^n)
    =
    \prod_{t=1}^n
    \frac{Q_j(x_t)}{U(x_t)}.
\]
Moreover,
\[
    \chi^2\left(
        \overline Q^{(n)}
        \middle\|
        U^{\otimes n}
    \right)
    =
    \mathbb E_{U^{\otimes n}}[L^2]-1,
\]
and
\[
    \operatorname{TV}
    \left(
        U^{\otimes n},
        \overline Q^{(n)}
    \right)
    \le
    \frac12
    \sqrt{
        \chi^2\left(
            \overline Q^{(n)}
            \middle\|
            U^{\otimes n}
        \right)
    }.
\]
\end{lemma}

\begin{proof}
The likelihood-ratio identities and the comparison between total
variation and $\chi^2$ divergence follow from standards facts in statistics, see, e.g.,
Fact~\ref{fact:tv-chi-square}.

For completeness we will present their proofs here. Since $U(x)=1/k>0$ for every $x\in[k]$, the measure
$U^{\otimes n}$ has full support on $[k]^n$. Thus, every
$Q_j^{\otimes n}$ and $\overline Q^{(n)}$ is absolutely continuous
with respect to $U^{\otimes n}$.

Since
\[
    \overline Q^{(n)}
    =
    \frac1k\sum_{j=1}^kQ_j^{\otimes n},
\]
linearity of the Radon--Nikodym derivative gives
\[
\begin{aligned}
    L
    &=
    \frac{d\overline Q^{(n)}}{dU^{\otimes n}} \\
    &=
    \frac1k
    \sum_{j=1}^k
    \frac{dQ_j^{\otimes n}}{dU^{\otimes n}} \\
    &=
    \frac1k\sum_{j=1}^kL_j.
\end{aligned}
\]
Because $Q_j^{\otimes n}$ and $U^{\otimes n}$ are product measures,
\[
    L_j(x^n)
    =
    \prod_{t=1}^n
    \frac{Q_j(x_t)}{U(x_t)}.
\]

By the definition of $\chi^2$-divergence,
\[
\begin{aligned}
    \chi^2\left(
        \overline Q^{(n)}
        \middle\|
        U^{\otimes n}
    \right)
    &=
    \int
    \left(
        \frac{d\overline Q^{(n)}}{dU^{\otimes n}}-1
    \right)^2
    dU^{\otimes n} \\
    &=
    \mathbb E_{U^{\otimes n}}[(L-1)^2].
\end{aligned}
\]
Since $\mathbb E_{U^{\otimes n}}[L]=1$, this becomes
\[
    \chi^2\left(
        \overline Q^{(n)}
        \middle\|
        U^{\otimes n}
    \right)
    =
    \mathbb E_{U^{\otimes n}}[L^2]-1.
\]

Finally,
\[
\begin{aligned}
    \operatorname{TV}
    \left(
        U^{\otimes n},
        \overline Q^{(n)}
    \right)
    &=
    \frac12\int|L-1|\,dU^{\otimes n} \\
    &=
    \frac12
    \mathbb E_{U^{\otimes n}}[|L-1|].
\end{aligned}
\]
By the Cauchy--Schwarz inequality (Fact~\ref{fact:basic-probability}),
\[
\begin{aligned}
    \mathbb E_{U^{\otimes n}}[|L-1|]
    &\le
    \left(
        \mathbb E_{U^{\otimes n}}[(L-1)^2]
    \right)^{1/2} \\
    &=
    \sqrt{
        \chi^2\left(
            \overline Q^{(n)}
            \middle\|
            U^{\otimes n}
        \right)
    }.
\end{aligned}
\]
The stated total-variation bound follows.
\end{proof}

\begin{lemma}[Second-moment expansion]
\label{lem:second-moment-expansion}
With the notation of Lemma~\ref{lem:likelihood},
\[
    \mathbb E_{U^{\otimes n}}[L^2]
    =
    \frac1{k^2}
    \sum_{i=1}^k\sum_{j=1}^k
    \mathbb E_{U^{\otimes n}}[L_iL_j],
\]
and, for every $i,j\in[k]$,
\[
    \mathbb E_{U^{\otimes n}}[L_iL_j]
    =
    \left(
        \sum_{x=1}^k
        \frac{Q_i(x)Q_j(x)}{U(x)}
    \right)^n.
\]
Since $U(x)=1/k$, equivalently,
\[
    \mathbb E_{U^{\otimes n}}[L_iL_j]
    =
    \left(
        k\sum_{x=1}^kQ_i(x)Q_j(x)
    \right)^n.
\]
\end{lemma}

\begin{proof}
By Lemma~\ref{lem:likelihood},
\[
    L=\frac1k\sum_{j=1}^kL_j.
\]
Hence, point-wise on $[k]^n$,
\[
    L^2
    =
    \left(
        \frac1k\sum_{j=1}^kL_j
    \right)^2
    =
    \frac1{k^2}
    \sum_{i=1}^k\sum_{j=1}^kL_iL_j.
\]
Taking expectations gives
\[
    \mathbb E_{U^{\otimes n}}[L^2]
    =
    \frac1{k^2}
    \sum_{i=1}^k\sum_{j=1}^k
    \mathbb E_{U^{\otimes n}}[L_iL_j].
\]

Using the product representation of the likelihood ratios,
\[
\begin{aligned}
    L_i(x^n)L_j(x^n)
    &=
    \prod_{t=1}^n
    \frac{Q_i(x_t)}{U(x_t)}
    \prod_{t=1}^n
    \frac{Q_j(x_t)}{U(x_t)} \\
    &=
    \prod_{t=1}^n
    \frac{Q_i(x_t)Q_j(x_t)}{U(x_t)^2}.
\end{aligned}
\]
Therefore,
\[
\begin{aligned}
    \mathbb E_{U^{\otimes n}}[L_iL_j]
    &=
    \sum_{x^n\in[k]^n}
    U^{\otimes n}(x^n)
    L_i(x^n)L_j(x^n) \\
    &=
    \sum_{x^n\in[k]^n}
    \prod_{t=1}^nU(x_t)
    \prod_{t=1}^n
    \frac{Q_i(x_t)Q_j(x_t)}{U(x_t)^2} \\
    &=
    \sum_{x^n\in[k]^n}
    \prod_{t=1}^n
    \frac{Q_i(x_t)Q_j(x_t)}{U(x_t)}.
\end{aligned}
\]
The summand factors across the sample coordinates, so
\[
\begin{aligned}
    \mathbb E_{U^{\otimes n}}[L_iL_j]
    &=
    \prod_{t=1}^n
    \left(
        \sum_{x=1}^k
        \frac{Q_i(x)Q_j(x)}{U(x)}
    \right) \\
    &=
    \left(
        \sum_{x=1}^k
        \frac{Q_i(x)Q_j(x)}{U(x)}
    \right)^n.
\end{aligned}
\]
Since $U(x)=1/k$, this becomes
\[
    \mathbb E_{U^{\otimes n}}[L_iL_j]
    =
    \left(
        k\sum_{x=1}^kQ_i(x)Q_j(x)
    \right)^n.
\]
\end{proof}

We next evaluate the diagonal and off-diagonal terms in this expansion.

\begin{lemma}[Diagonal and off-diagonal terms]
\label{lem:diagonal-offdiagonal}
Let
\[
    a:=\frac{1+\gamma}{k},
    \qquad
    b:=\frac{1-\gamma/(k-1)}{k}.
\]
Then, for every $i\in[k]$,
\[
    k\sum_{x=1}^kQ_i(x)^2
    =
    1+\frac{\gamma^2}{k-1},
\]
and, for every distinct $i,j\in[k]$,
\[
    k\sum_{x=1}^kQ_i(x)Q_j(x)
    =
    1-\frac{\gamma^2}{(k-1)^2}.
\]
Consequently,
\[
    \mathbb E_{U^{\otimes n}}[L_i^2]
    =
    \left(
        1+\frac{\gamma^2}{k-1}
    \right)^n,
\]
and, for $i\neq j$,
\[
    \mathbb E_{U^{\otimes n}}[L_iL_j]
    =
    \left(
        1-\frac{\gamma^2}{(k-1)^2}
    \right)^n.
\]
\end{lemma}

\begin{proof}
Each $Q_i$ has one coordinate equal to $a$ and the remaining $k-1$
coordinates equal to $b$. Hence,
\[
    \sum_{x=1}^kQ_i(x)^2
    =
    a^2+(k-1)b^2.
\]
Now,
\[
    a^2=\frac{(1+\gamma)^2}{k^2},
    \qquad
    (k-1)b^2
    =
    (k-1)
    \frac{(1-\gamma/(k-1))^2}{k^2}.
\]
Expanding,
\[
    (k-1)
    \left(
        1-\frac{\gamma}{k-1}
    \right)^2
    =
    (k-1)-2\gamma+\frac{\gamma^2}{k-1}.
\]
Therefore,
\[
\begin{aligned}
    a^2+(k-1)b^2
    &=
    \frac{
        (1+2\gamma+\gamma^2)
        +(k-1)-2\gamma+\gamma^2/(k-1)
    }{k^2} \\
    &=
    \frac{
        k+\gamma^2+\gamma^2/(k-1)
    }{k^2} \\
    &=
    \frac1k
    \left(
        1+\frac{\gamma^2}{k-1}
    \right).
\end{aligned}
\]
Thus,
\[
    k\sum_{x=1}^kQ_i(x)^2
    =
    1+\frac{\gamma^2}{k-1}.
\]

Now let $i\neq j$. At coordinate $i$,
$(Q_i(i),Q_j(i))=(a,b)$; at coordinate $j$,
$(Q_i(j),Q_j(j))=(b,a)$; and at every remaining coordinate both
probabilities equal $b$. Hence,
\[
    \sum_{x=1}^kQ_i(x)Q_j(x)
    =
    2ab+(k-2)b^2.
\]
Let
\[
    c:=1-\frac{\gamma}{k-1},
\]
so that $a=(1+\gamma)/k$ and $b=c/k$. Then
\[
    k\bigl(2ab+(k-2)b^2\bigr)
    =
    \frac{
        2(1+\gamma)c+(k-2)c^2
    }{k}.
\]
Since
\[
    c^2
    =
    1-\frac{2\gamma}{k-1}
    +\frac{\gamma^2}{(k-1)^2},
\]
we have
\[
\begin{aligned}
    2(1+\gamma)c+(k-2)c^2
    &=
    2(1+\gamma)
    \left(
        1-\frac{\gamma}{k-1}
    \right) \\
    &\quad+
    (k-2)
    \left(
        1-\frac{2\gamma}{k-1}
        +\frac{\gamma^2}{(k-1)^2}
    \right) \\
    &=
    2+2\gamma
    -\frac{2\gamma(1+\gamma)}{k-1}
    +k-2
    -\frac{2\gamma(k-2)}{k-1}
    +\frac{(k-2)\gamma^2}{(k-1)^2} \\
    &=
    k
    +
    \left(
        2\gamma
        -\frac{2\gamma(k-1+\gamma)}{k-1}
    \right)
    +
    \frac{(k-2)\gamma^2}{(k-1)^2} \\
    &=
    k
    -\frac{2\gamma^2}{k-1}
    +\frac{(k-2)\gamma^2}{(k-1)^2}.
\end{aligned}
\]
Combining the last two terms gives
\[
    -\frac{2\gamma^2}{k-1}
    +\frac{(k-2)\gamma^2}{(k-1)^2}
    =
    -\frac{k\gamma^2}{(k-1)^2}.
\]
Therefore,
\[
    2(1+\gamma)c+(k-2)c^2
    =
    k-\frac{k\gamma^2}{(k-1)^2}.
\]
Dividing by $k$ yields
\[
    k\sum_{x=1}^kQ_i(x)Q_j(x)
    =
    1-\frac{\gamma^2}{(k-1)^2}.
\]

The formulas for
$\mathbb E_{U^{\otimes n}}[L_i^2]$ and
$\mathbb E_{U^{\otimes n}}[L_iL_j]$ now follow from
Lemma~\ref{lem:second-moment-expansion}.
\end{proof}

The preceding calculation gives an exact expression for the
$\chi^2$-divergence.

\begin{lemma}[Exact $\chi^2$ formula]
\label{lem:chi2-exact}
We have
\[
    \mathbb E_{U^{\otimes n}}[L^2]
    =
    \frac1k
    \left(
        1+\frac{\gamma^2}{k-1}
    \right)^n
    +
    \frac{k-1}{k}
    \left(
        1-\frac{\gamma^2}{(k-1)^2}
    \right)^n.
\]
Consequently,
\[
\begin{aligned}
    \chi^2\left(
        \overline Q^{(n)}
        \middle\|
        U^{\otimes n}
    \right)
    &=
    \frac1k
    \left(
        1+\frac{\gamma^2}{k-1}
    \right)^n +
    \frac{k-1}{k}
    \left(
        1-\frac{\gamma^2}{(k-1)^2}
    \right)^n
    -1.
\end{aligned}
\]
\end{lemma}

\begin{proof}
By Lemma~\ref{lem:second-moment-expansion},
\[
    \mathbb E_{U^{\otimes n}}[L^2]
    =
    \frac1{k^2}
    \sum_{i=1}^k\sum_{j=1}^k
    \mathbb E_{U^{\otimes n}}[L_iL_j].
\]
There are $k$ diagonal terms and $k(k-1)$ off-diagonal terms.
Therefore, Lemma~\ref{lem:diagonal-offdiagonal} gives
\[
\begin{aligned}
    \mathbb E_{U^{\otimes n}}[L^2]
    &=
    \frac{k}{k^2}
    \left(
        1+\frac{\gamma^2}{k-1}
    \right)^n +
    \frac{k(k-1)}{k^2}
    \left(
        1-\frac{\gamma^2}{(k-1)^2}
    \right)^n \\&=
    \frac1k
    \left(
        1+\frac{\gamma^2}{k-1}
    \right)^n
    +
    \frac{k-1}{k}
    \left(
        1-\frac{\gamma^2}{(k-1)^2}
    \right)^n.
\end{aligned}
\]
The claimed $\chi^2$ formula follows from
Lemma~\ref{lem:likelihood}.
\end{proof}

We now show that this divergence vanishes below a sufficiently small
constant multiple of \(k\log k\).

\begin{lemma}[Small $\chi^2$ below the $k\log k$ scale]
\label{lem:small-chi2}
There exists a constant $c_\gamma>0$, depending only on $\gamma$, such
that, if
\[
    n\le c_\gamma k\log k,
\]
then
\[
    \chi^2\left(
        \overline Q^{(n)}
        \middle\|
        U^{\otimes n}
    \right)
    \longrightarrow0
    \qquad\text{as }k\to\infty.
\]
Consequently,
\[
    \operatorname{TV}
    \left(
        U^{\otimes n},
        \overline Q^{(n)}
    \right)
    \longrightarrow0.
\]
\end{lemma}

\begin{proof}
By Lemma~\ref{lem:chi2-exact},
\[
\begin{aligned}
    \chi^2\left(
        \overline Q^{(n)}
        \middle\|
        U^{\otimes n}
    \right)
    &=
    \frac1k
    \left(
        1+\frac{\gamma^2}{k-1}
    \right)^n +
    \frac{k-1}{k}
    \left(
        1-\frac{\gamma^2}{(k-1)^2}
    \right)^n
    -1.
\end{aligned}
\]
Write
\[
    A_k
    :=
    \frac1k
    \left(
        1+\frac{\gamma^2}{k-1}
    \right)^n
\]
and
\[
    B_k
    :=
    \frac{k-1}{k}
    \left(
        1-\frac{\gamma^2}{(k-1)^2}
    \right)^n.
\]
Then
\[
    \chi^2\left(
        \overline Q^{(n)}
        \middle\|
        U^{\otimes n}
    \right)
    =
    A_k+B_k-1.
\]

We first show that $A_k\to0$ for a suitable choice of $c_\gamma$.
Using $1+t\le e^t$ for $t\ge0$ (see Fact~\ref{fact:elementary-log-exp}), we obtain
\[
    A_k
    \le
    \frac1k
    \exp\left(
        \frac{n\gamma^2}{k-1}
    \right).
\]
If $n\le c_\gamma k\log k$, then, for all sufficiently large $k$,
\[
    \frac{n\gamma^2}{k-1}
    \le
    c_\gamma\gamma^2
    \frac{k}{k-1}
    \log k
    \le
    2c_\gamma\gamma^2\log k.
\]
Hence,
\[
\begin{aligned}
    A_k
    &\le
    \frac1k
    \exp\left(
        2c_\gamma\gamma^2\log k
    \right) \\
    &=
    k^{-1+2c_\gamma\gamma^2}.
\end{aligned}
\]
Choose $c_\gamma>0$ so that
\[
    2c_\gamma\gamma^2<1.
\]
For example, one may take
\[
    c_\gamma=\frac{1}{4\gamma^2}.
\]
Then
\[
    A_k
    \le
    k^{-1+2c_\gamma\gamma^2}
    \longrightarrow0.
\]

Next, since
\[
    0
    \le
    \frac{\gamma^2}{(k-1)^2}
    \le
    1
\]
for all sufficiently large $k$,
\[
    0
    \le
    \left(
        1-\frac{\gamma^2}{(k-1)^2}
    \right)^n
    \le
    1.
\]
Therefore,
\[
    B_k\le\frac{k-1}{k}.
\]
It follows that
\[
\begin{aligned}
    A_k+B_k-1
    &\le
    A_k+\frac{k-1}{k}-1 \\
    &=
    A_k-\frac1k \\
    &\le
    A_k.
\end{aligned}
\]
On the other hand, by definition of $\chi^2$-divergence,
\[
    \chi^2\left(
        \overline Q^{(n)}
        \middle\|
        U^{\otimes n}
    \right)
    \ge0.
\]
Thus,
\[
    0
    \le
    \chi^2\left(
        \overline Q^{(n)}
        \middle\|
        U^{\otimes n}
    \right)
    =
    A_k+B_k-1
    \le
    A_k.
\]
Since $A_k\to0$, the squeeze theorem yields
\[
    \chi^2\left(
        \overline Q^{(n)}
        \middle\|
        U^{\otimes n}
    \right)
    \longrightarrow0.
\]

Finally, Lemma~\ref{lem:likelihood} gives
\[
    \operatorname{TV}
    \left(
        U^{\otimes n},
        \overline Q^{(n)}
    \right)
    \le
    \frac12
    \sqrt{
        \chi^2\left(
            \overline Q^{(n)}
            \middle\|
            U^{\otimes n}
        \right)
    },
\]
and hence
\[
    \operatorname{TV}
    \left(
        U^{\otimes n},
        \overline Q^{(n)}
    \right)
    \longrightarrow0.
\]
\end{proof}

\subsection{Final proof of Theorem \ref{thm:minentropy-lower-bound}}

\begin{proof}[Proof of Theorem~\ref{thm:minentropy-lower-bound}]
Suppose, toward a contradiction, that an estimator $\widehat H$, with \(0<\varepsilon<\delta_\gamma/2\),
satisfies
\[
    \mathbb P_{p^{\otimes n}}
    \left(
        |\widehat H-H_\infty(p)|
        \le
        \varepsilon
    \right)
    \ge
    \frac23
\]
for every distribution $p \in \Delta_k$, while
\[
    n\le c_\gamma k\log k
\]
for infinitely many $k$, where $c_\gamma$ is the constant from
Lemma~\ref{lem:small-chi2}.

We have then shown in (\ref{eq:TVD_1_3}) by Le Cam's reduction that
\[
    \operatorname{TV}
    \left(
        U^{\otimes n},
        \overline Q^{(n)}
    \right)
    \ge
    \frac13.
\]
On the other hand, Lemma~\ref{lem:small-chi2} gives
\[
    \operatorname{TV}
    \left(
        U^{\otimes n},
        \overline Q^{(n)}
    \right)
    \longrightarrow0
    \qquad\text{as }k\to\infty,
\]
which is a contradiction. Therefore, for all sufficiently large $k$, no such estimator can exist
when $n\le c_\gamma k\log k$.
\end{proof}

\section{Lower Bound for Estimating R\'enyi Entropy}
\label{sec:Renyi_Entropy_LB}

In this section, we prove a fixed-accuracy minimax lower bound for
estimating the R\'enyi entropy of order $\alpha$ over distributions in
$\Delta_k$. The proof uses a hidden-heavy-coordinate prior,
Poissonization, likelihood ratios with respect to a product-Poisson
baseline, truncation of all occupancies above $\lfloor\alpha\rfloor$,
and a truncated second-moment calculation. For every fixed accuracy
$0<\epsilon<1$ and constant $c_0>0$, the argument gives
\[
   n
   \ge
   c_1\,\alpha k^{1-1/\alpha},
\]
where $c_1=c_1(c_0)>0$, uniformly over all real orders
\[
   1.001\le\alpha\le c_0\log k,
\]
after adjusting the absolute constants. Equivalently, since $\epsilon$
is fixed, this may be written as
\[
   n
   \ge
   C_\epsilon\,
   \frac{\alpha k^{1-1/\alpha}}{\epsilon^2},
\]
where $C_\epsilon>0$ does not depend on $\alpha$ or $k$. The constants
are not optimized.

\subsection{Statement}

For $\alpha>1$, recall that the R\'enyi entropy of a distribution
$p=(p_1,\ldots,p_k)\in\Delta_k$ and its $\alpha$-power sum are
\[
   H_\alpha(p)
   =
   \frac{1}{1-\alpha}\log P_\alpha(p),
   \qquad
   P_\alpha(p)
   =
   \sum_{i=1}^k p_i^\alpha.
\]
All logarithms are natural.

\begin{theorem}[Fixed-accuracy lower bound]
\label{thm:main_2}
Fix any accuracy $0<\epsilon<1$ and any constant $c_0>0$. There exist
constants
\[
   c_1=c_1(c_0)>0,
   \qquad
   K_0=K_0(c_0)<\infty,
\]
such that, for every $k\ge K_0$, every real order
\[
   1.001\le\alpha\le c_0\log k,
\]
and every estimator $\widehat H$ satisfying
\[
   \sup_{p\in\Delta_k}
   \Prob_p\!\left(
      |\widehat H-H_\alpha(p)|>\epsilon
   \right)
   \le
   \frac13,
\]
the number of samples $n$ must satisfy
\[
   n
   \ge
   c_1\,\alpha k^{1-1/\alpha}.
\]
Equivalently, because $\epsilon\in(0,1)$ and $c_0>0$ are fixed, there
is a constant
\[
   C_\epsilon=c_1\epsilon^2>0,
\]
independent of $\alpha$ and $k$, such that
\[
   n
   \ge
   C_\epsilon
   \frac{\alpha k^{1-1/\alpha}}{\epsilon^2}.
\]
\end{theorem}

\begin{remark}
The constants below are not optimized. Once $c_0$ is fixed, one can
choose $c_1(c_0)>0$ uniformly over
$1.001\le\alpha\le c_0\log k$. The construction does not depend on
$\epsilon$ except through the final testing separation, and the choice
$B=e^4$ gives an entropy gap larger than $2\epsilon$ for every
$0<\epsilon<1$.
\end{remark}

\subsection{Hard priors}

For $\alpha\ge1.001$, set
\[
   a
   =
   k^{-1+1/\alpha}
   =
   k^{-(\alpha-1)/\alpha},
\]
and fix
\[
   B=e^4.
\]
For $\theta\in\{1,B\}$ and $j\in[k]$, define a distribution
$p^{(\theta,j)}\in\Delta_k$ by
\[
   p^{(\theta,j)}_j
   =
   \theta a,
   \qquad
   p^{(\theta,j)}_i
   =
   \frac{1-\theta a}{k-1},
   \quad i\ne j.
\]
For all sufficiently large $k$, these are valid distributions uniformly
over $1.001\le\alpha\le c_0\log k$. Indeed,
\[
   a
   =
   k^{-(\alpha-1)/\alpha}
   \le
   k^{-1/1001},
\]
and hence $Ba<1$ for all sufficiently large $k$.

Define two priors on $\Delta_k$ by
\[
   \Pi_0:
   \qquad
   J\sim\Unif([k]),
   \qquad
   p=p^{(1,J)},
\]
and
\[
   \Pi_1:
   \qquad
   J\sim\Unif([k]),
   \qquad
   p=p^{(B,J)}.
\]
The index $J$ is the hidden heavy coordinate.

\subsection{Entropy separation}

We first show that all distributions in $\Pi_0$ have one R\'enyi
entropy value, all distributions in $\Pi_1$ have another, and the two
values differ by more than $2\epsilon$.

\begin{lemma}[Uniform elementary estimates]
\label{lem:elementary-estimates}
For every fixed $c_0>0$, uniformly over all real orders
$1.001\le\alpha\le c_0\log k$,
\[
   a
   =
   k^{-1+1/\alpha}
   \le
   k^{-1/1001},
   \qquad
   \alpha a\to0,
   \qquad
   \frac{\alpha}{k}\to0
\]
as $k\to\infty$. Consequently, for each fixed
$\theta\in\{1,B\}$,
\[
   \left(
      \frac{k}{k-1}
   \right)^{\alpha-1}
   (1-\theta a)^\alpha
   =
   1+o(1)
\]
uniformly over $1.001\le\alpha\le c_0\log k$.
\end{lemma}

\begin{proof}
Because $\alpha\ge1.001=1001/1000$,
\[
   \frac{\alpha-1}{\alpha}
   \ge
   \frac{1}{1001}.
\]
Therefore,
\[
   a
   =
   k^{-(\alpha-1)/\alpha}
   \le
   k^{-1/1001}.
\]
Since $\alpha\le c_0\log k$,
\[
   \alpha a
   \le
   c_0(\log k)k^{-1/1001}
   \to0,
   \qquad
   \frac{\alpha}{k}
   \le
   c_0\frac{\log k}{k}
   \to0.
\]
Moreover,
\[
\begin{aligned}
   \log
   \left(
      \frac{k}{k-1}
   \right)^{\alpha-1}
   &=
   (\alpha-1)
   \log\left(
      1+\frac{1}{k-1}
   \right) \\
   &=
   O\left(
      \frac{\alpha}{k}
   \right)
   =
   o(1).
\end{aligned}
\]
Since $\theta a=o(1)$,
\[
\begin{aligned}
   \log(1-\theta a)^\alpha
   &=
   \alpha\log(1-\theta a) \\
   &=
   -\theta\alpha a
   +O(\alpha a^2) \\
   &=
   o(1).
\end{aligned}
\]
Exponentiating proves the claim.
\end{proof}

\begin{lemma}[Entropy separation]
\label{lem:entropy-separation}
For all sufficiently large $k$, uniformly over
$1.001\le\alpha\le c_0\log k$ and all $j\in[k]$,
\[
   \left|
      H_\alpha(p^{(B,j)})
      -
      H_\alpha(p^{(1,j)})
   \right|
   >
   2\epsilon.
\]
\end{lemma}

\begin{proof}
For $\theta\in\{1,B\}$,
\begin{align*}
   P_\alpha(p^{(\theta,j)})
   &=
   (\theta a)^\alpha
   +(k-1)
   \left(
      \frac{1-\theta a}{k-1}
   \right)^\alpha \\
   &=
   \theta^\alpha k^{1-\alpha}
   +(k-1)^{1-\alpha}(1-\theta a)^\alpha \\
   &=
   k^{1-\alpha}
   \left[
      \theta^\alpha
      +
      \left(
         \frac{k}{k-1}
      \right)^{\alpha-1}
      (1-\theta a)^\alpha
   \right].
\end{align*}
By Lemma~\ref{lem:elementary-estimates}, for all sufficiently large
$k$,
\[
   \frac12
   \le
   \left(
      \frac{k}{k-1}
   \right)^{\alpha-1}
   (1-\theta a)^\alpha
   \le
   2
\]
uniformly over the allowed range of $\alpha$. Hence,
\[
   \frac{
      P_\alpha(p^{(B,j)})
   }{
      P_\alpha(p^{(1,j)})
   }
   \ge
   \frac{B^\alpha}{3}.
\]
Since $B=e^4$,
\[
   \log
   \frac{
      P_\alpha(p^{(B,j)})
   }{
      P_\alpha(p^{(1,j)})
   }
   \ge
   4\alpha-\log3.
\]
Therefore,
\[
\begin{aligned}
   \left|
      H_\alpha(p^{(B,j)})
      -
      H_\alpha(p^{(1,j)})
   \right|
   &=
   \frac{1}{\alpha-1}
   \left|
      \log
      \frac{
         P_\alpha(p^{(B,j)})
      }{
         P_\alpha(p^{(1,j)})
      }
   \right| \\
   &\ge
   \frac{4\alpha-\log3}{\alpha-1} \\
   &>
   4.
\end{aligned}
\]
Because $0<\epsilon<1$, the entropy gap is larger than $2\epsilon$.
\end{proof}

\subsection{Poissonization, ideal experiment and total variation distance representation}

Consider the Poissonized sampling experiment: first draw
$N\sim\Poi(n)$ and then draw $N$ i.i.d.\ samples from the unknown
distribution. Conditional on $p$, the count vector
$(N_1,\ldots,N_k)$ has independent coordinates (see Fact~\ref{fact:poissonization})
\[
   N_i\sim\Poi(np_i).
\]

It is enough to work with count vectors rather than ordered sample
sequences. Indeed, conditional on the count vector, the ordered sequence
is uniformly distributed over all sequences having those counts, and
this conditional distribution does not depend on $p$. Therefore, the
total variation distance between the induced laws on ordered
Poissonized samples equals the total variation distance between the
induced laws on count vectors.

Let $M_\theta$ denote the law of the count vector under the prior
$\Pi_\theta$, where $\theta\in\{1,B\}$. Put
\[
   \tau
   :=
   \frac{n}{k^{1-1/\alpha}},
   \qquad
   \lambda
   :=
   \frac{n}{k}
   =
   \tau k^{-1/\alpha}.
\]
Conditional on $J=j$, under $M_\theta$ the heavy coordinate has mean
\[
   n\theta a
   =
   \theta\tau,
\]
while every light coordinate has mean
\[
   \mu_\theta
   :=
   n\frac{1-\theta a}{k-1}.
\]

It is convenient to compare $M_\theta$ with an idealized mixture law.
Define $I_\theta$ as follows: draw $J\sim\Unif([k])$ and, conditional
on $J=j$, draw independent counts with
\[
   N_j\sim\Poi(\theta\tau),
   \qquad
   N_i\sim\Poi(\lambda),
   \quad i\ne j.
\]
Thus, $I_1$ is the ideal version of the null prior and $I_B$ is the
ideal version of the alternative prior.

Let $\bbB$ be the product-Poisson baseline law under which
\[
   N_1,\ldots,N_k
   \overset{\mathrm{ind}}{\sim}
   \Poi(\lambda).
\]
For $h>0$, define the one-coordinate likelihood ratio
\[
   R_h(r)
   :=
   \frac{
      \Prob(\Poi(h)=r)
   }{
      \Prob(\Poi(\lambda)=r)
   }
   =
   e^{-(h-\lambda)}
   \left(
      \frac{h}{\lambda}
   \right)^r,
   \qquad
   r=0,1,2,\ldots.
\]
Then the likelihood ratio of $I_\theta$ with respect to $\bbB$ is
\[
   L_\theta
   =
   L_\theta(N_1,\ldots,N_k)
   =
   \frac{dI_\theta}{d\bbB}
   =
   \frac1k
   \sum_{j=1}^k
   R_{\theta\tau}(N_j).
\]
Indeed, conditional on $J=j$, only coordinate $j$ changes from
$\Poi(\lambda)$ to $\Poi(\theta\tau)$; averaging over the uniformly
random hidden coordinate gives the displayed expression. Therefore, using the density representation of total variation
distance (Scheffé's theorem), see
Tsybakov~\cite[Chapter~2, Section~2.4]{Tsybakov2009},
\begin{equation}
\label{eq:TVD_Like}
   \TV(I_B,I_1)
   =
   \frac12
   \E_{\bbB}|L_B-L_1|.
\end{equation}

The remainder of the proof bounds this expectation when $\tau$ is at
most a sufficiently small constant multiple of $\alpha$.

\subsection{Choice of constants and truncation}

Recall that $c_0>0$ is fixed and $B=e^4$. Set
\[
   C_*
   =
   C_*(c_0)
   :=
   \min\left\{
      1,
      \frac{1}{8c_0(1+B^2)},
      \frac{e^{-40}}{eB},
      \frac{e^{-40}}{2eB^2}
   \right\}.
\]
The penultimate term is redundant, but we retain it because the proof
explicitly uses both
\[
   C_*
   \le
   \frac{e^{-40}}{eB}
   \qquad\text{and}\qquad
   C_*
   \le
   \frac{e^{-40}}{2eB^2}.
\]

We shall prove that if
\[
   \tau
   =
   \frac{n}{k^{1-1/\alpha}}
   \le
   C_*\alpha,
\]
then
\[
   \TV(M_B,M_1)
   \le
   \frac14
\]
for all sufficiently large $k$, uniformly over
\[
   1.001\le\alpha\le c_0\log k.
\]

Throughout the indistinguishability argument, assume
\[
   \tau\le C_*\alpha
\]
and let
\[
   m:=\lfloor\alpha\rfloor.
\]
Then
\[
   m\ge1,
   \qquad
   m\le\alpha<m+1.
\]

\begin{lemma}[Poisson tail bound]
\label{lem:poisson-tail}
If $Z\sim\Poi(u)$, then, for every integer $s\ge1$,
\[
   \Prob(Z\ge s)
   \le
   \left(
      \frac{eu}{s}
   \right)^s.
\]
\end{lemma}

\begin{proof}
For $t>0$, Markov's inequality  (Fact~\ref{fact:basic-probability}) gives
\[
\begin{aligned}
   \Prob(Z\ge s)
   &=
   \Prob(e^{tZ}\ge e^{ts}) \\
   &\le
   e^{-ts}\E e^{tZ} \\
   &=
   \exp\{u(e^t-1)-ts\}.
\end{aligned}
\]
If $u<s$, choose $e^t=s/u$ to obtain
\[
   \Prob(Z\ge s)
   \le
   \exp\{s-u-s\log(s/u)\}
   \le
   \left(
      \frac{eu}{s}
   \right)^s.
\]
If $u\ge s$, the displayed upper bound is at least $e^s>1$, so the
inequality is immediate.
\end{proof}

Define the truncation event
\[
   \mathcal E
   :=
   \left\{
      \max_{1\le i\le k}N_i\le m
   \right\}.
\]

\begin{lemma}[The truncation event is likely]
\label{lem:truncation-likely}
For each $\theta\in\{1,B\}$, if $\tau\le C_*\alpha$, then
\[
   I_\theta(\mathcal E^c)
   \le
   2e^{-40(m+1)}
   \le
   \frac{1}{1000}.
\]
\end{lemma}

\begin{proof}
Under $I_\theta$, conditional on the hidden coordinate, one coordinate
has mean
\[
   \theta\tau
   \le
   BC_*\alpha,
\]
and the remaining $k-1$ coordinates have mean
\[
   \lambda
   =
   \tau k^{-1/\alpha}
   \le
   C_*\alpha k^{-1/\alpha}.
\]

For the heavy coordinate, Lemma~\ref{lem:poisson-tail} and
$\alpha<m+1$ give
\[
\begin{aligned}
   \Prob(\Poi(\theta\tau)>m)
   &=
   \Prob(\Poi(\theta\tau)\ge m+1) \\
   &\le
   \left(
      \frac{eBC_*\alpha}{m+1}
   \right)^{m+1} \\
   &\le
   (eBC_*)^{m+1} \\
   &\le
   e^{-40(m+1)}.
\end{aligned}
\]

For a light coordinate,
\[
\begin{aligned}
   \Prob(\Poi(\lambda)>m)
   &\le
   \left(
      \frac{
         eC_*\alpha k^{-1/\alpha}
      }{
         m+1
      }
   \right)^{m+1} \\
   &\le
   (eC_*)^{m+1}
   k^{-(m+1)/\alpha}.
\end{aligned}
\]
A union bound over the light coordinates gives
\[
\begin{aligned}
   k\Prob(\Poi(\lambda)>m)
   &\le
   (eC_*)^{m+1}
   k^{1-(m+1)/\alpha} \\
   &\le
   (eC_*)^{m+1} \\
   &\le
   e^{-40(m+1)},
\end{aligned}
\]
because $m+1>\alpha$ and $eC_*\le e^{-40}$. Combining the heavy and
light bounds yields
\[
   I_\theta(\mathcal E^c)
   \le
   2e^{-40(m+1)}.
\]
Since $m+1\ge2$, the final quantity is smaller than $1/1000$.
\end{proof}

\subsection{Truncated likelihood-ratio second moment}

Define
\[
   \Delta(r)
   :=
   R_{B\tau}(r)-R_\tau(r),
\]
and let
\[
   Z
   :=
   \Delta(N)\one_{\{N\le m\}},
   \qquad
   N\sim\Poi(\lambda)
\]
under the baseline law $\bbB$. Let $Z_1,\ldots,Z_k$ be independent
copies of $Z$. On the event $\mathcal E$,
\[
   L_B-L_1
   =
   \frac1k\sum_{j=1}^k\Delta(N_j)
   =
   \frac1k\sum_{j=1}^k Z_j.
\]

\begin{lemma}[Truncated second moment]
\label{lem:second-moment}
If $\tau\le C_*\alpha$, then, for all sufficiently large $k$, uniformly
over
\[
   1.001\le\alpha\le c_0\log k,
\]
we have
\[
   \E_{\bbB}
   \left[
      \left(
         \frac1k\sum_{j=1}^k Z_j
      \right)^2
   \right]
   \le
   10^{-6}.
\]
\end{lemma}

\begin{proof}
First,
\begin{align*}
   \E Z
   &=
   \sum_{r=0}^m
   \left[
      \Prob(\Poi(B\tau)=r)
      -
      \Prob(\Poi(\tau)=r)
   \right] \\
   &=
   \Prob(\Poi(\tau)>m)
   -
   \Prob(\Poi(B\tau)>m).
\end{align*}
By the same tail estimate used in
Lemma~\ref{lem:truncation-likely},
\[
   |\E Z|
   \le
   2e^{-40(m+1)}.
\]

It remains to bound $\E Z^2/k$. Using
$(x-y)^2\le2x^2+2y^2$ (see Fact~\ref{fact:elementary-algebra}), we obtain
\[
   \E Z^2
   \le
   2S(B\tau)+2S(\tau),
\]
where, for $0<h\le B\tau$,
\[
   S(h)
   :=
   \sum_{r=0}^m
   \frac{
      \Prob(\Poi(h)=r)^2
   }{
      \Prob(\Poi(\lambda)=r)
   }.
\]
For a fixed $r$,
\[
   \frac{
      \Prob(\Poi(h)=r)^2
   }{
      \Prob(\Poi(\lambda)=r)
   }
   =
   e^{-2h+\lambda}
   \frac{(h^2/\lambda)^r}{r!}.
\]

For $r=0$,
\[
\begin{aligned}
   \frac1k
   \frac{
      \Prob(\Poi(h)=0)^2
   }{
      \Prob(\Poi(\lambda)=0)
   }
   &=
   \frac1k e^{-2h+\lambda} \\
   &\le
   \frac1k e^\lambda \\
   &\le
   \exp\{C_*\alpha-\log k\} \\
   &\le
   k^{-7/8},
\end{aligned}
\]
because $\alpha\le c_0\log k$ and $C_*c_0\le1/8$.

For $1\le r\le m$, using
$r!\ge(r/e)^r$ (Fact~\ref{fact:factorial-optimization}), $h\le B\tau$,
$\lambda=\tau k^{-1/\alpha}$, and
$\tau\le C_*\alpha$, we obtain
\begin{align*}
   \frac1k
   \frac{
      \Prob(\Poi(h)=r)^2
   }{
      \Prob(\Poi(\lambda)=r)
   }
   &=
   \frac1k
   e^{-2h+\lambda}
   \frac{(h^2/\lambda)^r}{r!} \\
   &\le
   \frac1k e^\lambda
   \left(
      \frac{eB^2\tau}{r}
   \right)^r
   k^{r/\alpha} \\
   &\le
   \exp\left\{
      \lambda
      +
      r\log\left(
         \frac{eB^2C_*\alpha}{r}
      \right)
      +
      \left(
         \frac{r}{\alpha}-1
      \right)\log k
   \right\}.
\end{align*}

Consider first $1\le r\le\alpha/2$. For any $A,x>0$,
\[
   x\log(A/x)\le A/e.
\]
Applying this inequality with
$A=eB^2C_*\alpha$ and $x=r$ gives
\[
   r\log\left(
      \frac{eB^2C_*\alpha}{r}
   \right)
   \le
   B^2C_*\alpha.
\]
Moreover,
\[
   \frac{r}{\alpha}-1
   \le
   -\frac12,
   \qquad
   \lambda\le C_*\alpha.
\]
Therefore,
\[
   \frac1k
   \frac{
      \Prob(\Poi(h)=r)^2
   }{
      \Prob(\Poi(\lambda)=r)
   }
   \le
   \exp\left\{
      (1+B^2)C_*\alpha
      -
      \frac12\log k
   \right\}
   \le
   k^{-3/8},
\]
because
\[
   (1+B^2)C_*c_0
   \le
   \frac18
\]
and $\alpha\le c_0\log k$.

Now consider $\alpha/2<r\le m$. Put
\[
   \rho:=\frac{r}{\alpha}\in(1/2,1],
   \qquad
   D:=eB^2C_*.
\]
Because
\[
   D\le\frac{e^{-40}}{2},
\]
the function
\[
   \rho\mapsto\rho\log(D/\rho)
\]
is decreasing on $[1/2,1]$, and hence
\[
   \rho\log(D/\rho)
   \le
   \frac12\log(2D)
   \le
   -20.
\]
Since
\[
   \left(
      \frac{r}{\alpha}-1
   \right)\log k
   \le0,
   \qquad
   \lambda\le C_*\alpha,
   \qquad
   C_*\le1,
\]
the preceding estimate gives
\[
   \frac1k
   \frac{
      \Prob(\Poi(h)=r)^2
   }{
      \Prob(\Poi(\lambda)=r)
   }
   \le
   e^{-19\alpha}.
\]

Combining the three ranges yields
\[
   \frac{S(h)}{k}
   \le
   k^{-7/8}
   +
   c_0(\log k)k^{-3/8}
   +
   \alpha e^{-19\alpha}.
\]
The function $x\mapsto xe^{-19x}$ is decreasing for
$x\ge1.001$, and therefore
\[
   \alpha e^{-19\alpha}
   \le
   1.001e^{-19.019}
   <
   10^{-8}.
\]
The first two terms converge to zero uniformly over
$1.001\le\alpha\le c_0\log k$. Hence, by increasing $K_0(c_0)$ if
necessary,
\[
   \sup_{0<h\le B\tau}
   \frac{S(h)}{k}
   \le
   2\cdot10^{-8}
\]
for all $k\ge K_0(c_0)$ and all allowed $\alpha$.

Consequently,
\[
   \frac{\E Z^2}{k}
   \le
   4
   \sup_{0<h\le B\tau}
   \frac{S(h)}{k}
   \le
   8\cdot10^{-8}.
\]
Finally, since $Z_1,\ldots,Z_k$ are i.i.d.\ copies of $Z$,
\[
\begin{aligned}
   \E_{\bbB}
   \left[
      \left(
         \frac1k\sum_{j=1}^kZ_j
      \right)^2
   \right]
   &=
   \frac1k\Var(Z)+(\E Z)^2 \\
   &\le
   \frac{\E Z^2}{k}+(\E Z)^2 \\
   &\le
   10^{-6}
\end{aligned}
\]
for all sufficiently large $k$.
\end{proof}

\begin{lemma}[The ideal experiments are close]
\label{lem:ideal-close}
If $\tau\le C_*\alpha$, then, for all sufficiently large $k$, uniformly
over
\[
   1.001\le\alpha\le c_0\log k,
\]
we have
\[
   \TV(I_B,I_1)
   \le
   \frac1{50}.
\]
\end{lemma}

\begin{proof}
Using the likelihood-ratio representation (\ref{eq:TVD_Like}) of total variation distance,
\[
   \TV(I_B,I_1)
   =
   \frac12
   \E_{\bbB}|L_B-L_1|.
\]
Split according to $\mathcal E$:
\begin{align*}
   \E_{\bbB}|L_B-L_1|
   &\le
   \E_{\bbB}
   \left[
      \one_{\mathcal E}|L_B-L_1|
   \right] \\
   &\quad+
   \E_{\bbB}
   \left[
      \one_{\mathcal E^c}(L_B+L_1)
   \right] \\
   &=
   \E_{\bbB}
   \left[
      \one_{\mathcal E}|L_B-L_1|
   \right]
   +
   I_B(\mathcal E^c)
   +
   I_1(\mathcal E^c).
\end{align*}
By Lemma~\ref{lem:second-moment} and the Cauchy--Schwarz inequality (Fact~\ref{fact:basic-probability}),
\[
\begin{aligned}
   \E_{\bbB}
   \left[
      \one_{\mathcal E}|L_B-L_1|
   \right]
   &=
   \E_{\bbB}
   \left[
      \one_{\mathcal E}
      \left|
         \frac1k\sum_{j=1}^kZ_j
      \right|
   \right] \\
   &\le
   \left(
      \E_{\bbB}
      \left[
         \left(
            \frac1k\sum_{j=1}^kZ_j
         \right)^2
      \right]
   \right)^{1/2} \\
   &\le
   10^{-3}.
\end{aligned}
\]
By Lemma~\ref{lem:truncation-likely},
\[
   I_B(\mathcal E^c)+I_1(\mathcal E^c)
   \le
   \frac{2}{1000}.
\]
Thus,
\[
   \E_{\bbB}|L_B-L_1|
   \le
   \frac{3}{1000},
\]
and therefore
\[
   \TV(I_B,I_1)
   \le
   \frac{3}{2000}
   <
   \frac1{50}.
\]
\end{proof}

\subsection{From the ideal experiment back to the true mixture}

We now show that replacing the true light-coordinate mean by
$\lambda=n/k$ changes the mixture law by a negligible amount in total
variation.

We use squared Hellinger distance under the convention
\[
   H^2(P,Q)
   :=
   \sum_x
   \left(
      \sqrt{P(x)}-\sqrt{Q(x)}
   \right)^2.
\]
Under this convention,
\[
   \TV(P,Q)\le H(P,Q);
\]
see, for example, \cite{Tsybakov2009}.

\begin{lemma}[The true and ideal mixtures are close]
\label{lem:true-ideal-close}
For each $\theta\in\{1,B\}$,
\[
   \TV(M_\theta,I_\theta)
   =
   o(1)
\]
as $k\to\infty$, uniformly over
\[
   1.001\le\alpha\le c_0\log k
   \qquad\text{and}\qquad
   \tau\le C_*\alpha.
\]
In particular, for all sufficiently large $k$,
\[
   \TV(M_\theta,I_\theta)
   \le
   \frac1{100}.
\]
\end{lemma}

\begin{proof}
Condition on $J=j$. The true and ideal conditional laws have the same
heavy-coordinate mean $\theta\tau$. They differ only on the $k-1$
light coordinates. The true light-coordinate mean is
\[
   \mu_\theta
   =
   n\frac{1-\theta a}{k-1},
\]
whereas the ideal light-coordinate mean is
\[
   \lambda
   =
   \frac{n}{k}.
\]
We have
\begin{align*}
   \mu_\theta-\lambda
   &=
   n
   \left(
      \frac{1-\theta a}{k-1}
      -
      \frac1k
   \right) \\
   &=
   n
   \frac{1-\theta ka}{k(k-1)}.
\end{align*}
Since
\[
   ka=k^{1/\alpha},
\]
for all sufficiently large $k$,
\[
   |1-\theta ka|
   \le
   2\theta k^{1/\alpha},
\]
and hence
\[
   |\mu_\theta-\lambda|
   \le
   \frac{4\theta\tau}{k}.
\]
Also,
\[
   \lambda
   =
   \tau k^{-1/\alpha},
\]
and $\lambda > 0$, so
\[
   \frac{
      |\mu_\theta-\lambda|
   }{
      \lambda
   }
   \le
   4\theta k^{-1+1/\alpha}
   =
   4\theta a
   =
   o(1)
\]
uniformly over the allowed range. Thus, for all sufficiently large $k$,
\[
   \min\{\mu_\theta,\lambda\}
   \ge
   \frac{\lambda}{2}.
\]

For Poisson laws (see Fact~\ref{fact:poisson-hellinger}),
\[
   H^2(\Poi(u),\Poi(v))
   \le
   (\sqrt u-\sqrt v)^2
   \le
   \frac{(u-v)^2}{\min\{u,v\}}.
\]
Therefore, for one light coordinate,
\[
\begin{aligned}
   H^2(\Poi(\mu_\theta),\Poi(\lambda))
   &\le
   \frac{
      2(\mu_\theta-\lambda)^2
   }{
      \lambda
   } \\
   &\le
   C_B\tau k^{-2+1/\alpha},
\end{aligned}
\]
where $C_B$ depends only on $B$.

We now pass from one light coordinate to the product of the $k-1$
independent light coordinates. For fixed $j$, write
\[
   P_{\theta,j}
   :=
   M_\theta(\,\cdot\,|J=j),
   \qquad
   Q_{\theta,j}
   :=
   I_\theta(\,\cdot\,|J=j).
\]
The two laws have the same heavy-coordinate marginal
$\Poi(\theta\tau)$ and differ only on the light coordinates. Thus,
\[
   P_{\theta,j}
   =
   \Poi(\theta\tau)
   \otimes
   \bigotimes_{i\ne j}\Poi(\mu_\theta),
\]
and
\[
   Q_{\theta,j}
   =
   \Poi(\theta\tau)
   \otimes
   \bigotimes_{i\ne j}\Poi(\lambda).
\]

Let
\[
   \delta_\theta
   :=
   H^2(
      \Poi(\mu_\theta),
      \Poi(\lambda)
   ).
\]
Recall the Hellinger affinity
\[
   A(P,Q)
   :=
   \sum_x\sqrt{P(x)Q(x)}.
\]
Under our convention,
\[
   H^2(P,Q)
   =
   2-2A(P,Q)
   =
   2\bigl(1-A(P,Q)\bigr),
\]
and hence
\[
   A(P,Q)
   =
   1-\frac12H^2(P,Q).
\]

For product measures, Hellinger affinity tensorizes (see Fact \ref{fact:hellinger}):
\[
   A\left(
      \bigotimes_{\ell=1}^sP_\ell,
      \bigotimes_{\ell=1}^sQ_\ell
   \right)
   =
   \prod_{\ell=1}^sA(P_\ell,Q_\ell).
\]
Indeed,
\begin{align*}
   A\left(
      \bigotimes_{\ell=1}^sP_\ell,
      \bigotimes_{\ell=1}^sQ_\ell
   \right)
   &=
   \sum_{x_1,\ldots,x_s}
   \sqrt{
      \prod_{\ell=1}^sP_\ell(x_\ell)
      \prod_{\ell=1}^sQ_\ell(x_\ell)
   } \\
   &=
   \sum_{x_1,\ldots,x_s}
   \prod_{\ell=1}^s
   \sqrt{
      P_\ell(x_\ell)Q_\ell(x_\ell)
   } \\
   &=
   \prod_{\ell=1}^s
   \sum_{x_\ell}
   \sqrt{
      P_\ell(x_\ell)Q_\ell(x_\ell)
   } \\
   &=
   \prod_{\ell=1}^sA(P_\ell,Q_\ell).
\end{align*}

Applying this identity with $s=k-1$, and observing that the
heavy-coordinate affinity is one because the heavy-coordinate
marginals are identical, gives
\[
   H^2(P_{\theta,j},Q_{\theta,j})
   =
   2
   \left[
      1-
      \left(
         1-\frac{\delta_\theta}{2}
      \right)^{k-1}
   \right].
\]
Since $0\le\delta_\theta\le2$, we may use
\[
   1-(1-x)^s
   \le
   sx,
   \qquad
   0\le x\le1,
\]
with $x=\delta_\theta/2$ and $s=k-1$. Therefore (see Fact~\ref{fact:hellinger}),
\[
\begin{aligned}
   H^2(P_{\theta,j},Q_{\theta,j})
   &\le
   2(k-1)\frac{\delta_\theta}{2} \\
   &=
   (k-1)\delta_\theta \\
   &=
   (k-1)
   H^2(
      \Poi(\mu_\theta),
      \Poi(\lambda)
   ).
\end{aligned}
\]
Using the one-coordinate bound,
\[
\begin{aligned}
   H^2(P_{\theta,j},Q_{\theta,j})
   &\le
   (k-1)
   C_B\tau k^{-2+1/\alpha} \\
   &\le
   C_B\tau k^{-1+1/\alpha}.
\end{aligned}
\]

Since
\[
   \tau\le C_*\alpha,
   \qquad
   \alpha\le c_0\log k,
\]
and
\[
   k^{-1+1/\alpha}
   =
   k^{-(\alpha-1)/\alpha}
   \le
   k^{-1/1001},
\]
we obtain
\[
   C_B\tau k^{-1+1/\alpha}
   \le
   C_BC_*c_0
   (\log k)k^{-1/1001}
   \to0.
\]
Thus, the conditional laws are $o(1)$ apart in Hellinger distance and
hence in total variation, uniformly in $j$. Averaging over the common
uniform hidden coordinate gives (see Fact~\ref{fact:tv-mixture})
\[
   \TV(M_\theta,I_\theta)
   =
   o(1).
\]
\end{proof}

\begin{proposition}[Poissonized mixture indistinguishability]
\label{prop:poissonized-indistinguishable}
For all sufficiently large $k$, uniformly over
\[
   1.001\le\alpha\le c_0\log k,
\]
if
\[
   n
   \le
   C_*\alpha k^{1-1/\alpha},
\]
then
\[
   \TV(M_B,M_1)
   \le
   \frac14.
\]
\end{proposition}

\begin{proof}
The condition on $n$ is equivalent to
\[
   \tau\le C_*\alpha.
\]
By the triangle inequality (see Fact~\ref{fact:tv-mixture}),
\[
\begin{aligned}
   \TV(M_B,M_1)
   &\le
   \TV(M_B,I_B)
   +
   \TV(I_B,I_1)
   +
   \TV(I_1,M_1).
\end{aligned}
\]
For all sufficiently large $k$,
Lemma~\ref{lem:true-ideal-close} bounds the first and third terms by
$1/100$, while Lemma~\ref{lem:ideal-close} bounds the middle term by
$1/50$. Hence,
\[
   \TV(M_B,M_1)
   \le
   \frac1{100}
   +
   \frac1{50}
   +
   \frac1{100}
   =
   \frac1{25}
   <
   \frac14.
\]
\end{proof}

\subsection{Application of the Le Cam reduction and de-Poissonization}

By Le Cam's reduction in Lemma \ref{lem:Le_Cam}, applied with \(F=H_\alpha\), the priors
\(\Pi_0,\Pi_1\), and the Poissonized count-vector experiment, any estimator
with success probability \(q\) over the two prior supports must satisfy
\[
\operatorname{TV}(M_1,M_B)\ge 2q-1.
\]

By Lemma~\ref{lem:entropy-separation}, the R\'enyi entropy values under
$\Pi_0$ and $\Pi_1$ differ by more than $2\epsilon$. By
Proposition~\ref{prop:poissonized-indistinguishable}, the corresponding
Poissonized count-vector laws have total variation distance at most
$1/4$, i.e., \(\operatorname{TV}(M_1,M_B)\le 1/4\), whenever
\[
   n
   \le
   C_*\alpha k^{1-1/\alpha}.
\]
Consequently, no Poissonized estimator with mean sample size
\[
   n
   \le
   C_*\alpha k^{1-1/\alpha}
\]
can estimate $H_\alpha$ to error $\epsilon$ uniformly over the two
prior supports with success probability \(q>5/8\).


It remains to de-Poissonize, that is, to transfer this lower bound to
the fixed-sample model. Suppose, toward a contradiction, that there
exists a fixed-sample estimator using
\[
   s
   \le
   \frac{C_*}{20}
   \alpha k^{1-1/\alpha}
\]
samples and succeeding with probability at least $2/3$ uniformly over
$\Delta_k$. The case $s=0$ is impossible for sufficiently large $k$
already by Lemma~\ref{lem:entropy-separation} and the same testing
argument, so assume $s\ge1$.

Run a Poissonized experiment with mean
\[
   n=20s.
\]
If the realized sample size $N$ satisfies $N\ge s$, apply the
fixed-sample estimator to the first $s$ samples and ignore the rest. If
$N<s$, output an arbitrary value. Since $N$ is independent of the
sample values, conditional on $N\ge s$, the first $s$ samples are
i.i.d.\ from the unknown distribution. Hence, the constructed
Poissonized estimator succeeds with probability at least
\[
   \frac23
   \Prob(\Poi(20s)\ge s).
\]
The multiplicative Chernoff lower-tail bound gives
\[
   \Prob(\Poi(20s)<s)
   \le
   \exp\left(
      -\frac{20s(0.95)^2}{2}
   \right)
   <
   \frac1{32},
   \qquad
   s\ge1.
\]
Therefore, the success probability is at least
\[
   \frac23\cdot\frac{31}{32}
   =
   \frac{31}{48}
   >
   \frac58.
\]
But
\[
   n
   =
   20s
   \le
   C_*\alpha k^{1-1/\alpha},
\]
contradicting the preceding Poissonized lower bound.

Thus, no fixed-sample estimator can use at most
\[
   \frac{C_*}{20}
   \alpha k^{1-1/\alpha}
\]
samples for all sufficiently large $k$. Taking, for example,
\[
   c_1(c_0)
   :=
   \frac{C_*(c_0)}{40}
\]
and increasing $K_0(c_0)$ if necessary proves
Theorem~\ref{thm:main_2}.

\subsection{Context}

For fixed integer $\alpha>1$, the classical lower bounds of Acharya,
Orlitsky, Suresh, and Tyagi identify the dependence
$k^{1-1/\alpha}$
\cite{AcharyaOrlitskySureshTyagi2015}. Lower bounds for real
$\alpha>1$ were subsequently sharpened by Obremski and Skorski
\cite{ObremskiSkorski2017}. The argument above uses a hidden heavy
coordinate and a truncated likelihood calculation to obtain the
uniform lower bound
\[
   \Omega_{c_0}\!\left(
      \alpha k^{1-1/\alpha}
   \right)
\]
over
\[
   1.001\le\alpha\le c_0\log k.
\]
For integer orders
\[
   2\le\alpha\le c_0\log k,
\]
this matches the upper bound proved in
Section~\ref{sec:Renyi_Entropy_UB}. For noninteger orders, the present
argument provides only the displayed lower bound and is not intended to
replace the stronger fixed-order bounds known in the literature.

\section{Conclusion and Open Problems}
\label{sec:conclusion}

This paper separates the sample complexity of min-entropy estimation
from that of Shannon entropy estimation. Shannon entropy has
fixed-accuracy sample complexity $\Theta(k/\log k)$, whereas
min-entropy has fixed-accuracy sample complexity $\Theta(k\log k)$.
The min-entropy upper bound follows from a refined analysis of the
empirical maximum, while the matching lower bound uses a
hidden-heavy-coordinate testing construction. The result also explains
why the Valiant--Valiant relative-earthmover estimator cannot be applied
as a black box to min-entropy: min-entropy is not continuous with
respect to the metric required by that framework.

For R\'{e}nyi entropy, we identify the correct growing-order scale for
integer orders $2\le\alpha\le c_0\log k$. In this regime,
fixed-accuracy estimation has sample complexity
$
   \Theta_{c_0}\!\left(
      \alpha k^{1-1/\alpha}
   \right).
$
The upper bound is achieved by a falling-factorial estimator of the
power sum, together with a variance analysis that retains the dependence
on the growing order. The matching lower bound uses a hidden-heavy-
coordinate prior, Poissonization, and a truncated second-moment
argument. Thus, the factor $\alpha$, which is hidden in fixed-order
notation, is statistically unavoidable when the integer order grows
with $k$.

More generally, for every real order
$1.001\le\alpha\le c_0\log k$, we prove the uniform lower bound
$
   \Omega_{c_0}\!\left(
      \alpha k^{1-1/\alpha}
   \right).
$
For noninteger orders, we do not claim that this lower bound is sharp.
The matching upper bound in the growing-order regime remains restricted
to integer $\alpha$, where unbiased falling-factorial estimators of the
power sum are available.

For sufficiently large orders, min-entropy becomes a uniformly accurate
deterministic approximation to $H_\alpha$. In particular, the bound
$
   0
   \le
   H_\alpha(p)-H_\infty(p)
   \le
   \frac{\log k}{\alpha-1}
$
transfers the min-entropy estimator to high-order R\'{e}nyi entropy.
For fixed accuracy and confidence, this yields
$\Theta_\varepsilon(k\log k)$ sample complexity once
$\alpha$ is a sufficiently large multiple of $\log k$.

Several questions remain open. First, it would be valuable to obtain a
constructive upper bound with explicit growing-$\alpha$ dependence for
noninteger $\alpha>1$. Such a result would clarify the gap between the
uniform lower bound proved here and the substantially different
fixed-order behavior known for noninteger orders.

Second, the dependence of the lower bounds on the additive accuracy
parameter $\varepsilon$ is not optimized. Determining the minimax sample
complexity jointly as a function of $k$, $\alpha$, $\varepsilon$, and
$\delta$ would provide a more complete characterization beyond the
fixed-accuracy regime.

Third, it would be useful to design a single adaptive estimator that
moves automatically between the finite-order collision-estimation
regime and the min-entropy regime. Such an estimator would attain the
appropriate rate without requiring a separate choice of method
depending on whether $\alpha$ lies below, near, or above the logarithmic
threshold.
\section*{Acknowledgments} We acknowledge the use of ChatGPT 5.5 to assist in brainstorming possible proof strategies for the analysis of bounds for R\'{e}nyi entropy estimation. All mathematical statements, proofs, technical arguments, and conclusions were developed and verified by the authors, who take full responsibility for the content of the paper.

\begin{appendix}

\section{The $\alpha$-Dependent Multipliers in the Integer-Order R\'{e}nyi Entropy Bounds in \cite{AcharyaOrlitskySureshTyagi2017}}\label{app:Acharya_LB_UB}

\newcommand{\fall}[2]{#1^{\underline{#2}}}
\newcommand{\bbE}{\mathbb{E}}

Acharya, Orlitsky, Suresh, and Tyagi \cite{AcharyaOrlitskySureshTyagi2017} prove that, for every fixed integer
$\alpha>1$, the sample complexity of estimating $H_\alpha(p)$ over a
$k$-symbol distribution is $\Theta(k^{1-1/\alpha})$ for fixed accuracy and
confidence. In this section we analyze the constants hidden in the proofs of the
integer-order lower and upper bounds. We denote the additive entropy accuracy
as $\Delta$ and the failure probability as $\varepsilon$; the paper often uses
$\delta$ both for accuracy and for a perturbation parameter, so this notation
avoids ambiguity. The main conclusion is that the proof-derived upper
multiplier grows like $\alpha^2$ for fixed additive entropy accuracy when
$\alpha\to\infty$, while the best extraction from the lower-bound construction
does not grow with $\alpha$ for fixed small $\Delta$.  Consequently the
paper's fixed-$\alpha$ theorem should not be read as a uniform theorem for
arbitrary growing integer orders $\alpha=\alpha(k)$.

Throughout Section \ref{app:Acharya_LB_UB} we refer to theorems, lemmas and formulas within the paper of Acharya, Orlitsky, Suresh, and Tyagi~\cite{AcharyaOrlitskySureshTyagi2017}.

\subsection{Setup and convention}

For a distribution $p=(p_x)$, recall the power sum and R\'{e}nyi entropy
\begin{equation}
  P_\alpha(p) := \sum_x p_x^\alpha,
  \qquad
  H_\alpha(p) := \frac{1}{1-\alpha}\log P_\alpha(p),
  \qquad \alpha\neq 1.
\end{equation}
The paper works with fixed accuracy and confidence and defines asymptotic
notation so that hidden constants may depend on $\alpha$, the accuracy, and the
confidence. Thus, the statement
\begin{equation}
  S_\alpha(k,\Delta,\eps)=\Theta(k^{1-1/\alpha}),
  \qquad 1<\alpha\in\mathbb{N},
\end{equation}
is a fixed-$\alpha$ statement, not automatically a uniform statement over a
sequence $\alpha=\alpha(k)$.

Throughout this section set
\begin{equation}
  s_\alpha := 1-\frac{1}{\alpha}.
\end{equation}
We use natural logarithms.  If one uses base-2 logarithms, replace
$e^{-x}$ below by $2^{-x}$, or equivalently insert factors of $\log 2$ in the
accuracy parameter. This does not affect the dependence on $\alpha$.

We track multipliers in bounds of the schematic form
\begin{equation}
  S_\alpha(k,\Delta,\eps)
  \ge \Omega\!\left(C_1(\alpha;\Delta,\eps) \cdot k^{s_\alpha}\right),
  \qquad
  S_\alpha(k,\Delta,\eps)
  \le O\!\left(C_2(\alpha;\Delta,\eps) \cdot k^{s_\alpha}\right).
\end{equation}
When $\Delta$ and $\eps$ are fixed, one often suppresses them and writes $C_i(\alpha)$.

\subsection{Upper-bound multiplier $C_2(\alpha)$}

\subsubsection{The variance bound in Theorem 11 in \cite{AcharyaOrlitskySureshTyagi2017}}

For integer $\alpha>1$, the paper estimates $P_\alpha(p)$ by the
bias-corrected estimator
\begin{equation}
  \widehat P_\alpha^{\,u}
  := \sum_x \frac{\fall{N_x}{\alpha}}{n^\alpha},
\end{equation}
where $N_x$ is the multiplicity of symbol $x$ under Poisson sampling and
$\fall{N_x}{\alpha}=N_x(N_x-1)\cdots(N_x-\alpha+1)$.  This estimator is unbiased:
\begin{equation}
  \bbE[\widehat P_\alpha^{\,u}] = P_\alpha(p).
\end{equation}
The proof of Theorem 11 bounds its relative variance as
\begin{equation}
\label{eq:upper-var}
  \frac{\Var(\widehat P_\alpha^{\,u})}{P_\alpha(p)^2}
  \le
  \sum_{r=0}^{\alpha-1}
  \left(
    \frac{\alpha^2 k^{(\alpha-1)/\alpha}}{n}
  \right)^{\alpha-r}.
\end{equation}
The paper then notes that the right side is at most $\rho^2/12$ if
\begin{equation}
\label{eq:paper-condition}
  \frac{\alpha^2 k^{1-1/\alpha}}{n}
  \le
  \frac{11\rho^2}{144},
\end{equation}
where here $\rho$ is a relative accuracy for estimating $P_\alpha(p)$.
Together with the median trick (Fact~\ref{fact:median-amplification}), this gives the power-sum guarantee
\begin{equation}
\label{eq:power-upper}
  S_\alpha(k,\Delta,\eps)
  \le
  K\,\frac{\alpha^2}{\rho^2}
  k^{1-1/\alpha}\log\frac{2}{\eps}
\end{equation}
for an absolute numerical constant $K$.  If one reads constants directly from
\eqref{eq:paper-condition} and Lemma 8, one may take $K$ to be a large absolute
constant, for example on the order of $36\cdot 144/11$ before fixed-sample and
Poissonization cleanups.  The exact numerical value is less important here than
the factor $\alpha^2/\rho^2$.

\subsubsection{Converting entropy accuracy $\Delta$ to power-sum accuracy $\rho$}

Assume
\begin{equation}
  \widehat P_\alpha = P_\alpha(p)(1+u),
  \qquad |u|\le \rho<1.
\end{equation}
Then
\begin{equation}
  |\widehat H_\alpha-H_\alpha(p)|
  = \frac{1}{\alpha-1}\, |\log(1+u)|.
\end{equation}
To guarantee additive entropy error at most $\Delta$, it is enough to choose
\begin{equation}
\label{eq:rho-delta}
  \rho = \rho_{\alpha,\Delta}
  := 1-e^{-(\alpha-1)\Delta}.
\end{equation}
Indeed, if $|u|\le \rho_{\alpha,\Delta}$, then
\begin{equation}
  1+u\ge 1-\rho_{\alpha,\Delta}=e^{-(\alpha-1)\Delta},
\end{equation}
and also
\begin{equation}
  1+u\le 1+\rho_{\alpha,\Delta}=2-e^{-(\alpha-1)\Delta}
  \le e^{(\alpha-1)\Delta},
\end{equation}
where the last inequality follows from $e^x+e^{-x}\ge 2$.  Hence
\begin{equation}
  e^{-(\alpha-1)\Delta}
  \le \frac{\widehat P_\alpha}{P_\alpha(p)}
  \le e^{(\alpha-1)\Delta},
\end{equation}
which implies $|\widehat H_\alpha-H_\alpha(p)|\le \Delta$.

Substituting \eqref{eq:rho-delta} into \eqref{eq:power-upper} yields
\begin{equation}
\label{eq:C2-main}
  S_\alpha(k,\Delta,\eps)
  \le
  O\!\left(
    \frac{\alpha^2}{\left(1-e^{-(\alpha-1)\Delta}\right)^2}
    k^{1-1/\alpha}\log\frac{1}{\eps}
  \right).
\end{equation}
Thus the multiplier extracted from the proof is
\begin{equation}
\boxed{
  C_2(\alpha;\Delta,\eps)
  \lesssim
  \frac{\alpha^2}{\left(1-e^{-(\alpha-1)\Delta}\right)^2}
  \log\frac{1}{\eps}.
}
\end{equation}

\subsubsection{Regimes for $C_2$}

Let $x=(\alpha-1)\Delta$.

If $x\ll 1$, then $1-e^{-x}\asymp x$, so
\begin{equation}
\label{eq:C2-small}
  C_2(\alpha;\Delta,\eps)
  \lesssim
  \left(\frac{\alpha}{\alpha-1}\right)^2
  \frac{1}{\Delta^2}\log\frac{1}{\eps}.
\end{equation}
For integer $\alpha\ge2$, $(\alpha/(\alpha-1))^2\le4$, so this is uniform in
$\alpha$ when the product $(\alpha-1)\Delta$ is small.

If $\Delta>0$ is fixed and $\alpha\to\infty$, then
$1-e^{-(\alpha-1)\Delta}\to1$, and therefore
\begin{equation}
\label{eq:C2-fixed-delta}
  C_2(\alpha;\Delta,\eps)
  \lesssim_{\Delta}
  \alpha^2\log\frac{1}{\eps}.
\end{equation}
This is the reason the upper-bound proof gives, for growing integer
$\alpha=\alpha(k)$ and fixed additive entropy accuracy $\Delta$,
\begin{equation}
\label{eq:growing-upper}
  S_{\alpha(k)}(k,\Delta,\eps)
  \le
  O_{\Delta}\!\left(
    \alpha(k)^2 k^{1-1/\alpha(k)}\log\frac{1}{\eps}
  \right),
\end{equation}
not $O(k^{1-1/\alpha(k)})$ with a $k$-independent constant.

\subsection{Lower-bound multiplier $C_1(\alpha)$}

\subsubsection{The two distributions in Theorem 15 in \cite{AcharyaOrlitskySureshTyagi2017}}

The lower-bound proof uses Le Cam's two-point method. Let
\begin{equation}
  a := k^{-s_\alpha}=k^{-(1-1/\alpha)},
\end{equation}
and introduce a perturbation parameter $t>0$.  Define distributions $p$ and
$q$ on $[k]$ by
\begin{equation}
  p_1=a,
  \qquad
  p_x=\frac{1-a}{k-1}\quad (2\le x\le k),
\end{equation}
\begin{equation}
  q_1=(1+t)a,
  \qquad
  q_x=\frac{1-(1+t)a}{k-1}\quad (2\le x\le k).
\end{equation}
The paper writes this perturbation as $\delta$; here it is $t$ to distinguish
it from the entropy accuracy $\Delta$.

The power sums are
\begin{equation}
  P_\alpha(p)=a^\alpha+\frac{(1-a)^\alpha}{(k-1)^{\alpha-1}},
  \qquad
  P_\alpha(q)=((1+t)a)^\alpha+
  \frac{(1-(1+t)a)^\alpha}{(k-1)^{\alpha-1}}.
\end{equation}
The paper displays the weaker bound
\begin{equation}
  \frac{|P_\alpha(q)-P_\alpha(p)|}{P_\alpha(p)}\ge \frac{t}{4}
\end{equation}
for $k$ sufficiently large.  This already proves the stated fixed-$\alpha$
lower bound.  However, the same proof contains the sharper information
$(1+t)^\alpha\ge1+\alpha t$, and retaining the full nonlinear term gives a
better description of the $\alpha$-dependence.

\subsubsection{A sharper extraction from the same construction}

First note that $a^\alpha=k^{1-\alpha}$.  Since $a\ge1/k$ for $\alpha\ge1$,
\begin{equation}
  \frac{(1-a)^\alpha}{(k-1)^{\alpha-1}}
  \le
  \frac{(1-1/k)^\alpha}{(k-1)^{\alpha-1}}
  = \frac{k-1}{k^\alpha}
  \le k^{1-\alpha}.
\end{equation}
Therefore
\begin{equation}
\label{eq:Pp-upper}
  P_\alpha(p)\le 2k^{1-\alpha}.
\end{equation}
Next write
\begin{align}
  P_\alpha(q)-P_\alpha(p)
  &=\big((1+t)^\alpha-1\big)k^{1-\alpha}  
  \nonumber\\
  &\quad
  -\frac{(1-a)^\alpha-(1-(1+t)a)^\alpha}{(k-1)^{\alpha-1}}.
\end{align}
By the mean value theorem,
\begin{equation}
  (1-a)^\alpha-(1-(1+t)a)^\alpha
  \le \alpha t a(1-a)^{\alpha-1}.
\end{equation}
Using $(1+t)^\alpha-1\ge \alpha t$ and $a\ge1/k$, the ratio of this negative
term to the positive term is at most
\begin{equation}
  a\left(\frac{1-a}{1-1/k}\right)^{\alpha-1}\le a=k^{-s_\alpha}.
\end{equation}
Hence, whenever $k^{s_\alpha}\ge2$,
\begin{equation}
\label{eq:power-diff-lower}
  P_\alpha(q)-P_\alpha(p)
  \ge
  \frac12\big((1+t)^\alpha-1\big)k^{1-\alpha}.
\end{equation}
Combining \eqref{eq:Pp-upper} and \eqref{eq:power-diff-lower},
\begin{equation}
\label{eq:relative-gap-nonlinear}
  \frac{P_\alpha(q)-P_\alpha(p)}{P_\alpha(p)}
  \ge
  \frac14\big((1+t)^\alpha-1\big).
\end{equation}
Thus the entropy separation obeys
\begin{equation}
\label{eq:entropy-gap-lower}
  |H_\alpha(p)-H_\alpha(q)|
  =\frac{1}{\alpha-1}\log\frac{P_\alpha(q)}{P_\alpha(p)}
  \ge
  \frac{1}{\alpha-1}
  \log\left(1+\frac14\big((1+t)^\alpha-1\big)\right).
\end{equation}

\subsubsection{Choosing the perturbation to get entropy gap $2\Delta$}

For Le Cam's method with an additive entropy estimator of accuracy $\Delta$, it
is enough to make the entropy gap at least $2\Delta$.  A convenient choice is
\begin{equation}
\label{eq:t-alpha-delta}
  t_{\alpha,\Delta}
  :=
  \left[1+4\big(e^{2(\alpha-1)\Delta}-1\big)\right]^{1/\alpha}-1.
\end{equation}
Then \eqref{eq:entropy-gap-lower} gives
\begin{equation}
  |H_\alpha(p)-H_\alpha(q)|\ge 2\Delta.
\end{equation}
For small fixed $\Delta$ this $t_{\alpha,\Delta}$ is below one for all large
$\alpha$; for instance, $t_{\alpha,\Delta}\to e^{2\Delta}-1$, so
$\Delta<\tfrac12\log2$ suffices eventually.  This is consistent with the
paper's lower bound, which is stated for sufficiently small accuracy.

The squared Hellinger distance is the one computed in the paper:
\begin{align}
  h^2(p,q)
  &=\left(\sqrt{a}-\sqrt{(1+t)a}\right)^2
  +\left(\sqrt{1-a}-\sqrt{1-(1+t)a}\right)^2 
  \nonumber\\
  &\le \frac{2t^2}{k^{s_\alpha}},
\end{align}
for $0<t\le1$ and $k$ large enough that the first probabilities are at most
$1/2$.  Consequently,
\begin{equation}
  \|p^n-q^n\|_{\TV}
  \le \sqrt{{n}h^2(p,q)}
  \le t\sqrt{\frac{2n}{k^{s_\alpha}}}.
\end{equation}
Choose a constant $\eta_\eps<1-2\eps$.  If
\begin{equation}
  n\le \eta_\eps^2\frac{k^{s_\alpha}}{2t^2},
\end{equation}
then $\|p^n-q^n\|_{\TV}\le\eta_\eps$, and Le Cam's lemma implies that no
estimator can have error at most $\Delta$ with failure probability at most
$\eps$ uniformly over both $p$ and $q$.  Taking $t=t_{\alpha,\Delta}$ gives
\begin{equation}
\label{eq:C1-main}
  S_\alpha(k,\Delta,\eps)
  \ge
  \Omega_\eps\!\left(
    \frac{1}{t_{\alpha,\Delta}^2}
    k^{1-1/\alpha}
  \right),
\end{equation}
where $t_{\alpha,\Delta}$ is defined in \eqref{eq:t-alpha-delta}.  Therefore
one may take the lower-bound multiplier, up to absolute numerical constants, as
\begin{equation}
\boxed{
  C_1(\alpha;\Delta,\eps)
  \gtrsim_\eps
  \frac{1}{\left(\left[1+4\big(e^{2(\alpha-1)\Delta}-1\big)\right]^{1/\alpha}-1\right)^2}.
}
\end{equation}

\subsubsection{Regimes for $C_1$}

Again let $x=(\alpha-1)\Delta$.

If $x\ll1$, then
\begin{equation}
  e^{2x}-1\sim 2x,
  \qquad
  t_{\alpha,\Delta}
  =\left(1+8x+o(x)\right)^{1/\alpha}-1
  \asymp \frac{\alpha-1}{\alpha}\Delta.
\end{equation}
Thus
\begin{equation}
\label{eq:C1-small}
  C_1(\alpha;\Delta,\eps)
  \gtrsim_\eps
  \left(\frac{\alpha}{\alpha-1}\right)^2\frac{1}{\Delta^2}.
\end{equation}
For integer $\alpha\ge2$, this is bounded below by a constant multiple of
$1/\Delta^2$.

If $\Delta>0$ is fixed and $\alpha\to\infty$, then
\begin{equation}
  t_{\alpha,\Delta}
  \longrightarrow e^{2\Delta}-1,
\end{equation}
and hence
\begin{equation}
\label{eq:C1-fixed-delta}
  C_1(\alpha;\Delta,\eps)
  \gtrsim_{\Delta,\eps}
  1.
\end{equation}
More explicitly, the limiting lower-bound multiplier is of order
\begin{equation}
  \frac{1}{(e^{2\Delta}-1)^2},
\end{equation}
which is approximately $1/(4\Delta^2)$ for small fixed $\Delta$.  In particular,
the best extraction from the two-point lower-bound proof does \emph{not} grow
like $\alpha^2$ as $\alpha\to\infty$ with fixed $\Delta$.

\subsection{What happens when $\alpha=\alpha(k)$ grows?}

The paper's theorem is fixed in $\alpha$.  Still, the displayed proof
calculations suggest the following proof-level behavior, subject to the side
conditions in the lower-bound construction and to the fact that the original
paper does not state a uniform triangular-array theorem.

\subsubsection{The $k$-power itself}

The main $k$-factor is
\begin{equation}
  k^{1-1/\alpha(k)}
  = k\exp\left(-\frac{\log k}{\alpha(k)}\right).
\end{equation}
Thus:
\begin{itemize}[leftmargin=2em]
  \item if $\alpha(k)\ll \log k$, then $k^{1-1/\alpha(k)}=o(k)$, although its exponent tends to $1$;
  \item if $\alpha(k)\sim c\log k$, then $k^{1-1/\alpha(k)}\sim e^{-1/c}k$;
  \item if $\alpha(k)\gg \log k$, then $k^{1-1/\alpha(k)}\sim k$.
\end{itemize}

\subsubsection{Fixed entropy accuracy $\Delta$}

For fixed $0<\Delta<\tfrac12\log2$ and fixed $\eps<1/2$, the optimized lower
construction gives
\begin{equation}
  S_{\alpha(k)}(k,\Delta,\eps)
  \ge
  \Omega_{\Delta,\eps}\!\left(k^{1-1/\alpha(k)}\right)
\end{equation}
whenever the construction's side conditions are satisfied.  The lower-bound
multiplier is bounded away from zero as $\alpha(k)\to\infty$; it does not grow
with $\alpha(k)$.

By contrast, the upper-bound proof gives
\begin{equation}
  S_{\alpha(k)}(k,\Delta,\eps)
  \le
  O_{\Delta}\!\left(
    \alpha(k)^2 k^{1-1/\alpha(k)}\log\frac{1}{\eps}
  \right).
\end{equation}
The factor $\alpha(k)^2$ comes directly from the variance bound
\eqref{eq:upper-var}.  For fixed $\Delta$, the entropy-to-power-sum tolerance
$1-e^{-(\alpha-1)\Delta}$ saturates at a constant instead of contributing a
cancelling factor of $(\alpha-1)^2\Delta^2$.

Therefore, for fixed additive entropy accuracy,
\begin{equation}
  C_1(\alpha(k);\Delta,\eps)=\Omega_{\Delta,\eps}(1),
  \qquad
  C_2(\alpha(k);\Delta,\eps)=O_{\Delta}\!\left(\alpha(k)^2\log\frac1\eps\right)
\end{equation}
from these proof extractions.  The claim that both multipliers grow with
$\alpha(k)$ is not supported by the optimized lower-bound proof; only the
upper-bound multiplier grows quadratically in this reading.

\subsubsection{Shrinking entropy accuracy, $(\alpha-1)\Delta\ll1$}

If the additive accuracy shrinks with $\alpha$ so that
$(\alpha-1)\Delta\ll1$, then the two multipliers have matching
$\alpha$-dependence:
\begin{equation}
  C_1(\alpha;\Delta,\eps)
  \gtrsim_\eps
  \left(\frac{\alpha}{\alpha-1}\right)^2\frac{1}{\Delta^2},
  \qquad
  C_2(\alpha;\Delta,\eps)
  \lesssim
  \left(\frac{\alpha}{\alpha-1}\right)^2
  \frac{1}{\Delta^2}\log\frac{1}{\eps}.
\end{equation}
For integer $\alpha\ge2$, the factor $(\alpha/(\alpha-1))^2$ is between $1$ and
$4$.  This is the regime in which the usual fixed-$\alpha$ statement hides only
constant factors in $\alpha$.

\subsection{Summary table}

\begin{center}
\begin{tabular}{@{}lll@{}}
\toprule
Regime & Lower multiplier from proof & Upper multiplier from proof \\
\midrule
$(\alpha-1)\Delta\ll1$
& $C_1\gtrsim_\eps \left(\frac{\alpha}{\alpha-1}\right)^2\Delta^{-2}$
& $C_2\lesssim \left(\frac{\alpha}{\alpha-1}\right)^2\Delta^{-2}\log(1/\eps)$ \\
\addlinespace
Fixed $\Delta$, $\alpha\to\infty$
& $C_1\gtrsim_{\Delta,\eps} 1$
& $C_2\lesssim_\Delta \alpha^2\log(1/\eps)$ \\
\bottomrule
\end{tabular}
\end{center}

The most important distinction is that the paper proves
$S_\alpha(k,\Delta,\eps)=\Theta(k^{1-1/\alpha})$ with $\alpha$ fixed.  If
$\alpha=\alpha(k)$ grows and $\Delta$ is fixed, the upper-bound proof yields an
extra factor of order $\alpha(k)^2$, while the lower-bound construction yields
only a constant multiplier times $k^{1-1/\alpha(k)}$.  Thus the fixed-$\alpha$
matching theorem does not, by itself, provide a uniform matching theorem for
growing integer orders.

\section{Extracting the $\alpha$-Dependence in the Renyi Entropy Lower Bound in \cite{ObremskiSkorski2017}}\label{app:Obremski_Skorski_LB}

\newcommand{\dd}{\Delta}
\newcommand{\R}{\mathcal R}

We will extract the factor in front of the lower bound expression of the lower bound $\Omega(k^{1-1/\alpha})$ for real $\alpha > 1$ by Obremski and Skorski \cite{ObremskiSkorski2017}, and show that it does not asymptotically depend on $\alpha$.

Fix the additive accuracy $\eps>0$ and the success probability $p>1/2$.
Only the alphabet size $k$ and the order $\alpha>1$ are allowed to vary.  Choose once and for all a fixed entropy gap
\[
  \dd>2\eps .
\]
All logarithms are base $2$, as in Obremski--Skorski \cite{ObremskiSkorski2017}.  Constants below may depend on the fixed numbers $\eps,p,\dd$, but not on $k$ or $\alpha$.

\paragraph{The quantitative step in the lower-bound proof.}
Corollary 9 of Obremski--Skorski \cite{ObremskiSkorski2017} says that if two distributions $X,Y$ satisfy
\[
  \TV(X,Y)\le \gamma
  \quad\text{and}\quad
  |H_\alpha(X)-H_\alpha(Y)|>2\eps,
\]
then the sample size is at least a constant multiple of $\gamma^{-1}$.  In the paper's $2/3$-success normalization this constant is $1/3$; for success probability $p$ it is $c_p:=2p-1$.

Use the pair from the proof of Lemma 11:
\[
  X=U_k,
  \qquad
  Y=\left(\frac1k+\gamma,\frac1k-\frac{\gamma}{k-1},\ldots,
           \frac1k-\frac{\gamma}{k-1}\right).
\]
Then $\TV(X,Y)=\gamma$.  Equation (10) in the paper gives
\[
\R
:=
\frac{\sum_x Y(x)^\alpha}{\sum_x X(x)^\alpha}
=
k^{-1}\left((1+k\gamma)^\alpha
+(k-1)\left(1-\gamma\frac{k}{k-1}\right)^\alpha\right).
\]
Since $Y$ is more peaked than $X$,
\[
  |H_\alpha(X)-H_\alpha(Y)|
  =\frac{1}{\alpha-1}\log \R .
\]
The displayed estimate immediately before equation (12) in the paper, obtained from Proposition 10, yields
\[
\R \ge
\begin{cases}
1+\dfrac{\alpha-1}{3}k^{\alpha-1}\gamma^\alpha,
& 1<\alpha<2 \quad (k\gamma>1),\\[8pt]
1+k^{\alpha-1}\gamma^\alpha,
& \alpha\ge 2.
\end{cases}
\]
The condition $k\gamma>1$ in the first line is satisfied for all sufficiently large $k$ with the choice below; the remaining bounded-$k^{1-1/\alpha}$ cases are absorbed into the absolute constant.

Now set
\[
\gamma=A_\dd(\alpha)k^{-1+1/\alpha},
\]
where
\[
A_\dd(\alpha)=
\begin{cases}
\left(\dfrac{3\bigl(2^{(\alpha-1)\dd}-1\bigr)}{\alpha-1}\right)^{1/\alpha},
&1<\alpha<2,\\[14pt]
\left(2^{(\alpha-1)\dd}-1\right)^{1/\alpha},
&\alpha\ge 2.
\end{cases}
\]
With this choice, the preceding lower bound on $\R$ gives
\[
  \R\ge 2^{(\alpha-1)\dd},
  \qquad\text{hence}\qquad
  |H_\alpha(X)-H_\alpha(Y)|\ge \dd>2\eps .
\]
Therefore Corollary 9 gives
\[
  n\ge c_p\,A_\dd(\alpha)^{-1}k^{1-1/\alpha}.
\]
Thus the constant multiplying $k^{1-1/\alpha}$ extracted from their proof is
\[
\boxed{
C_{\dd,p}(\alpha)=c_p\begin{cases}
\left(\dfrac{\alpha-1}{3\bigl(2^{(\alpha-1)\dd}-1\bigr)}\right)^{1/\alpha},
&1<\alpha<2,\\[14pt]
\left(2^{(\alpha-1)\dd}-1\right)^{-1/\alpha},
&\alpha\ge 2.
\end{cases}}
\]
For the paper's $2/3$-success convention, take $c_p=1/3$.

\paragraph{Uniform boundedness in $\alpha$.}
We show that $C_{\dd,p}(\alpha)=\Theta_{\dd,p}(1)$ on the whole range $\alpha>1$.
The prefactor $c_p$ is fixed and positive, so it suffices to inspect the two branches.

First let $1<\alpha<2$ and put $t=\alpha-1\in(0,1)$.  The inner factor is
\[
  f(t)=\frac{t}{3(2^{t\dd}-1)}.
\]
It extends continuously to $t=0$ with
\[
  f(0):=\lim_{t\downarrow0}\frac{t}{3(2^{t\dd}-1)}
  =\frac{1}{3\dd\log 2},
\]
and it is positive and finite at $t=1$.  Hence
\[
  g(t):=f(t)^{1/(1+t)}
\]
extends to a positive continuous function on the compact interval $[0,1]$.  Therefore there are constants $0<m_1(\dd)\le M_1(\dd)<\infty$ such that
\[
  m_1(\dd)
  \le
  \left(\frac{\alpha-1}{3(2^{(\alpha-1)\dd}-1)}\right)^{1/\alpha}
  \le
  M_1(\dd)
  \qquad (1<\alpha<2).
\]

Next let $\alpha\ge2$ and put $t=\alpha-1\ge1$.  The second branch is
\[
  h(t)=\left(2^{t\dd}-1\right)^{-1/(t+1)} .
\]
It is positive and continuous for $t\ge1$, and
\[
  \log_2 h(t)
  =-\frac{\log_2(2^{t\dd}-1)}{t+1}
  \longrightarrow -\dd
  \qquad(t\to\infty).
\]
Thus
\[
  h(t)\longrightarrow 2^{-\dd}>0.
\]
Consequently $h$ is bounded above and below by positive constants on $[1,\infty)$: for large $t$ this follows from the limit, and on any finite interval it follows from continuity and positivity.

Combining the two regimes, there exist fixed constants
\[
  0<c_-(\dd,p)\le c_+(\dd,p)<\infty
\]
such that, for every real $\alpha>1$,
\[
  c_-(\dd,p)
  \le C_{\dd,p}(\alpha)
  \le c_+(\dd,p).
\]
Equivalently,
\[
  C_{\dd,p}(\alpha)=\Theta_{\dd,p}(1)
  \qquad\text{uniformly for all }\alpha>1.
\]
In particular,
\[
  \lim_{\alpha\to\infty}C_{\dd,p}(\alpha)=c_p\,2^{-\dd},
\]
so the lower-bound constant stays bounded away from zero even when $\alpha$ grows without bound.

\paragraph{Conclusion.}
For fixed additive accuracy and fixed success probability, the lower-bound proof gives
\[
  n\ge C_{\dd,p}(\alpha)k^{1-1/\alpha},
  \qquad
  C_{\dd,p}(\alpha)=\Theta(1),
\]
where the $\Theta(1)$ constants are independent of both $k$ and $\alpha$.

\section{Falling-factorial fact}
Recall, that for an integer $r\ge 0$, the falling factorial is
\[
    (z)_r=z(z-1)\cdots(z-r+1),
\]
with the convention $(z)_0=1$. For an integer $n\ge 0$, if $r>n$ then $(n)_r=0$.

For a set $A\subseteq[k]$ and an integer $r\ge 1$, write
\[
    P_r(A)=\sum_{x\in A}p_x^r.
\]

\begin{fact}[Factorial moment identity for a binomial random variable]\label{app:Fact_binomial-factorial}
Let
\[
    Y \sim \Bin(n,q),
\]
where $n\ge 0$ is an integer and $q\in[0,1]$. Then, for every integer $r\ge 0$,
\[
    \E[(Y)_r]=(n)_r q^r .
\]
In particular, if $0\le r\le n$, then
\[
    \E\bigl[Y(Y-1)\cdots(Y-r+1)\bigr]
    =n(n-1)\cdots(n-r+1)q^r,
\]
and if $r>n$, both sides are zero.
\end{fact}

\begin{proof}
Write $Y$ as a sum of independent Bernoulli random variables:
\[
    Y=B_1+\cdots+B_n,
    \qquad
    B_i\sim \operatorname{Bernoulli}(q),
\]
with $B_1,\ldots,B_n$ independent.

For $r=0$, the identity is immediate because $(Y)_0=(n)_0=q^0=1$. Now assume $r\ge 1$. Let
\[
    [n]^r_{\neq}
    =
    \{(i_1,\ldots,i_r)\in[n]^r:
      i_a\neq i_b\text{ whenever }a\neq b\}
\]
be the set of ordered $r$-tuples of distinct indices. We claim that
\[
    (Y)_r
    =
    \sum_{(i_1,\ldots,i_r)\in [n]^r_{\neq}}
    B_{i_1}B_{i_2}\cdots B_{i_r}.
\]
Indeed, after the Bernoulli variables are realized, suppose exactly $Y=y$ of them equal one. The right-hand side counts the number of ordered $r$-tuples chosen from those $y$ successful indices, which is exactly
\[
    y(y-1)\cdots(y-r+1)=(y)_r.
\]
Thus the displayed identity holds pointwise.

Taking expectations and using independence gives
\begin{align*}
    \E[(Y)_r]
    &=
    \sum_{(i_1,\ldots,i_r)\in [n]^r_{\neq}}
    \E[B_{i_1}B_{i_2}\cdots B_{i_r}] \\
    &=
    \sum_{(i_1,\ldots,i_r)\in [n]^r_{\neq}}
    q^r .
\end{align*}
The number of ordered $r$-tuples of distinct indices from $[n]$ is
\[
    |[n]^r_{\neq}|=(n)_r.
\]
Therefore
\[
    \E[(Y)_r]=(n)_r q^r.
\]
If $r>n$, then $[n]^r_{\neq}$ is empty and $(n)_r=0$, so the same formula remains valid.
\end{proof}

\subsection{Application to the falling-factorial estimator}\label{app:appl_falling}
Let $X_1,\ldots,X_m$ be i.i.d. samples from a distribution $p$ on $[k]$, and for a fixed symbol $x\in[k]$ define the count
\[
    N_x=\sum_{i=1}^m \one\{X_i=x\}.
\]
Then
\[
    N_x\sim \Bin(m,p_x).
\]
Applying Fact~\ref{app:Fact_binomial-factorial} with $n=m$, $q=p_x$, and $r=\alpha$ gives
\[
    \E[(N_x)_\alpha]=(m)_\alpha p_x^\alpha .
\]
Consequently, for the batch estimator
\[
    \widehat P_\alpha
    =
    \frac{1}{(m)_\alpha}
    \sum_{x=1}^k (N_x)_\alpha,
\]
one obtains
\[
    \E[\widehat P_\alpha]
    =
    \frac{1}{(m)_\alpha}
    \sum_{x=1}^k \E[(N_x)_\alpha]
    =
    \sum_{x=1}^k p_x^\alpha
    =
    P_\alpha(p).
\]
This is the unbiasedness identity for the integer-order power-sum estimator.

\begin{fact}[Tuple-sum form of the restricted estimator]\label{app:fact:tuple-sum-estimator}
Let $\alpha\ge 1$ and $m\ge \alpha$ be integers. Let $X_1,\ldots,X_m$ be samples taking values in $[k]$. For $A\subseteq[k]$, define
\[
    N_x=\sum_{i=1}^m \one\{X_i=x\},
    \qquad x\in[k],
\]
and define the restricted falling-factorial estimator
\[
    \widehat P_\alpha(A)
    =
    \frac{1}{(m)_\alpha}
    \sum_{x\in A}(N_x)_\alpha .
\]
Let
\[
    \Tcal
    =
    \{(i_1,\ldots,i_\alpha)\in[m]^\alpha:
      i_a\neq i_b\text{ whenever }a\neq b\}
\]
be the set of ordered $\alpha$-tuples of distinct sample positions. For $t=(i_1,\ldots,i_\alpha)\in\Tcal$, define
\[
    I_t^A
    =
    \one\{X_{i_1}=X_{i_2}=\cdots=X_{i_\alpha}\in A\}.
\]
Then the estimator has the exact pointwise representation
\[
    \widehat P_\alpha(A)
    =
    \frac{1}{(m)_\alpha}
    \sum_{t\in\Tcal} I_t^A .
\]
\end{fact}

\begin{proof}
The proof uses only the definition of the estimator and a counting identity for falling factorials. Fix a realization of the samples $X_1,\ldots,X_m$.

For a fixed symbol $x\in[k]$, the quantity $(N_x)_\alpha$ counts the number of ordered $\alpha$-tuples of distinct sample positions whose entries are all equal to $x$. Equivalently,
\[
    (N_x)_\alpha
    =
    \sum_{t=(i_1,\ldots,i_\alpha)\in\Tcal}
    \prod_{r=1}^\alpha \one\{X_{i_r}=x\}.
\]
Indeed, if $N_x=n_x$, then there are exactly
\[
    n_x(n_x-1)\cdots(n_x-\alpha+1)=(n_x)_\alpha
\]
ordered ways to choose $\alpha$ distinct positions among the $n_x$ positions at which the sample value is $x$.

Substituting this identity into the definition of $\widehat P_\alpha(A)$ gives
\begin{align*}
    \widehat P_\alpha(A)
    &=
    \frac{1}{(m)_\alpha}
    \sum_{x\in A}(N_x)_\alpha \\
    &=
    \frac{1}{(m)_\alpha}
    \sum_{x\in A}
    \sum_{t=(i_1,\ldots,i_\alpha)\in\Tcal}
    \prod_{r=1}^\alpha \one\{X_{i_r}=x\} \\
    &=
    \frac{1}{(m)_\alpha}
    \sum_{t=(i_1,\ldots,i_\alpha)\in\Tcal}
    \sum_{x\in A}
    \prod_{r=1}^\alpha \one\{X_{i_r}=x\}.
\end{align*}
For a fixed tuple $t=(i_1,\ldots,i_\alpha)$, the inner sum
\[
    \sum_{x\in A}
    \prod_{r=1}^\alpha \one\{X_{i_r}=x\}
\]
equals $1$ exactly when all values $X_{i_1},\ldots,X_{i_\alpha}$ are equal to one common symbol in $A$, and equals $0$ otherwise. In other words,
\[
    \sum_{x\in A}
    \prod_{r=1}^\alpha \one\{X_{i_r}=x\}
    =
    I_t^A.
\]
Therefore
\[
    \widehat P_\alpha(A)
    =
    \frac{1}{(m)_\alpha}
    \sum_{t\in\Tcal} I_t^A,
\]
as claimed.
\end{proof}

\begin{fact}[Joint moment of overlapping tuple indicators]\label{app:fact:overlap-indicator-moment}
Let $\alpha\ge 1$ and $m\ge \alpha$ be integers, and let $X_1,\ldots,X_m$ be i.i.d. samples from a distribution $p$ on $[k]$. Let $A\subseteq[k]$. Let $\Tcal$ and $I_t^A$ be defined as in Fact~\ref{app:fact:tuple-sum-estimator}. For $t=(i_1,\ldots,i_\alpha)\in\Tcal$, write
\[
    S(t)=\{i_1,\ldots,i_\alpha\}
\]
for the set of sample positions appearing in the tuple. Suppose $s,t\in\Tcal$ satisfy
\[
    |S(s)\cap S(t)|=j
\]
for some $1\le j\le \alpha$. Then
\[
    \E\!\left[I_s^A I_t^A\right]
    =
    \sum_{x\in A}p_x^{2\alpha-j}
    =
    P_{2\alpha-j}(A).
\]
\end{fact}

\begin{proof}
Let
\[
    U=S(s)\cup S(t).
\]
Since both $s$ and $t$ contain $\alpha$ distinct sample positions and their position sets overlap in exactly $j$ positions,
\[
    |U|=|S(s)|+|S(t)|-|S(s)\cap S(t)|=2\alpha-j.
\]

We first identify the event $\{I_s^A I_t^A=1\}$. If $I_s^A I_t^A=1$, then there exist symbols $a,b\in A$ such that all sample values in the positions of $s$ are equal to $a$, while all sample values in the positions of $t$ are equal to $b$. Because $j\ge 1$, the two tuples share at least one sample position. At a shared position the sample value must equal both $a$ and $b$, so $a=b$. Hence all sample values in the union $U$ are equal to one common symbol in $A$.

Conversely, if all sample values in the positions of $U$ are equal to one common symbol $x\in A$, then both $I_s^A$ and $I_t^A$ are equal to one. Therefore
\[
    \{I_s^A I_t^A=1\}
    =
    \bigcup_{x\in A}
    \{X_i=x\text{ for every }i\in U\}.
\]
The events in this union are disjoint for different $x$. Since the samples are independent,
\[
    \Prob\{X_i=x\text{ for every }i\in U\}
    =
    p_x^{|U|}
    =
    p_x^{2\alpha-j}.
\]
Thus
\[
    \E\!\left[I_s^A I_t^A\right]
    =
    \Prob\{I_s^A I_t^A=1\}
    =
    \sum_{x\in A}p_x^{2\alpha-j}
    =
    P_{2\alpha-j}(A),
\]
as claimed.
\end{proof}

\begin{remark}
The condition $j\ge 1$ in Fact~\ref{app:fact:overlap-indicator-moment} is essential. If $j=0$, then the two tuples use disjoint sample positions. In that case the two indicators are independent and
\[
    \E\!\left[I_s^A I_t^A\right]
    =
    P_\alpha(A)^2,
\]
not $P_{2\alpha}(A)$ in general.
\end{remark}

\section{Standard Inequalities and Probabilistic Facts}
\label{app:standard-facts}

This appendix records several standard inequalities used throughout the
paper. We collect them here to avoid repeating routine arguments.
Our main references are \cite{mitzenmacher2017probability} for basic
probability inequalities and concentration,
\cite{Tsybakov2009} for statistical distances and testing,
\cite{hardy1952inequalities} for classical analytic inequalities, and
\cite{barbour1992poisson} for Poisson distributions and
Poissonization. The power-sum formulations in
Fact~\ref{fact:power-sums} are also used explicitly in the analysis of
R\'enyi entropy estimation in
\cite{AcharyaOrlitskySureshTyagi2017}.

\begin{fact}[Basic probability inequalities]
\label{fact:basic-probability}
The following are the union bound, Markov's inequality, Chebyshev's
inequality, and the Cauchy--Schwarz inequality; see
\cite[Chapters~3--4]{mitzenmacher2017probability}.
For events $A_1,A_2,\ldots$,
\[
   \Prob\left(\bigcup_i A_i\right)
   \le
   \sum_i\Prob(A_i).
\]
If $X\ge0$ and $t>0$, then
\[
   \Prob(X\ge t)
   \le
   \frac{\E X}{t}.
\]
If $X$ has finite variance, then
\[
   \Prob\bigl(|X-\E X|\ge t\bigr)
   \le
   \frac{\Var(X)}{t^2}.
\]
Finally, if $X$ is square-integrable, then
\[
   \E|X|
   \le
   \sqrt{\E[X^2]}.
\]
\end{fact}

\begin{fact}[Total variation and $\chi^2$ divergence]
\label{fact:tv-chi-square}
These identities and inequalities are standard properties of
statistical distances; see
\cite[Section~2.4]{Tsybakov2009}.
Let $P\ll Q$, and write $L=dP/dQ$. Then
\[
   \TV(P,Q)
   =
   \frac12\E_Q|L-1|.
\]
Moreover,
\[
   \chi^2(P\|Q)
   =
   \E_Q[(L-1)^2]
   =
   \E_Q[L^2]-1,
\]
and
\[
   \TV(P,Q)
   \le
   \frac12\sqrt{\chi^2(P\|Q)}.
\]
\end{fact}

\begin{fact}[Testing and total variation]
\label{fact:testing-tv}
The following identity is the two-point testing characterization of
total variation distance; see
\cite[Section~2.2 and Lemma~2.2]{Tsybakov2009}.
For any probability measures $P$ and $Q$,
\[
   \inf_{\phi}
   \frac{
      P(\phi=1)+Q(\phi=0)
   }{2}
   =
   \frac{1-\TV(P,Q)}{2},
\]
where the infimum is over all tests $\phi$ taking values in
$\{0,1\}$. In particular, if $\TV(P,Q)\le\tau$, then every test has
average error at least $(1-\tau)/2$.
\end{fact}

\begin{fact}[Triangle and mixture inequalities]
\label{fact:tv-mixture}
The following are standard consequences of the norm representation of
total variation distance; see
\cite[Section~2.4]{Tsybakov2009}.
For any probability measures $P,Q,R$,
\[
   \TV(P,R)
   \le
   \TV(P,Q)+\TV(Q,R).
\]
More generally, if $\pi$ is a probability measure and
$\{P_\theta\}$ and $\{Q_\theta\}$ are two families of probability
measures, then
\[
   \TV\left(
      \int P_\theta\,d\pi(\theta),
      \int Q_\theta\,d\pi(\theta)
   \right)
   \le
   \int
      \TV(P_\theta,Q_\theta)
   \,d\pi(\theta).
\]
\end{fact}

\begin{fact}[Hellinger distance and tensorization]
\label{fact:hellinger}
Standard properties of Hellinger distance and affinity are collected in
\cite[Section~2.4]{Tsybakov2009}.
Under the convention
\[
   H^2(P,Q)
   :=
   \int
   \left(
      \sqrt{dP}-\sqrt{dQ}
   \right)^2,
\]
we have
\[
   \TV(P,Q)\le H(P,Q).
\]
If
\[
   A(P,Q)
   :=
   \int\sqrt{dP\,dQ}
\]
denotes the Hellinger affinity, then
\[
   H^2(P,Q)
   =
   2\bigl(1-A(P,Q)\bigr)
\]
and
\[
   A\left(
      \bigotimes_{\ell=1}^mP_\ell,
      \bigotimes_{\ell=1}^mQ_\ell
   \right)
   =
   \prod_{\ell=1}^m A(P_\ell,Q_\ell).
\]
Consequently,
\[
   H^2\left(
      \bigotimes_{\ell=1}^mP_\ell,
      \bigotimes_{\ell=1}^mQ_\ell
   \right)
   \le
   \sum_{\ell=1}^mH^2(P_\ell,Q_\ell).
\]
\end{fact}

\begin{fact}[Hellinger distance between Poisson laws]
\label{fact:poisson-hellinger}
The Poisson affinity formula follows by summing the Poisson probability
mass functions; see also the treatment of Poisson measures in
\cite{barbour1992poisson}. For all $u,v\ge0$,
\[
   H^2\bigl(\Poi(u),\Poi(v)\bigr)
   =
   2\left(
      1-
      \exp\left\{
         -\frac{(\sqrt u-\sqrt v)^2}{2}
      \right\}
   \right).
\]
In particular,
\[
   H^2\bigl(\Poi(u),\Poi(v)\bigr)
   \le
   (\sqrt u-\sqrt v)^2.
\]
If $u,v>0$, then
\[
   (\sqrt u-\sqrt v)^2
   =
   \frac{(u-v)^2}{(\sqrt u+\sqrt v)^2}
   \le
   \frac{(u-v)^2}{\min\{u,v\}}.
\]
\end{fact}

\begin{fact}[Elementary exponential and logarithmic inequalities]
\label{fact:elementary-log-exp}
These standard inequalities follow from convexity and concavity of the
exponential and logarithm functions; see
\cite{hardy1952inequalities}.
For every $x\in\mathbb R$,
\[
   1+x\le e^x.
\]
For every $x>-1$,
\[
   \log(1+x)\le x.
\]
For every $x\in[0,1)$,
\[
   \log(1-x)\le-x
   \qquad\text{and}\qquad
   \log(1-x)\ge-\frac{x}{1-x}.
\]
\end{fact}

\begin{fact}[Elementary algebraic bounds]
\label{fact:elementary-algebra}
The following elementary convexity, binomial, and geometric-series
bounds are standard; see \cite{hardy1952inequalities}.
For all real $x,y$,
\[
   (x-y)^2
   \le
   2x^2+2y^2.
\]
For every $x\in[0,1]$ and every integer $s\ge1$,
\[
   1-(1-x)^s
   \le
   sx.
\]
If $a_r\ge0$ and
$a_{r+1}\le\rho a_r$ for all $r\ge r_0$, where
$\rho\in[0,1)$, then
\[
   \sum_{r=r_0}^{\infty}a_r
   \le
   \frac{a_{r_0}}{1-\rho}.
\]
\end{fact}

\begin{fact}[Factorial and optimization bounds]
\label{fact:factorial-optimization}
The factorial bound is a standard consequence of Stirling's inequality,
and the second bound follows by maximizing the function
$x\mapsto x\log(A/x)$; see
\cite{mitzenmacher2017probability,hardy1952inequalities}.
For every integer $r\ge1$,
\[
   r!
   \ge
   \left(\frac re\right)^r.
\]
For every $A>0$ and $x>0$,
\[
   x\log\frac Ax
   \le
   \frac Ae.
\]
\end{fact}

\begin{fact}[Power-sum inequalities]
\label{fact:power-sums}
These inequalities follow from Jensen's inequality, H\"older's
inequality, and monotonicity of finite-dimensional $\ell_r$ norms; see
\cite{hardy1952inequalities}. Their power-sum forms are also used in
\cite{AcharyaOrlitskySureshTyagi2017}.
Let $p\in\Delta_k$, and define
\[
   P_r(p):=\sum_{i=1}^kp_i^r.
\]
For every $r\ge1$,
\[
   P_r(p)\ge k^{1-r}.
\]
For every $r>0$ and $s\ge0$,
\[
   P_{r+s}(p)
   \le
   P_r(p)^{(r+s)/r}.
\]
For every $r>0$ and $s\in[0,r]$,
\[
   P_{r-s}(p)
   \le
   k^{s/r}
   P_r(p)^{(r-s)/r}.
\]
In particular,
\[
   P_{2r}(p)
   \le
   P_r(p)^2.
\]
\end{fact}

\begin{fact}[Median amplification]
\label{fact:median-amplification}
This is the standard median trick obtained by applying a Chernoff bound
to the number of unsuccessful repetitions; see
\cite[Chapter~4]{mitzenmacher2017probability}.
Let $\widehat\theta_1,\ldots,\widehat\theta_B$ be independent
estimators such that
\[
   \Prob\bigl(
      |\widehat\theta_b-\theta|\le\varepsilon
   \bigr)
   \ge
   \frac34
   \qquad
   \text{for every }b.
\]
If
\[
   B
   \ge
   8\log\frac1\delta,
\]
then their median satisfies
\[
   \Prob\left(
      \left|
         \operatorname{median}
         (\widehat\theta_1,\ldots,\widehat\theta_B)
         -\theta
      \right|
      >
      \varepsilon
   \right)
   \le
   \delta.
\]
\end{fact}

\begin{fact}[Poissonization]
\label{fact:poissonization}
The following Poisson splitting property is standard; see
\cite{barbour1992poisson} and
\cite{mitzenmacher2017probability}.
Let $N\sim\Poi(n)$, and conditional on $N$, draw
$X_1,\ldots,X_N$ independently from $p\in\Delta_k$. If $N_i$
denotes the number of occurrences of symbol $i$, then
\[
   N_1,\ldots,N_k
\]
are independent and
\[
   N_i\sim\Poi(np_i),
   \qquad i\in[k].
\]
\end{fact}
\begin{fact}[The stated logarithmic-exponential inequality]\label{app:Fact_4}
For every $x\in[-1/2,0]$,
\[
    \log(1+x)\ge 2x.
\]
Equivalently, for every $y\in[0,1/2]$,
\[
    \log(1-y)\ge -2y,
\]
and therefore
\[
    1-y\ge e^{-2y}.
\]
\end{fact}

\begin{proof}
Let
\[
    f(x)=\log(1+x)-2x
\]
for $x\in[-1/2,0]$. Then
\[
    f(0)=0
\]
and
\[
    f'(x)=\frac{1}{1+x}-2
    =\frac{-1-2x}{1+x}.
\]
For $x\in[-1/2,0]$, the denominator $1+x$ is positive and the numerator $-1-2x$ is nonpositive. Hence
\[
    f'(x)\le 0
\]
on $[-1/2,0]$. Therefore $f$ is nonincreasing on $[-1/2,0]$. Since $f(0)=0$, for every $x\in[-1/2,0]$ we have
\[
    f(x)\ge f(0)=0.
\]
Thus
\[
    \log(1+x)-2x\ge 0,
\]
which proves
\[
    \log(1+x)\ge 2x
\]
for every $x\in[-1/2,0]$.

Finally, substituting $x=-y$ gives, for $y\in[0,1/2]$,
\[
    \log(1-y)\ge -2y.
\]
Exponentiating both sides gives
\[
    1-y\ge e^{-2y}.
\]
\end{proof}

\end{appendix}

\bibliographystyle{unsrt}
\bibliography{biblio}

\end{document}